%% file: LassoSampler-R10.tex
\documentclass[11pt]{article}
\usepackage{geometry}                % See geometry.pdf to learn the layout options. There are lots.
\usepackage{graphicx}
\usepackage{amsmath,amssymb, amsthm,epsfig,bm}
\usepackage{setspace}
\usepackage{epstopdf}
\usepackage{psfrag}
\usepackage{color,soul}
\usepackage{verbatim}
\usepackage{algorithm,algorithmic}
\usepackage{multirow}
\usepackage{natbib}
\usepackage{xr}
\usepackage{mathrsfs}
\bibpunct{(}{)}{;}{a}{}{,} 
\sethlcolor{yellow}
\numberwithin{equation}{section}

\theoremstyle{definition}
\newtheorem{theorem}{Theorem}
\newtheorem{lemma}{Lemma}

\newtheorem{corollary}{Corollary}

\theoremstyle{definition}
\newtheorem{definition}{Definition}
\newtheorem{remark}{Remark}

\newtheorem{assumption}{Assumption}
\newtheorem{routine}{Routine}

\input{NewCommands}

\begin{document}

\title{Monte Carlo Simulation for Lasso-Type Problems \\by Estimator Augmentation}
\author{Qing Zhou\thanks{UCLA Department of Statistics (email: zhou@stat.ucla.edu). This work was supported by NSF grants DMS-1055286 and DMS-1308376. The author thanks one referee for providing useful references.}}
\date{}
\maketitle

\begin{abstract}

Regularized linear regression under the $\ell_1$ penalty, such as the Lasso, has been shown
to be effective in variable selection and sparse modeling. The sampling distribution
of an $\ell_1$-penalized estimator $\hbmbeta$ is hard to determine 
as the estimator is defined by an optimization problem that in general can only be solved numerically
and many of its components may be exactly zero. Let $\bfS$ be the subgradient of the $\ell_1$ norm
of the coefficient vector $\bmbeta$ evaluated at $\hbmbeta$. We find that the joint sampling distribution
of $\hbmbeta$ and $\bfS$, together called an augmented estimator, 
is much more tractable and has a closed-form density under a normal error distribution in both low-dimensional ($p\leq n$)
and high-dimensional ($p>n$) settings. 
Given $\bmbeta$ and the error variance $\sigma^2$,
one may employ standard Monte Carlo methods, such as Markov chain Monte Carlo and importance sampling,
to draw samples from the distribution of the augmented estimator
and calculate expectations with respect to the sampling distribution of $\hbmbeta$. 
We develop a few concrete Monte Carlo algorithms and 
demonstrate with numerical examples that our approach may offer huge advantages 
and great flexibility in studying sampling distributions in $\ell_1$-penalized linear regression. 
We also establish nonasymptotic bounds on the difference between the true sampling distribution of $\hbmbeta$ and
its estimator obtained by plugging in estimated parameters, which justifies the validity of Monte Carlo simulation from an estimated sampling distribution even when $p\gg n\to \infty$.


\end{abstract}

\section{Introduction}

Consider the linear regression model,
\begin{equation}\label{eq:linearmodel}
\bfY=\bfX\bmbeta+ \bmeps,
\end{equation}
where $\bfY$ is an $n$-vector, $\bfX$ an $n\times p$ design matrix, $\bmbeta=(\beta_j)_{1:p}$ the vector
of coefficients, and $\bmeps$ i.i.d. random errors with mean zero and variance $\sigma^2$.
Recently, $\ell_1$-penalized estimation methods \citep{Tibshirani96, Chen98} 
have been widely used to find sparse estimates of the coefficient vector.
Given positive weights $w_j$, $j=1,\ldots,p$, and a tuning parameter $\lambda>0$,
an $\ell_1$-penalized estimator $\hbmbeta=(\hbeta_j)_{1:p}$ is defined by minimizing
the following penalized loss function,
\begin{equation}\label{eq:lassoloss}
\ell(\bmbeta)=\frac{1}{2}\|\bfY-\bfX\bmbeta \|_2^2 + n \lambda \sum_{j=1}^p w_j |\beta_j|.
\end{equation}
By different ways of choosing $w_j$, the estimator corresponds to the Lasso \citep{Tibshirani96}, the adaptive Lasso \citep{Zou06}, and the one-step linear local approximation (LLA) estimator \citep{Zou08} among others. 
We call such an estimator a Lasso-type estimator.

In many applications of $\ell_1$-penalized regression, it is desired to quantify the uncertainty in the estimates. 
However, except for very special cases, the sampling distribution of a Lasso-type estimator is complicated and
difficult to approximate. Closed-form approximations to the covariance matrices of 
the estimators in \cite{Tibshirani96}, \cite{FanLi01}, and \cite{Zou06} are unsatisfactory, 
as they all give zero variance for a zero component of the estimators and thus fail to quantify
the uncertainty in variable selection. Theoretical results on finite-sample distributions
and confidence sets of some
Lasso-type estimators have been developed \citep{Potscher09,Potscher10}
but only under orthogonal designs, which clearly limits general applications of these results. 
The bootstrap can be used to approximate the sampling distribution of a Lasso-type estimator,
in which numerical optimization is needed to minimize \eqref{eq:lassoloss} for every bootstrap sample.
Although there are efficient algorithms, such as the Lars \citep{Efron04}, the homotopy algorithm \citep{Osborne00}, 
and coordinate descent \citep{Friedman07,WuLange08},
to solve this optimization problem, it is still time-consuming to apply these algorithms
hundreds or even thousands of times in bootstrap sampling. 
As pointed out by \cite{KnightFu00} and \cite{Chatterjee10}, the bootstrap may not be consistent for estimating 
the sampling distribution of the Lasso under certain circumstances.
To overcome this difficulty, a modified bootstrap \citep{Chatterjee11} and a perturbation resampling approach \citep{Minnier11} have been proposed, both justified under a fixed-$p$ asymptotic framework.
\cite{ZhangZhang11} have developed methods for constructing confidence intervals for individual coefficients and their linear combinations in high-dimensional ($p>n$) regression with sufficient conditions for the asymptotic normality of the proposed estimators. There are several recent articles on significance test and confidence region construction for sparse high-dimensional linear models \citep{Javanmard13b,Javanmard13a,Lockhart13,vandeGeer13}, all based on asymptotic distributions for various functions of the Lasso. On the other hand, knowledge on sampling distributions is also useful for distribution-based model selection with $\ell_1$ penalization, as demonstrated by stability selection \citep{Meinshausen10}
and the Bolasso \citep{Bach08}. 

A possible alternative to the bootstrap or resampling
is to simulate from a sampling distribution %(with estimated parameters) 
by Monte Carlo methods, such as Markov chain Monte Carlo (MCMC). 
An obvious obstacle to using these methods for a Lasso-type estimator
is that its sampling distribution does not have a closed-form density. In this article, we study
the joint distribution of a Lasso-type estimator $\hbmbeta$ and the subgradient $\bfS$ of $\|\bmbeta\|_1$
evaluated at $\hbmbeta$.
Interestingly, this joint distribution has a density that can be calculated explicitly 
assuming a normal error distribution, regardless of the relative size between $n$ and $p$.
Thus, one can develop Monte Carlo algorithms to draw samples from this joint distribution and estimate
various expectations of interest with respect to the sampling distribution of $\hbmbeta$, which is simply a marginal distribution. This approach offers great flexibility in studying the sampling distribution of 
a Lasso-type estimator. For instance,
one may use importance sampling (IS) to accurately estimate a tail probability (small p-value) with respect
to the sampling distribution under a null hypothesis, 
which can be orders of magnitude more efficient than any method directly targeting at the sampling distribution.
Another potential advantage of this approach is that, at each iteration, an MCMC algorithm only
evaluates a closed-form density, which is much faster than minimizing \eqref{eq:lassoloss} 
numerically as used in the bootstrap. Furthermore, our method can be interpreted as an MCMC algorithm targeting at a multivariate normal distribution with locally reparameterized moves and hence is expected to be computationally tractable. 

The remaining part of this article is organized as follows. After a high-level description of the basic idea, Section~\ref{sec:distribution} derives the density of the joint distribution of $\hbmbeta$ and $\bfS$ in the low-dimensional setting with $p\leq n$, and Section~\ref{sec:MC} develops MCMC algorithms for this setting. The density 
in the high-dimensional setting with $p>n$ is derived in Section~\ref{sec:highdim}. In Section~\ref{sec:pvalue}, 
we construct applications of the high-dimensional result in p-value calculation
for Lasso-type inference by IS. Numerical examples are provided in Sections~\ref{sec:MC}
and \ref{sec:pvalue} to demonstrate the efficiency of the Monte Carlo algorithms. Section~\ref{sec:errorbounds} provides theoretical justifications for simulation from an estimated sampling distribution of the Lasso by establishing its  consistency as $p\gg n\to \infty$.
Section~\ref{sec:generalization} includes generalizations to random designs, a connection to model selection consistency,
and a Bayesian interpretation of the sampling distribution. The article concludes with a brief discussion and some remarks.
Technical proofs are relegated to Section~\ref{sec:proofs}.

Notations for vectors and matrices are defined here. 
All vectors are regarded as column vectors. 
Let $A=\{j_1,\ldots,j_k\}\subseteq \{1,\ldots,m\}$ and $B=\{i_1,\ldots,i_{\ell}\} \subseteq \{1,\ldots,n\}$ be two index sets. 
For vectors $\bfv=(v_j)_{1:m}$ and $\bfu=(u_i)_{1:n}$,
we define $\bfv_A=(v_j)_{j\in A}=(v_{j_1},\ldots,v_{j_k})$, $\bfv_{-A}=(v_j)_{j\notin A}$, 
and $(\bfv_A,\bfu_B)=(v_{j_1},\ldots,v_{j_k},u_{i_1},\ldots,u_{i_{\ell}})$.
For a matrix $\bfM=(M_{ij})_{m\times n}$,
write its columns as $\bfM_j$, $j=1,\ldots,n$. Then $\bfM_B=(\bfM_{j})_{j\in B}$ extracts the columns in $B$,
the submatrix $\bfM_{AB}=(M_{ij})_{i\in A, j \in B}$ extracts the rows in $A$ and the columns in $B$,
and $\bfM_{A\bullet}=(M_{ij})_{i \in A}$ extracts the rows in $A$. Furthermore, $\bfM^{\trans}_B$ and
$\bfM_{AB}^{\trans}$ are understood as $(\bfM_B)^{\trans}$ and $(\bfM_{AB})^{\trans}$, respectively.
We denote the row space, the null space, and the rank of $\bfM$ by $\row(\bfM)$, $\nul(\bfM)$, and $\rank(\bfM)$, respectively.
Denote by $\diag(\bfv)$ the $m\times m$ diagonal matrix with $\bfv$ as the diagonal elements,
and by $\diag(\bfM,\bfM')$ the block diagonal matrix with $\bfM$ and $\bfM'$ as the diagonal blocks,
where the submatrices $\bfM$ and $\bfM'$ may be of different sizes and may not be square.
For a square matrix $\bfM$, $\diag(\bfM)$ extracts the diagonal elements.
Denote by $\bfI_n$ the $n\times n$ identity matrix.

\section{Estimator augmentation}\label{sec:distribution}

\subsection{The basic idea}\label{sec:basic}

Let $\bfW=\diag(w_1,\ldots,w_p)$.
A minimizer $\hbmbeta$ of \eqref{eq:lassoloss} is given by the Karush-Kuhn-Tucker (KKT) condition 
\begin{equation}\label{eq:lassograd}
\frac{1}{n}\bfX^{\trans}\bfY = \frac{1}{n}\bfX^{\trans} \bfX \hbmbeta + \lambda \bfW \bfS,
\end{equation}
where $\bfS=(S_j)_{1:p}$ is the subgradient of the function $g(\bmbeta)=\|\bmbeta\|_1$ evaluated at 
the solution $\hbmbeta$. Therefore, 
\begin{equation}\label{eq:defS}
\left\{
\begin{array}{ll}
S_j=\sgn(\hbeta_j) & \mbox{ if } \hbeta_j \ne 0, \\
S_j \in [-1,1] & \mbox{ if } \hbeta_j=0,
\end{array}
\right.
\end{equation}
for $j=1,\ldots,p$. Hereafter, we may simply call $\bfS$ the subgradient if the meaning is clear from context.
Lemma~\ref{lm:prelim} reviews a few basic facts about the uniqueness of $\hbmbeta$ and $\bfS$. 
\begin{lemma}\label{lm:prelim}
For any $\bfY$, $\bfX$ and $\lambda>0$, every minimizer $\hbmbeta$ of \eqref{eq:lassoloss}
gives the same fitted value $\bfX \hbmbeta$ and the same subgradient $\bfS$. Moreover, if the columns of $\bfX$ are in general position, then $\hbmbeta$ is unique for any $\bfY$ and $\lambda>0$.
\end{lemma}
\begin{proof}
See Lemma 1 and Lemma 3 in \cite{Tibshirani13} for proof of the uniqueness of the fitted value $\bfX\hbmbeta$ and
the uniqueness of $\hbmbeta$. Since $\bfS$ is
a (vector-valued) function of $\bfX\hbmbeta$ from the KKT condition \eqref{eq:lassograd}, 
it is also unique for fixed $\bfY$, $\bfX$ and $\lambda$. 
\end{proof}

We regard $\hbmbeta$ and $\bfS$ together as the solution to Equation \eqref{eq:lassograd}.
Lemma~\ref{lm:prelim} establishes that $(\hbmbeta,\bfS)$ is unique for any $\bfY$ 
assuming the columns of $\bfX$ are in general position (for a technical definition see \citealp{Tibshirani13}), regardless of the sizes of $n$ and $p$.
We call the vector $(\hbmbeta,\bfS)$ the augmented estimator in an $\ell_1$-penalized regression
problem. The augmented estimator will play a central role in our study of the sampling distribution of $\hbmbeta$.

Let $\bfU = \frac{1}{n}\bfX^{\trans}\bmeps=\frac{1}{n}\bfX^{\trans}\bfY-\bfC\bmbeta$, where 
$\bfC=\frac{1}{n}\bfX^{\trans} \bfX$ is the Gram matrix. By definition, $\bfU\in\row(\bfX)$. Rewrite the KKT condition as 
\begin{equation}\label{eq:lassoKKTinU}
\bfU = \bfC \hbmbeta + \lambda \bfW \bfS - \bfC \bmbeta \defi \bfH(\hbmbeta,\bfS;\bmbeta),
\end{equation}
which shows that $\bfU$ is a function of $(\hbmbeta,\bfS)$. On the other hand,
$\bfY$ determines $(\hbmbeta,\bfS)$ only through $\bfU$, which implies that 
$(\hbmbeta,\bfS)$ is unique for any $\bfU$ as long as it is unique for any $\bfY$.
Therefore, under the assumptions for the uniqueness of $\hbmbeta$, 
$\bfH$ is a bijection between $(\hbmbeta,\bfS)$ and $\bfU$. 
For a fixed $\bfX$, the only source of randomness in the linear model \eqref{eq:linearmodel} 
is the noise vector $\bmeps$, which determines the distribution of $\bfU$.
With the bijection between $\bfU$ and $(\hbmbeta,\bfS)$, one may derive the joint distribution of $(\hbmbeta,\bfS)$,
which has a closed-form density under a normal error distribution. Then we develop Monte Carlo algorithms
to sample from this joint distribution and obtain the sampling distribution of $\hbmbeta$. 
This is the key idea of this article, which works for both the low-dimensional setting ($p\leq n$) and the
high-dimensional setting ($p>n$). Although the basic strategy is the same, the technical details 
are slightly more complicated for the high-dimensional setting.
For the sake of understanding, we first focus on the low-dimensional case in the remaining part of Section~\ref{sec:distribution} and Section~\ref{sec:MC}, and then generalize the results to the high-dimensional setting
in Section~\ref{sec:highdim}.

Before going through all the technical details, we take a glimpse of the utility of this work in a couple concrete examples. Given a design matrix $\bfX$ and a value of $\lambda$, our method gives a closed-form joint density $\pi$ for the Lasso-type estimator $\hbmbeta$ and the subgradient $\bfS$ under assumed values of the true parameters (Theorems~\ref{thm:joint} and \ref{thm:jointhigh}). Targeting at this density, we have developed MCMC algorithms, such as the Lasso sampler in Section~\ref{sec:MHLasso}, to draw samples from the joint distribution of $(\hbmbeta,\bfS)$. Such MCMC samples allow for approximation of marginal distributions for a Lasso-type estimator. Figure~\ref{fig:diagnostic} demonstrates the results of the Lasso sampler applied on a simulated dataset with $n=500$, $p=100$ and a normal error distribution. The scatter plot in Figure~\ref{fig:diagnostic}(a) confirms that $\hbeta_j$ indeed may have a positive probability to be exactly zero in its sampling distribution. Accordingly, the distribution of the subgradient $S_j$ in Figure~\ref{fig:diagnostic}(b) has a continuous density on $(-1,1)$ and two point masses on $\pm 1$. The fast mixing and low autocorrelation shown in the figure are surprisingly satisfactory for a simple MCMC algorithm in such a high-dimensional and complicated space ($\R^p \times 2^{\{1,\ldots,p\}}$, see \eqref{eq:space}). Exploiting the explicit form of the bijection $\bfH$, we achieve the goal of sampling from the joint distribution of $(\hbmbeta,\bfS)$ via an MCMC algorithm essentially targeting at a multivariate normal distribution (Section~\ref{sec:viewofMCMC}). This approach does not need numerical optimization in any step, which makes it highly efficient compared to bootstrap or resampling-based methods. Another huge potential of our method is its ability to estimate tail probabilities, such as small p-values in a significance test (Section~\ref{sec:pvalue}). Estimating tail probabilities is challenging for any simulation method. With a suitable proposal distribution, having an explicit density makes it possible to accurately estimate tail probabilities by importance weights. For example, our method can estimate a tail probability on the order of $10^{-20}$, with a coefficient of variation around 2, by simulating only 5,000 samples from a proposal distribution. This is absolutely impossible when bootstrapping the Lasso or simulating from the sampling distribution directly.

\begin{figure}[t]
\centering
  \begin{minipage}[b]{0.4\linewidth}
   \centering
\includegraphics[width=\linewidth,trim=0in 0.2in 0in 0.5in,clip]{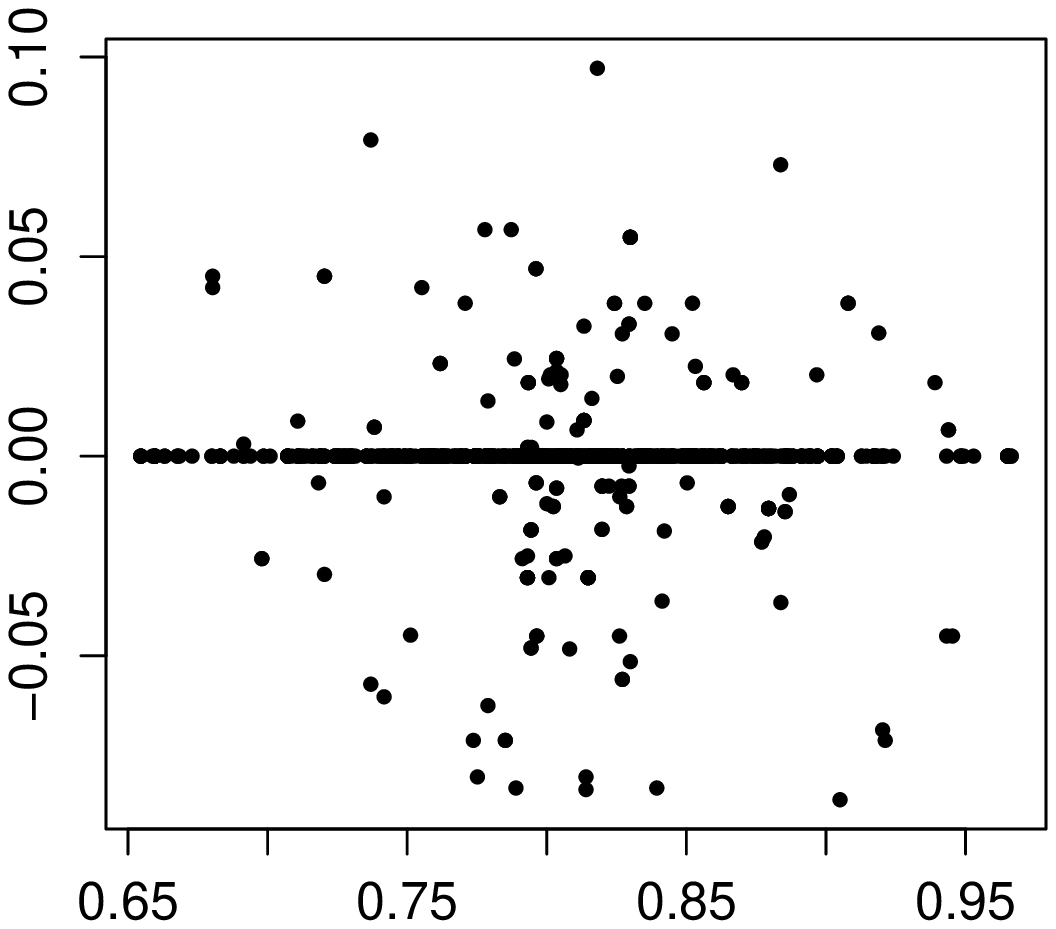} \\
(a)
\end{minipage}
\begin{minipage}[b]{0.4\linewidth}
   \centering
   \includegraphics[width=\linewidth,trim=0in 0.2in 0in 0.5in,clip]{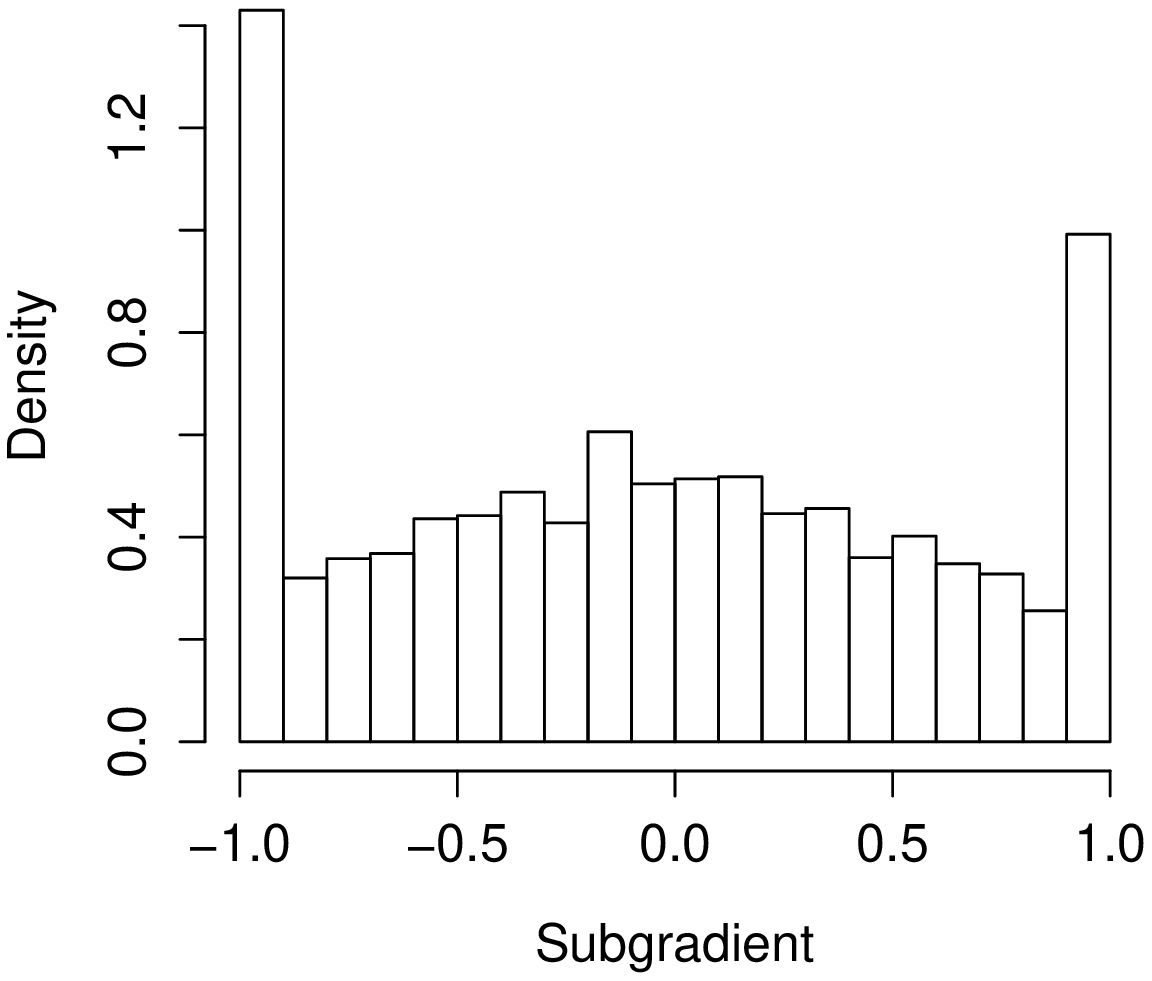} \\
   (b)
\end{minipage}   \\
  \begin{minipage}[b]{0.4\linewidth}
   \centering
   \includegraphics[width=\linewidth,trim=0in 0.2in 0in 0.5in,clip]{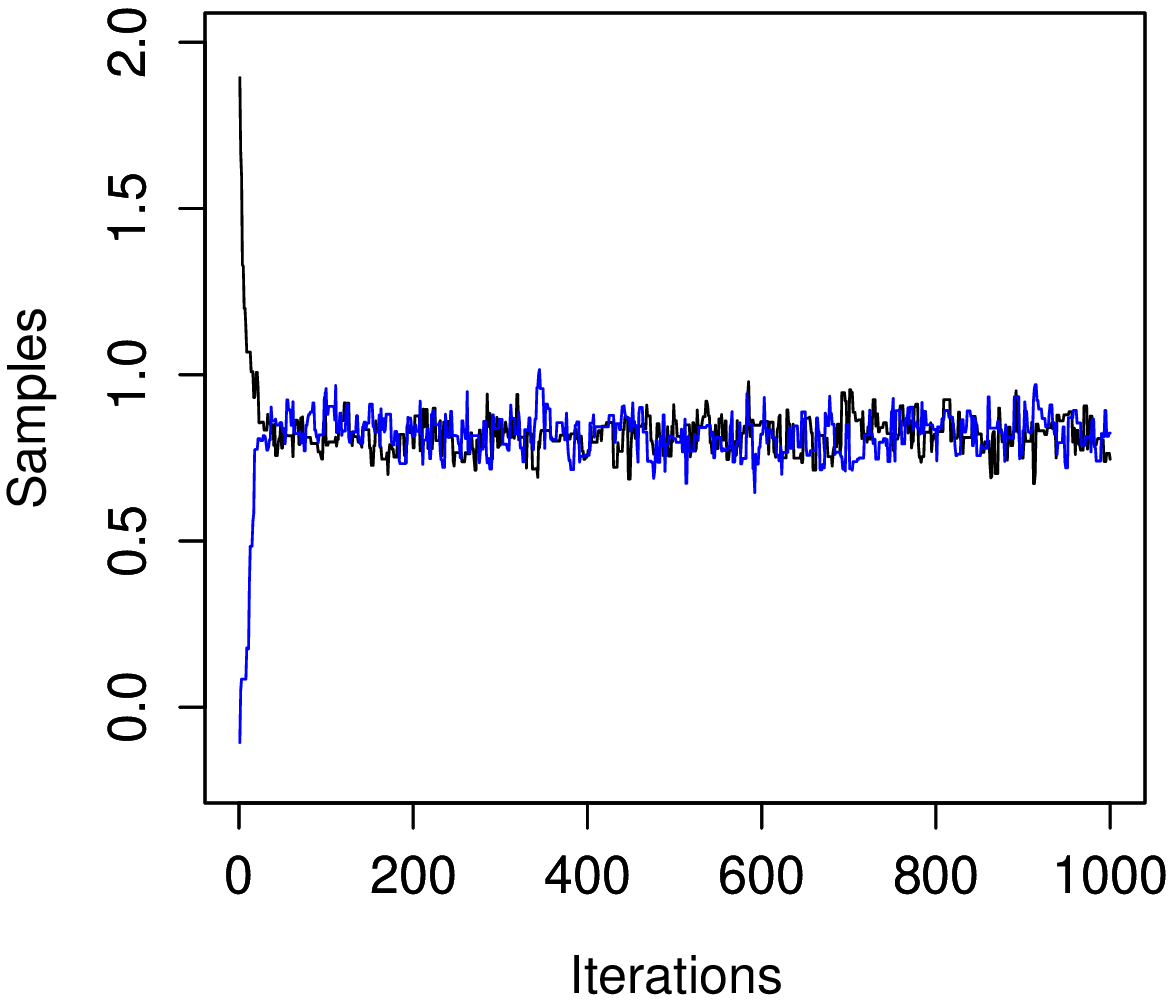} \\
   (c)
\end{minipage}
\begin{minipage}[b]{0.4\linewidth}
   \centering
   \includegraphics[width=\linewidth,trim=0in 0.2in 0in 0.5in,clip]{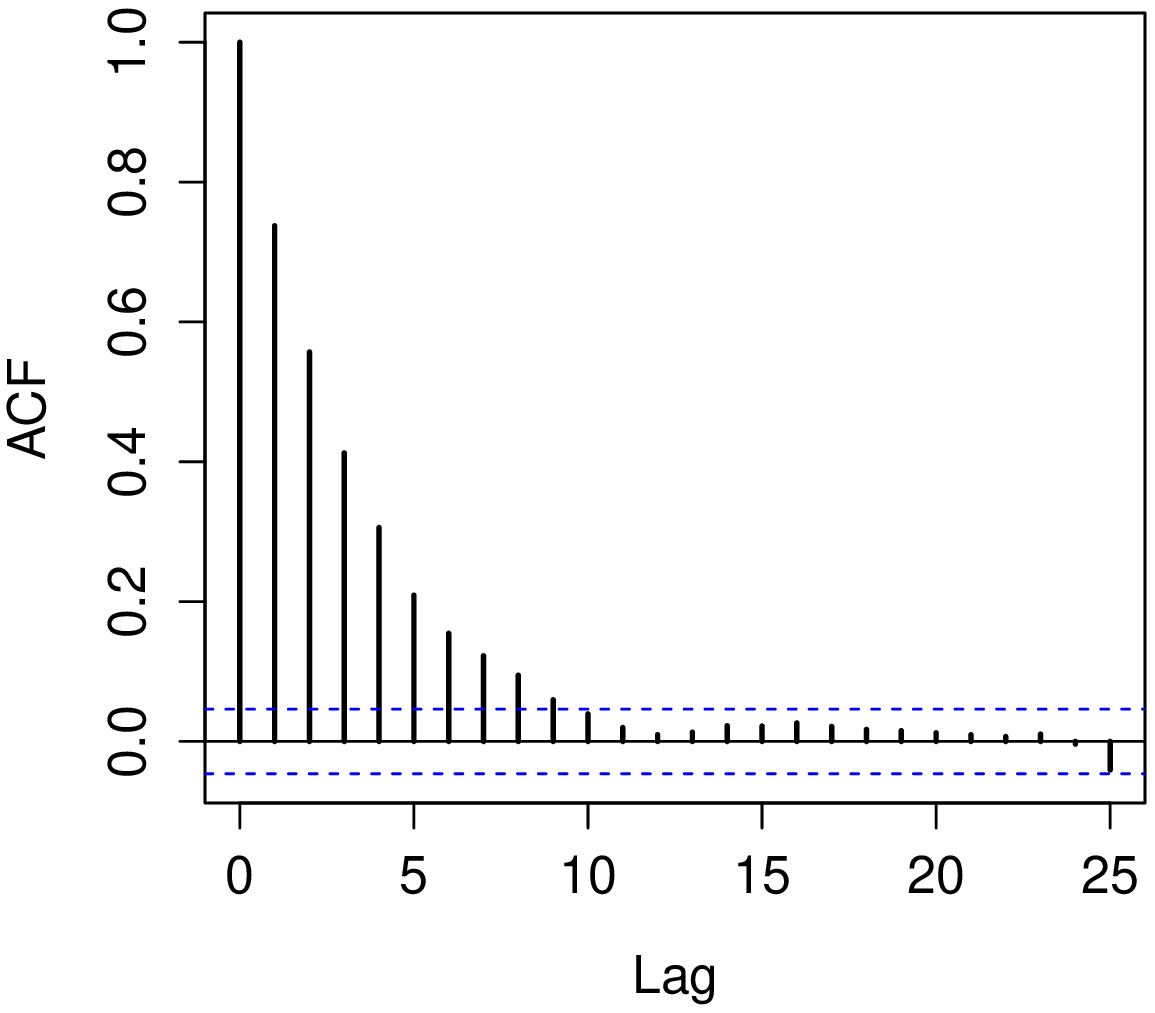} \\
   (d)
\end{minipage}   \\
   \caption{Demonstration of the Lasso sampler on a dataset with $p=100$. (a) Scatter plot of the samples of $\hbeta_1$ (x-axis) and $\hbeta_{50}$ (y-axis); (b) histogram of the subgradient $S_{50}$; (c) two sample paths of $\hbeta_1$ with diverse initial values; (d) a typical autocorrelation function.\label{fig:diagnostic}}
\end{figure}

\subsection{The bijection}

In the low-dimensional setting, we assume that $\rank(\bfX)=p\leq n$, which guarantees that the columns of $\bfX$ are in general position. 

Before writing down the bijection explicitly, we first examine the respective spaces for $\bfU$
and $(\hbmbeta,\bfS)$. Under the assumption that $\rank(\bfX)=p$, the row space of $\bfX$ is simply $\R^p$,
which is the space for $\bfU$. Let $\calA=\supp(\hbmbeta)\defi\{j:\hbeta_j \ne 0\}$ be the active set of $\hbmbeta$ and 
$\calI=\{1,\ldots,p\}\setminus \calA$ be the inactive set, i.e., the set of the zero components of $\hbmbeta$. 
After removing the degeneracies among its components as given in \eqref{eq:defS}, 
the vector $(\hbmbeta,\bfS)$ can be equivalently represented by the triple $(\hbmbeta_{\calA},\bfS_{\calI},\calA)$.
They are equivalent because %(i) $(\hbmbeta,\bfS)$ determines $(\hbmbeta_{\calA},\bfS_{\calI},\calA)$ and
from $(\hbmbeta_{\calA},\bfS_{\calI},\calA)$ one can unambiguously recover $(\hbmbeta,\bfS)$,
by setting $\hbmbeta_{\calI}=\bfzr$ and $\bfS_{\calA}=\sgn(\hbmbeta_{\calA})$ \eqref{eq:defS}, and vice versa.
It is more convenient and transparent to work with this equivalent representation. % of the augmented estimator.
One sees immediately that $(\hbmbeta_{\calA},\bfS_{\calI},\calA)$ lies in  
\begin{equation}\label{eq:space}
\Omega= \{(\bfb_A,\bfs_I,A): A\subseteq\{1,\ldots,p\}, \bfb_A\in (\R\setminus\{0\})^{|A|}, \bfs_I \in [-1,1]^{p-|A|}\},
\end{equation}
where $I=\{1,\ldots,p\}\setminus A$. Hereafter, we always understand $(\bfb_A,\bfs_I,A)$ as the equivalent
representation of $(\bfb,\bfs)=((b_j)_{1:p},(s_j)_{1:p})$ with $\supp(\bfb)=A$ and $\bfs_A=\sgn(\bfb_A)$.
Clearly, $\Omega \subset \R^p\times 2^{\{1,\ldots,p\}}$, where $2^{\{1,\ldots,p\}}$ 
is the collection of all subsets of $\{1,\ldots,p\}$, and thus $(\hbmbeta_{\calA},\bfS_{\calI},\calA)$ 
lives in the product space of $\R^p$ and a finite discrete space.

Partition $\hbmbeta$ as $(\hbmbeta_\calA,\hbmbeta_\calI)=(\hbmbeta_\calA,\bfzr)$ and
$\bfS$ as $(\bfS_\calA,\bfS_\calI)=(\sgn(\hbmbeta_\calA),\bfS_\calI)$. 
Then the KKT condition \eqref{eq:lassoKKTinU} can be rewritten,
\begin{eqnarray}
\bfU & = &(\bfC_\calA \mid \bfC_{\calI})
\left(\begin{array}{c}
\hbmbeta_{\calA} \\
\bfzr \\
\end{array}\right)
+ \lambda (\bfW_{\calA} \mid \bfW_{\calI})
\left(\begin{array}{c}
\bfS_{\calA} \\
\bfS_{\calI} \\
\end{array}\right)
-\bfC\bmbeta, \label{eq:lassojoint_a} \\
& = & \bfD(\calA)\left(
\begin{array}{l}
\hbmbeta_{\calA}  \\
\bfS_{\calI} 
\end{array}
\right) +  \lambda \bfW_{\calA} \sgn(\hbmbeta_\calA)-  \bfC\bmbeta \defi \bfH(\hbmbeta_\calA,\bfS_\calI, \calA; \bmbeta), 
\label{eq:lassojoint}
\end{eqnarray}
where $\bfD(\calA)=(\bfC_{\calA} \mid \lambda\bfW_{\calI})$ is a $p\times p$ matrix. Permuting the rows of $\bfD(\calA)$,
one sees that
\begin{equation}\label{eq:defD}
|\det\bfD(\calA)|=
\det\left(
\begin{array}{cc}
\bfC_{\calA\calA} & \bfzr \\ 
\bfC_{\calI\calA} & \lambda\bfW_{\calI\calI}
\end{array}
\right) = \lambda^{|\calI|}\det(\bfC_{\calA\calA})\prod_{j\in \calI}w_j>0
\end{equation}
if $\bfC_{\calA\calA}>0$.
Due to the equivalence between $(\hbmbeta_\calA,\bfS_\calI, \calA)$ and $(\hbmbeta,\bfS)$, 
the map $\bfH$ defined here is essentially 
the same as the one defined in \eqref{eq:lassoKKTinU}.
\begin{lemma}\label{lm:bijectioninlowdim}
If $\rank(\bfX)=p$, then for any $\bmbeta$ and $\lambda>0$, the mapping $\bfH:\Omega\to \R^p$ defined in \eqref{eq:lassojoint} is a bijection that maps $\Omega$ onto $\R^p$.
\end{lemma}
\begin{proof}
For any $\bfU \in \R^p$, there is a unique solution $(\hbmbeta,\bfS)$ to Equation \eqref{eq:lassoKKTinU} if $\rank(\bfX)=p$, 
and thus, a unique $(\hbmbeta_\calA,\bfS_\calI,\calA)\in\Omega$ 
such that $\bfH(\hbmbeta_\calA,\bfS_\calI, \calA; \bmbeta)=\bfU$. 
For any $(\hbmbeta_\calA,\bfS_\calI,\calA) \in \Omega$, $\bfH$ maps it into $\R^p$.
\end{proof}

It is helpful for understanding the map $\bfH$ to consider its inverse $\bfH^{-1}$ and its restriction to $\calA=A$,
where $A$ is a fixed subset of $\{1,\ldots,p\}$. 
For any $\bfU \in \R^p$, 
if $\bfH^{-1}(\bfU; \bmbeta) = (\hbmbeta_\calA,\bfS_\calI, \calA)$, then the unique solution to Equation~\eqref{eq:lassojoint_a} is $(\hbmbeta_\calA,\bfS_\calI, \calA)$. 
Given a fixed $A$, $(\hbmbeta_A,\bfS_I)$ lives in the subspace 
\begin{equation}\label{eq:subspace}
\Omega_A=\{(\bfb_A,\bfs_I) \in \R^p: \bfb_A\in (\R\setminus\{0\})^{|A|}, \bfs_I \in [-1,1]^{p-|A|}\}.
\end{equation}
Let $\bfH_A(\bfb_A,\bfs_I;\bmbeta) = \bfH(\bfb_A,\bfs_I,A;\bmbeta)$ for $(\bfb_A,\bfs_I)\in \Omega_A$ 
and $U_A=\bfH_A(\Omega_A;\bmbeta)$
be the image of $\Omega_A$ under the map $\bfH_A$. Now imagine we plug different $\bfU \in \R^p$
into Equation~\eqref{eq:lassojoint_a} and solve for $(\hbmbeta_\calA,\bfS_\calI, \calA)$. 
Then the set $\Omega_A\times \{A\}$ is the collection of all
possible solutions such that $\supp(\hbmbeta)=A$, the set $U_A$ is the collection of all $\bfU$ that give
these solutions, and $\bfH_A$ is a bijection between the two sets. It is easy to see that 
$\Omega=\bigcup_A \Omega_A \times \{A\}$,
i.e., $\{\Omega_A\times\{A\}\}$, for $A$ extending over all subsets of $\{1,\ldots,p\}$, 
form a partition of the space $\Omega$. The bijective
nature of $\bfH$ implies that $\{U_A\}$ also form a partition of $\R^p$, the space of $\bfU$.
Figure~\ref{fig:bijection} illustrates the bijection $\bfH$ for $p=2$ and the space partitioning by $A$. 
In this case, $\bfH_A$ map the four subspaces $\Omega_A$ for $A=\varnothing,\{1\},\{2\},\{1,2\}$,
each in a different $\R^2$, onto the space of $\bfU$ which is another $\R^2$.

\begin{figure}[ht]
\centering
   \includegraphics[width=0.55\linewidth,angle=-90,trim=0.5in 0.5in 0.5in 0.5in,clip]{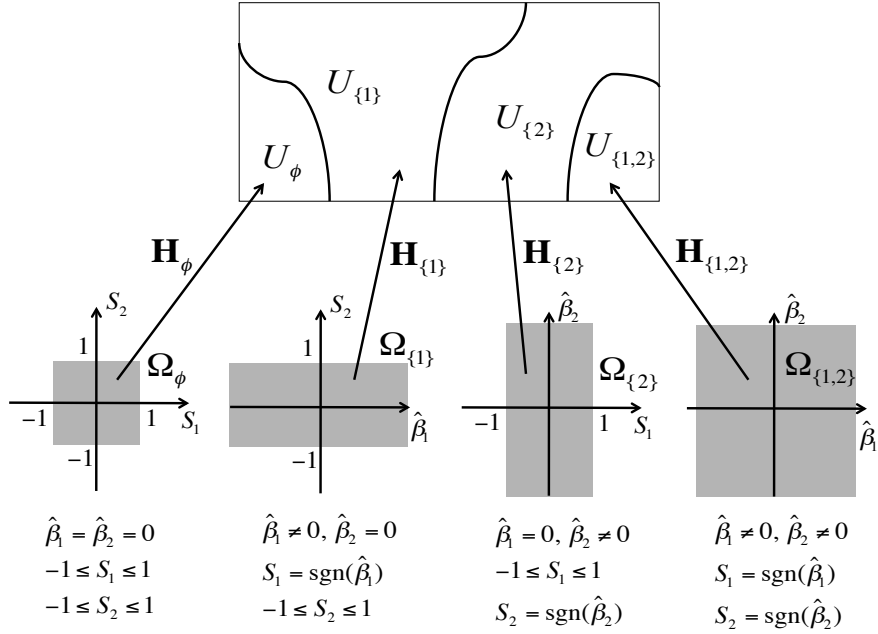} \\
   \caption{The bijection $\bfH$, its restrictions $\bfH_A$, the four subspaces $\Omega_A$ (shaded areas) and
   the corresponding partition in the space of $\bfU$ for $p=2$. 
    \label{fig:bijection}}
\end{figure}

\begin{remark}\label{rm:bijectionlowdim}
The simple fact that $\bfH$ maps every point in $\Omega$ into $\row(\bfX)=\R^p$ 
is crucial to the derivation of the sampling distribution of 
$(\hbmbeta_{\calA},\bfS_{\calI},\calA)$ in the low-dimensional setting. This means that every  
$(\bfb_A,\bfs_I,A) \in \Omega$ is the solution to Equation \eqref{eq:lassojoint_a}
for $\bfU=\bfu=\bfH(\bfb_A,\bfs_I,A;\bmbeta)$, and therefore one can simply find the probability density 
of $(\hbmbeta_{\calA},\bfS_{\calI},\calA)$ at $(\bfb_A,\bfs_I,A)$ by the density of $\bfU$ at  $\bfu$. 
This is not the case when $p>n$ (Section~\ref{sec:highdim}).
\end{remark}

\subsection{The sampling distribution}

Now we can use the bijection $\bfH$ to find the distribution of $(\hbmbeta_\calA,\bfS_\calI, \calA)$ from the distribution of $\bfU$. Let $\xi_k$ denote $k$-dimensional Lebesgue measure.  

\begin{theorem}\label{thm:joint}
Assume that $\rank(\bfX)=p$ and let $f_{\bfU}$ be the probability density of $\bfU$ with respect to $\xi_p$.
For $(\bfb_A,\bfs_I,A)\in\Omega$, the joint distribution of $(\hbmbeta_\calA,\bfS_\calI, \calA)$ is given by
\begin{eqnarray}\label{eq:jointdensity}
 P(\hbmbeta_A\in d\bfb_A,\bfS_I \in d\bfs_I,\calA=A) 
&=&f_{\bfU}(\bfH(\bfb_A,\bfs_I,A;\bmbeta))|\det\bfD(A)| \xi_p(d\bfb_A d\bfs_I) \nonumber \\
& \defi & \pi(\bfb_A,\bfs_I,A) \xi_p(d\bfb_A d\bfs_I),
\end{eqnarray}
and the distribution of $(\hbmbeta_\calA, \calA)$ is a marginal distribution given by
\begin{align}\label{eq:distrnbeta}
 P(  \hbmbeta_A \in d\bfb_A, \calA=A) %  \nonumber\\
 =\left[\int_{[-1,1]^{p-|A|}}\pi(\bfb_A,\bfs_I,A)\xi_{p-|A|}(d\bfs_I)\right]\xi_{|A|}(d\bfb_A). 
%& =|\det\bfD(A)|\left[\int_{[-1,1]^{p-|A|}}f_{\bfU}(\bfH(\bfb_A,\bfs_I,A;\bmbeta))\xi_{p-|A|}(d\bfs_I)\right]
\end{align}
\end{theorem}
\begin{proof}
Let $\bfu=\bfH(\bfb_A,\bfs_I,A;\bmbeta)= \bfH_A(\bfb_A,\bfs_I;\bmbeta)$. 
From \eqref{eq:lassojoint} and \eqref{eq:subspace}, one sees that
for any fixed $A$, $b_j\ne 0$ for all $j\in A$ and $\bfH_A$ is differentiable. Differentiating $\bfu$ with respect to $(\bfb_A,\bfs_I)$,
\begin{equation*}
d\bfu=\frac{\partial \bfH_A}{\partial(\bfb_A,\bfs_I)^{\trans}}\left(\begin{array}{l}
d\bfb_A  \\ d\bfs_I \end{array}\right)
=\bfD(A) \left(\begin{array}{l}
d\bfb_A  \\ d\bfs_I \end{array}\right)
\end{equation*}
and thus $\xi_p(d\bfu)=|\det\bfD(A)|\xi_p(d\bfb_A d\bfs_I)$.  Since $\bfH$ and 
$\bfH_A:\Omega_A \to U_A$ are bijections, a change of variable gives
\begin{align*}
& P(\hbmbeta_A\in d\bfb_A,\, \bfS_I \in d\bfs_I,\calA=A)  =  P(\bfU \in d\bfu) \\
   &\quad\quad =  f_{\bfU}(\bfH(\bfb_A,\bfs_I,A;\bmbeta)) |\det\bfD(A)|\xi_p(d\bfb_A d\bfs_I).
\end{align*}
Integrating \eqref{eq:jointdensity} over $\bfs_I \in [-1,1]^{p-|A|}$ 
gives \eqref{eq:distrnbeta}.
\end{proof}

\begin{remark}\label{rm:compare}
Equation~\eqref{eq:jointdensity} gives the joint distribution of $(\hbmbeta_\calA,\bfS_\calI, \calA)$ and 
effectively the joint distribution of $(\hbmbeta,\bfS)$. The density $\pi(\bfb_A,\bfs_I,A)$ is defined 
with respect to the product of $\xi_p$ and counting measure on $2^{\{1,\ldots,p\}}$. Analogously, the sampling distribution
of $\hbmbeta$ is given by the distribution of $(\hbmbeta_\calA, \calA)$ in \eqref{eq:distrnbeta}. To be rigorous, \eqref{eq:jointdensity} is derived by assuming that $(\bfb_A,\bfs_I)$ is an interior point of $\Omega_A$. Note that $(\bfb_A,\bfs_I)\in\Omega_A$ is not in the interior if and only if $|s_j|=1$ for some $j\in I$, and thus the Lebesgue measure of the union of these points is zero. Therefore, it will cause no problem at all to use $\pi$ as the density for all points in $\Omega$. The joint distribution of $(\hbmbeta_\calA,\bfS_\calI, \calA)$ has at least two nice properties which make it much more tractable than the distribution of $\hbmbeta$. First, the density $\pi$ does not involve multidimensional integral and has a closed-form expression that can be calculated explicitly if $f_{\bfU}$ is given. 
Second, the continuous components $(\hbmbeta_\calA,\bfS_\calI)$ % in \eqref{eq:jointdensity}
always have the same dimension $(=p)$ for any value of $\calA$, 
while $\hbmbeta_\calA$ lives in $\R^{|\calA|}$ whose dimension changes
with $\calA$. These two properties are critical to the development of MCMC to sample from $\pi$. 
See Section~\ref{sec:reversibility} for more discussion. We explicitly include the dominating Lebesgue measure to clarify
the dimension of a density.
\end{remark}

\begin{remark}\label{rm:marginal}
The distribution of $\hbmbeta=(\hbeta_1,\ldots,\hbeta_p)$ in \eqref{eq:distrnbeta} is essentially 
defined for each $A$. In many problems, one may be interested in the marginal distribution of $\hbeta_j$
such as for calculating p-values and constructing confidence intervals. To obtain such a marginal distribution,
we need to sum over all possible active sets, which cannot be done analytically. Our strategy is to 
draw samples from the joint distribution of $(\hbmbeta_\calA, \bfS_\calI,\calA)$ by a Monte Carlo method. Then from 
the Monte Carlo samples one can easily approximate any marginal distribution of interest, such as that of $\hbeta_j$. 
This is exactly our motivation
for estimator augmentation, which is in spirit similar to the use of auxiliary variables in the MCMC literature.
\end{remark}

To further help our understanding of the density $\pi$, 
consider a few conditional and marginal distributions derived from the joint
distribution \eqref{eq:jointdensity}. First, the sampling distribution of the active set $\calA$ is given by
\begin{equation}\label{eq:distrnforA}
P(\calA=A)=\int_{\Omega_A} \pi(\bfb_A,\bfs_I,A) \xi_p(d\bfb_A d\bfs_I)\defi Z_A,
\end{equation}
where $\Omega_A$ is the subspace for $(\hbmbeta_A,\bfS_I)$ defined in \eqref{eq:subspace}. In other words,
$Z_A$ is the probability of $\Omega_A\times\{A\}$ with respect to the joint distribution $\pi$. 
Second, the conditional density of $(\hbmbeta_A,\bfS_I)$ given $\calA=A$ (with respect to $\xi_p$) is
\begin{equation}\label{eq:condforbs}
\pi(\bfb_A,\bfs_I \mid A)=\frac{1}{Z_A} \pi(\bfb_A,\bfs_I,A)\propto f_{\bfU}(\bfH(\bfb_A,\bfs_I,A;\bmbeta))
\end{equation}
for $(\bfb_A,\bfs_I)\in \Omega_A \subset \R^p$. 
Using $p=2$ as an illustration, the joint density
$\pi$ is defined over all four shaded areas in Figure~\ref{fig:bijection}, 
while a conditional density $\pi(\cdot \mid A)$ is defined on each one of them.
To give a concrete probability calculation, for $a_2>a_1>0$, 
\begin{eqnarray*}
P(\hbeta_1 \in [a_1,a_2], \hbeta_2=0) & = & P(\hbeta_1 \in[a_1,a_2], \calA=\{1\}) \\
& = & \int_{a_1}^{a_2} \int_{-1}^{1} \pi(b_1,s_2,\{1\}) ds_2 db_1,
\end{eqnarray*}
which is an integral over the rectangle $[a_1,a_2] \times [-1,1]$ in $\Omega_{\{1\}}$ (Figure~\ref{fig:bijection}).
Clearly, this probability can be approximated by Monte Carlo integration if we have enough samples from $\pi$.

\begin{remark}
We emphasize that the weights $w_j$ and the tuning parameter $\lambda$ are assumed to be fixed in Theorem~\ref{thm:joint}. For the adaptive Lasso, one may choose $w_j=|\tdbeta_j|^{-1}$ based on an initial estimate $\tdbeta_j$. The distribution \eqref{eq:jointdensity} is valid for the adaptive Lasso only if we ignore the randomness in $\tdbeta_j$ and regard $w_j$ as constants during the repeated sampling procedure.
\end{remark}

\subsection{Normal errors}

Denote by $\dnorm_k(\bmmu,\bmSigma)$ the $k$-variate normal distribution with mean $\bmmu$ and covariance
matrix $\bmSigma$, and by $\phi_k(\bfz;\bmmu,\bmSigma)$ its probability density function.
If the error $\bmeps \sim \dnorm_n(\bfzr,\sigma^2 \bfI_n)$ and $\rank(\bfX)=p$, 
then $\bfU \sim \dnorm_p(\bfzr,\frac{\sigma^2}{n}\bfC)$.
In this case, the joint density $\pi$ \eqref{eq:jointdensity} has a closed-form expression.
Recall that $\bfs_A=\sgn(\bfb_A)$ and define
\begin{eqnarray}
 \bmmu (A,\bfs_A;\bmbeta) & = & [\bfD(A)]^{-1} \left( \bfC\bmbeta - \lambda  \bfW_A \bfs_A\right),  \label{eq:munorm} \\
\bmSigma(A;\sigma^2) & = & \frac{\sigma^2}{n}[\bfD(A)]^{-1}\bfC[\bfD(A)]^{-\trans}.  \label{eq:varnorm}
\end{eqnarray}
\begin{corollary}\label{cor:normaljoint}
If $\rank(\bfX)=p$ and $\bmeps \sim \dnorm_n(\bfzr,\sigma^2 \bfI_n)$, then 
the joint density of $(\hbmbeta_{\calA},\bfS_{\calI},\calA)$ is
\begin{equation}\label{eq:jointdensitynormal}
\pi(\bfb_A,\bfs_I,A) =\phi_p(\bfz ; \bmmu(A,\bfs_A;\bmbeta),\bmSigma(A;\sigma^2))\I((\bfz,A)\in \Omega), 
\end{equation}
where $\bfz=(\bfb_A,\bfs_I)\in\R^p$ and $\I(\cdot)$ is an indicator function.
\end{corollary}
\begin{proof}
First note that 
\begin{equation}\label{eq:affineKKT}
\bfH(\bfb_A,\bfs_I,A;\bmbeta)=\bfD(A)[\bfz-\bmmu(A,\bfs_A;\bmbeta)].
\end{equation}
Under the assumptions, $\bfU \sim \dnorm_p(\bfzr,\frac{\sigma^2}{n}\bfC)$. By Theorem~\ref{thm:joint},
\begin{eqnarray*}
\pi(\bfb_A,\bfs_I,A)&=&\phi_p\left(\bfD(A)[\bfz-\bmmu(A,\bfs_A;\bmbeta)]; \bfzr,n^{-1}{\sigma^2}\bfC\right)|\det\bfD(A)| \\
& = & \phi_p\left(\bfz ; \bmmu(A,\bfs_A;\bmbeta),n^{-1}{\sigma^2}[\bfD(A)]^{-1}\bfC[\bfD(A)]^{-\trans}\right) \\
& = & \phi_p(\bfz ; \bmmu(A,\bfs_A;\bmbeta),\bmSigma(A;\sigma^2))
\end{eqnarray*}
for $(\bfb_A,\bfs_I,A)=(\bfz,A)\in \Omega$.
\end{proof}

Without the normal error assumption, 
Corollary~\ref{cor:normaljoint} is still a good approximation when $n$ is large and $p$ is fixed, 
since $\sqn \bfU \toinL \dnorm_p(\bfzr,\sigma^2\bfC)$ assuming $\frac{1}{n}\bfX^{\trans}\bfX\to \bfC>0$ as $n\to\infty$.  

Note that both the continuous components $\bfz$ and the active set $A$ are 
arguments of the density \eqref{eq:jointdensitynormal}.
For different $A$ and $\bfs_A$, the normal density $\phi_p$ has different parameters.
Given a particular $A^*$ and $\bfs^*\in \{\pm 1\}^{|A^*|}$, let $I^*=\{1,\ldots,p\}\setminus A^*$ and
\begin{equation}\label{eq:subsetAs}
\Omega_{A^*,\bfs^*}=\{(\bfb_{A^*},\bfs_{I^*})\in\Omega_{A^*}: \sgn(\bfb_{A^*})=\bfs^*\}.
\end{equation}
Then $\Omega_{A^*,\bfs^*}\times\{A^*\}$ is the subset of $\Omega$ corresponding to 
the event $\{\calA=A^*, \sgn(\hbmbeta_{A^*})=\bfs^*\}$.
For $\bfz \in\Omega_{A^*,\bfs^*}$, the density $\pi(\bfz,A^*)$ is identical to
$\phi_p(\bfz ; \bmmu(A^*,\bfs^*;\bmbeta),\bmSigma(A^*;\sigma^2))$, i.e.,
\begin{equation*}
\pi(\bfz,A^*)\I(\bfz \in\Omega_{A^*,\bfs^*})
=\phi_p(\bfz ; \bmmu(A^*,\bfs^*;\bmbeta),\bmSigma(A^*;\sigma^2))\I(\bfz \in\Omega_{A^*,\bfs^*}).
\end{equation*}
Intuitively, this is because $\bfH$ restricted to $A=A^*$ and $\bfs_{A^*}=\bfs^*$ is simply an affine map
[see \eqref{eq:affineKKT}].
Consequently, the probability of $\Omega_{A^*,\bfs^*}\times\{A^*\}$ with respect to $\pi$ is 
\begin{equation}\label{eq:signpm}
P(\calA=A^*, \sgn(\hbmbeta_{A^*})=\bfs^*)=\int_{\Omega_{A^*,\bfs^*}} \phi_p(\bfz ; \bmmu(A^*,\bfs^*;\bmbeta),\bmSigma(A^*;\sigma^2)) \xi_p(d\bfz),
\end{equation}
and $[\hbmbeta_{A^*},\bfS_{I^*} \mid \calA=A^*,\sgn(\hbmbeta_{A^*})=\bfs^*]$
is the truncated $\dnorm_p(\bmmu(A^*,\bfs^*;\bmbeta),\bmSigma(A^*;\sigma^2))$ on $\Omega_{A^*,\bfs^*}$.
For $p=2$, if $A^*=\{1\}$, and $\bfs^*=-1$, the region $\Omega_{\{1\},-1}=(-\infty,0)\times[-1,1]$ is
the left half of the $\Omega_{\{1\}}$ in Figure~\ref{fig:bijection} and the density $\pi$ restricted to this region is
the same as the part of a bivariate normal density on the same region.

If $\bfC=\bfI_p$ and $\bfW=\bfI_p$, the Lasso is equivalent to soft-thresholding the ordinary least-squares estimator $\hbmbeta^{\text{OLS}}=(\hbeta^{\text{OLS}}_j)_{1:p}$. In this case, \eqref{eq:munorm} and \eqref{eq:varnorm} have simpler forms:
\begin{eqnarray*}
\bmmu (A,\bfs_A;\bmbeta)& =&
\left(
\begin{array}{c}
\bmbeta_{A}-\lambda \bfs_{A}  \\
\lambda^{-1} \bmbeta_I
\end{array}
\right), \\
\bmSigma(A;\sigma^2)& = & 
\frac{\sigma^2}{n}\left( \begin{array}{cc}
\bfI_{|A|} &  \bfzr \\
\bfzr & \lambda^{-2} \bfI_{|I|}
 \end{array}\right).
\end{eqnarray*}
By \eqref{eq:signpm} we find, for example,
\begin{align*}
&P(\calA=A,\sgn(\hbmbeta_A)=(1,\ldots,1)) \\ 
& = \prod_{j\in A} \int_{0}^{\infty} \phi\left(b_j;\beta_j-\lambda,\frac{\sigma^2}{n}\right) db_j \cdot
\prod_{j\in I} \int_{-1}^1 \phi\left(s_j;\frac{\beta_j}{\lambda},\frac{\sigma^2}{\lambda^2 n}\right) ds_j \\
&= \prod_{j\in A} P(\hbeta^{\text{OLS}}_j> \lambda) \cdot  \prod_{j\in I} P(|\hbeta^{\text{OLS}}_j| \leq \lambda),
\end{align*}
where the last equality is due to that $\hbmbeta^{\text{OLS}} \sim \dnorm_p(\bmbeta,n^{-1}\sigma^2 \bfI_p)$. One sees that our result is consistent with that  obtained directly from soft-thresholding each component of $\hbmbeta^{\text{OLS}}$ by $\lambda$.

\subsection{Estimation}\label{sec:asymest}

To apply Theorem~\ref{thm:joint} in practice, one needs to estimate $f_{\bfU}$ and $\bmbeta$ if they are not given.
Suppose that $f_{\bfU}$ is estimated by $\hat{f}_{\bfU}$ and $\bmbeta$ is estimated by 
$\ckbmbeta$. Then, the corresponding estimate of the density $\pi$ is 
\begin{equation}\label{eq:estjointdensity}
\hat{\pi}(\bfb_A,\bfs_I,A) = \hat{f}_{\bfU}(\bfH(\bfb_A,\bfs_I,A;\ckbmbeta))|\det\bfD(A)|.
\end{equation}

Since $\E(\bfU)=\bfzr$ and $\Var(\sqn\bfU)=\sigma^2\bfC$, estimating $f_{\bfU}$
reduces to estimating $\sigma^2$ when $\bmeps$ is normally distributed or when the sample size $n$ is large.
A consistent estimator of $\sigma^2$ can be constructed given a consistent estimator of $\bmbeta$. 
For example, when $p<n$ one may use
\begin{equation}\label{eq:estsigma2}
\hat{\sigma}^2 = \frac{\| \bfY-\bfX \ckbmbeta\|_2^2}{n-p},
\end{equation}
provided that $\ckbmbeta$ is consistent for $\bmbeta$.
If $\bmeps$ does not follow a normal distribution,
one can apply other parametric or nonparametric methods to estimate $f_{\bfU}$.
Here, we propose a bootstrap-based approach under the assumption that $\bfU$ is elliptically symmetric.
That is, $\bftdU=\bfC^{-1/2}\bfU$ is spherically symmetric: For $\bfv_1,\bfv_2 \in \R^p$,
if $\|\bfv_1\|_2=\|\bfv_2\|_2$ then $f_{\bftdU}(\bfv_1)=f_{\bftdU}(\bfv_2)$, 
where $f_{\bftdU}$ is the density of $\bftdU$.
Generate bootstrap samples, $\bmeps^{(i)}=(\varepsilon^{(i)}_1,\ldots,\varepsilon^{(i)}_n)$ for
$i=1,\ldots,K$, by resampling with replacement from 
$\hat{\bmeps}=(\hat{\varepsilon}_1,\ldots,\hat{\varepsilon}_n)=\bfY-\bfX \ckbmbeta$,
and calculate $\bftdU^{(i)}=\frac{1}{n}\bfC^{-1/2}\bfX^{\trans}\bmeps^{(i)}$ for each $i$.
Given $0=h_0 < h_1 < \cdots < h_M<\infty$, let $K_m=|\{i:h_{m-1} \leq \|\bftdU^{(i)}\|_2 < h_m\}|$ for $m=1,\ldots,M$.
The density of $\bftdU$ is then estimated by
\begin{equation}
\hat{f}_{\bftdU}(\bfv) \propto \sum_{m=1}^M \frac{K_m}{h_m^p - h_{m-1}^p}  \I (h_{m-1} \leq \|\bfv\|_2 < h_m)
\end{equation}
for $\|\bfv\|_2 \in [0,h_M)$. The density for $\|\bfv\|_2 \geq h_M$ can be estimated by linear extrapolation 
of $\log\hat{f}_{\bftdU}$. Finally, set $\hat{f}_{\bfU}(\bfu) = \hat{f}_{\bftdU}(\bfC^{-1/2}\bfu)(\det\bfC)^{-1/2}$.

In general, estimating $f_{\bfU}$ is difficult when $p$ is large. One may have to assume some parametric
density for $\bfU$, which reduces the problem to the estimation of a few unknown parameters.
Besides normality, one may assume that $\bfU$ follows a multivariate $t$ distribution, which 
is motivated from a Bayesian perspective to be discussed in Section~\ref{sec:Bayesian}.

Sampling from $\pi$ or $\hat{\pi}$ can be very useful for statistical inference based on a Lasso-type estimator. We may directly draw $(\sbmbeta,\bfS^*)$ from $\hat{\pi}$ given $(\ckbmbeta,\hsigma)$ and use the conditional distribution $[(\sbmbeta-\ckbmbeta)\mid \ckbmbeta,\hsigma]$ to construct confidence regions around $\hbmbeta$. Under some assumptions, the conditional distribution $[(\sbmbeta-\ckbmbeta)\mid \ckbmbeta,\hsigma]$ provides a valid approximation to the true sampling distribution of $(\hbmbeta-\bmbeta)$. We derive nonasymptotic error bounds for this approximation in Section~\ref{sec:errorbounds} after the development of our method in the high-dimensional setting. If $\bmbeta$ is specified in the null hypothesis in a significance test, then samples from $\pi$ can be used to calculate p-values. This aspect will be explored in Section~\ref{sec:pvalue}.

\section{MCMC algorithms}\label{sec:MC}

In this section, we develop MCMC algorithms to sample from $\pi$ given $\bmbeta$
and $f_{\bfU}$ (or $\sigma^2$). Before that,
we first introduce a direct sampling approach which includes the residual bootstrap method as a special case. 
\begin{routine}[Direct sampler]\label{rt:directsample}
Assume the error distribution is $\calD(\bfzr,\sigma^2 \bfI_n)$. For $t=1,\ldots,L$
\begin{itemize}
\item[(1)]{draw $\bmeps^{(t)}\sim \calD(\bfzr,\sigma^2 \bfI_n)$ and set $\bfY^{(t)} = \bfX \bmbeta + \bmeps^{(t)}$;}
\item[(2)]{find the minimizer $\hbmbeta^{(t)}$ of \eqref{eq:lassoloss} with $\bfY^{(t)}$ in place of $\bfY$;}
\item[(3)] if needed, calculate the subgradient vector $\bfS^{(t)}=(n\lambda\bfW)^{-1}\bfX^{\trans}(\bfY^{(t)}-\bfX\hbmbeta^{(t)})$.
\end{itemize}
\end{routine}
This approach directly draws $\bfY^{(t)}$ from its sampling distribution and requires a numerical optimization
algorithm in step (2) for each sample. Moreover, step (1) will be complicated if
we cannot draw independent samples from $\calD(\bfzr,\sigma^2 \bfI_n)$. If $\bmeps^{(t)}$ is drawn by
resampling residuals, then Routine~\ref{rt:directsample} is equivalent to the bootstrap method of \cite{KnightFu00}.

As the density ${\pi}(\bfb_A,\bfs_I,A)$ \eqref{eq:jointdensity} has a closed-form expression given $\bmbeta$
and $f_{\bfU}$, MCMC and IS 
can be applied to sample from and calculate expectations with respect to the distribution.
These methods may offer much more flexible and efficient alternatives to the direct sampling approach, 
although the samples are either dependent or weighted.
In what follows, we propose a few special designs targeting at different applications to 
exemplify the use of MCMC methods. Examples of IS will be given in Section~\ref{sec:pvalue} under
the high-dimensional setting.

\subsection{Reversibility}\label{sec:reversibility}
 
Our goal is to design a reversible Markov chain on the space $\Omega$, which is composed
of a finite number of subspaces $\Omega_A$, each having the same dimension $p$. Therefore, moves with an ordinary 
Metropolis-Hastings (MH) ratio are sufficient, which can be seen as follows.
For any $(\bfb_A,\bfs_I,A)\in\Omega$, let $\bmtheta=(\theta_1,\ldots,\theta_p)$ with components given by
\begin{equation}\label{eq:definetheta}
\theta_j=\left\{\begin{array}{ll}
b_j & \mbox{ if } j\in A \\
s_j & \mbox { otherwise},
\end{array}\right.
\end{equation}
i.e., $\bmtheta_A=\bfb_A$ and $\bmtheta_I=\bfs_I$.  
Then our target distribution is
$\pi(\bmtheta_A,\bmtheta_I,A)\xi_p(d\bmtheta)$.
Suppose that $(\bmtheta,A)$ is the current state and we 
have a proposal for a new state $(\bmtheta^{\dag},A^{\dag})$.
In general, the proposal may only change some components of $\bmtheta$,
say $\theta_j$ for $j \in B\subseteq \{1,\ldots,p\}$, such that $\bmtheta^{\dag}_{-B}=\bmtheta_{-B}$. Let 
$q((\bmtheta,A),(\bmtheta^{\dag},A^{\dag}))$ be the density of this proposal with respect to $\xi_{|B|}$ and $I^{\dag}=\{1,\ldots,p\}\setminus A^{\dag}$. The MH ratio in terms of probability measures is 
\begin{eqnarray}
& &\min\left\{1, \frac{\pi(\bmtheta^{\dag}_{A^{\dag}},\bmtheta^{\dag}_{I^{\dag}},A^{\dag}) 
\xi_p(d\bmtheta^{\dag})}{\pi(\bmtheta_A,\bmtheta_I,A)\xi_p(d\bmtheta)} \frac{q((\bmtheta^{\dag},A^{\dag}),(\bmtheta,A))
\xi_{|B|}(d\bmtheta_B)}{q((\bmtheta,A),(\bmtheta^{\dag},A^{\dag}))\xi_{|B|}(d\bmtheta^{\dag}_B)} \right\} \nonumber\\
&=& \min\left\{1, \frac{\pi(\bmtheta^{\dag}_{A^{\dag}},\bmtheta^{\dag}_{I^{\dag}},A^{\dag})}{\pi(\bmtheta_A,\bmtheta_I,A)} 
\frac{q((\bmtheta^{\dag},A^{\dag}),(\bmtheta,A))}{q((\bmtheta,A),(\bmtheta^{\dag},A^{\dag}))}
\frac{\xi_{p-|B|}(d\bmtheta^{\dag}_{-B})}{\xi_{p-|B|}(d\bmtheta_{-B})} \right\}.
\label{eq:MHratio}
\end{eqnarray}
As $\bmtheta^{\dag}_{-B}=\bmtheta_{-B}$, the dominating measures in \eqref{eq:MHratio} cancel out 
and the ratio reduces to a standard MH ratio involving only densities.

Now we see that our strategy of estimator augmentation plays two roles in MCMC sampling. First,
$\bfS_\calI$ plays the role of an auxiliary variable: The target distribution $\pi$ for $(\hbmbeta_{\calA},\bfS_\calI,\calA)$
has a closed-form density which allows one to design an MCMC algorithm, while the distribution of interest,
that for $(\hbmbeta_{\calA},\calA)$, is a marginal distribution of $\pi$ without a closed-form density. 
Second, $\bfS_\calI$ also plays the role
of dimension matching so that the continuous components $(\hbmbeta_{\calA},\bfS_\calI)$ always have the
same dimension in any subspace. This eliminates the need for reversible jump MCMC \citep{Green95}. On the contrary,
if we were to sample $(\hbmbeta_{\calA},\calA)$ (assuming a closed-form approximation to its density),
moves between two subspaces of different dimensions
would require reversible jumps, which are usually much harder to design.

\subsection{The MH Lasso sampler}\label{sec:MHLasso}

We develop an MH algorithm, called the MH Lasso sampler (MLS), with coordinate-wise update. That is,
to sequentially update each $\theta_j$, $j=1,\ldots,p$, while holding other components fixed.
Suppose the current state is $(\bmtheta,A)$. 
We design four moves to propose a new state $(\bmtheta^{\dag},A^{\dag})$, which are grouped into two types
according to whether $A^{\dag}=A$ or not. In the following proposals,
$\theta_j^{\dag}=b_j^{\dag}$ if $j\in A^{\dag}$ and $\theta_j^{\dag}=s_j^{\dag}$ otherwise.
\begin{definition} \label{def:proposals}
Proposals in the MLS for a given $j\in\{1,\ldots,p\}$.
\begin{itemize}
\item{Parameter-update proposals: 
(P1) If $j\in A$, draw $b_j^{\dag}\sim \dnorm(b_j,\tau_j^2)$. 
(P2) If $j \notin A$, draw $s_j^{\dag}\sim\text{Unif}(-1,1)$. Set $A^{\dag}=A$ in both (P1) and (P2).
}
\item{Model-update proposals:
(P3) If $j\in A$, set $A^{\dag}=A\setminus \{j\}$ and draw $s_j^{\dag}\sim\text{Unif}(-1,1)$. 
(P4) If $j \notin A$, set $A^{\dag}=A\cup \{j\}$ and draw $b_j^{\dag}\sim \dnorm(0,\tau_j^2)$. 
}
\end{itemize}
\end{definition}

The two parameter-update proposals, (P1) and (P2), are symmetric. They only change the value of $\theta_j$ and leave
$A^{\dag} =A$ so that $\det \bfD(A^{\dag})=\det \bfD(A)$. From \eqref{eq:jointdensity}, 
one sees that the MH ratio is simply
\begin{equation*} 
\min\left\{1,\frac{f_{\bfU}(\bfH(\bmtheta^{\dag}_{A},\bmtheta^{\dag}_{I},A;\bmbeta))}
{f_{\bfU}(\bfH(\bmtheta_A,\bmtheta_I,A;\bmbeta))}\right\},
\end{equation*}
which can be computed very efficiently, especially for a normal error distribution. 
The proposal (P3) removes a variable from the active set and (P4) adds a variable to the active set. 
Both propose moves between two subspaces. The two proposals are the
reverse of each other and have a simple one-dimensional density. 
To be concrete, the MH ratio for proposal (P3) is
\begin{equation*}
\min\left\{1, \frac{\pi(\bmtheta^{\dag}_{A^{\dag}},\bmtheta^{\dag}_{I^{\dag}},A^{\dag})}
{\pi(\bmtheta_A,\bmtheta_I,A)}\cdot\frac{\phi(b_j; 0, \tau^2_j)}{1/2} \right\},
\end{equation*}
and analogously for proposal (P4).
One needs to calculate the ratio between two determinants for these MH ratios, 
\begin{equation}\label{eq:detDratio}
\frac{|\det \bfD(A^\dag)|}{|\det \bfD(A)|} = \frac{\det\bfC_{A^{\dag}A^{\dag}}}{\det\bfC_{AA}} (w_j\lambda)^{|A|-|A^{\dag}|}
\end{equation} 
by \eqref{eq:defD}. As the two sets $A$ and $A^{\dag}$ differ by only one element, the ratio on the right-hand side can be calculated efficiently. When $|A|$ is large, we use the sweep operator to dynamically update $\bfC_{AA}^{-1}$ (the inverse of $\bfC_{AA}$) and obtain the ratio. See Appendix for further details. In general, however, a model-update proposal is more time-consuming than a parameter-update proposal. 

This computational efficiency consideration motivates
the following scheme in the MLS which uses both types of proposals.
Let $K$ be an integer between 1 and $p$ and $\bmalpha=(\alpha_j)_{1:p}$ be a vector with every
$\alpha_j>0$.
\begin{routine}[MLS]\label{rt:MHLS}
Suppose the current state is $(\bmtheta^{(t)},A^{(t)})$. 
\begin{itemize}
\item[(1)]{Draw $K$ elements without replacement from $\{1,\ldots,p\}$ with the probability of drawing 
$j$ proportional to $\alpha_j$ for each $j$. Let $M^{(t)}$ be the set of the $K$ elements.}
\item[(2)]{For $j \in M^{(t)}$, sequentially update each $\theta_j$ and the active set $A$ 
by an MH step with a model-update proposal.}
\item[(3)]{For $j \notin M^{(t)}$, sequentially update each $\theta_j$ by an MH step with a parameter-update proposal.}
\end{itemize}
After the above $p$ MH steps in an iteration, the state is updated to $(\bmtheta^{(t+1)},A^{(t+1)})$.
\end{routine}

The MLS has three input parameters, $K$, $\bmalpha$, and $(\tau^2_j)_{1:p}$. 
Specification of these parameters that gives good empirical performance will be provided in the numerical
examples (Section~\ref{sec:numerical}).

\subsection{The Gibbs Lasso sampler}

Let $a_j=\I(j \in A)$ and $\bfa=(a_j)_{1:p}$. 
Conditional densities $\pi(\theta_j, a_j \mid \bmtheta_{-j}, \bfa_{-j})$ can be derived from the joint density $\pi$,
which allows for the development of a Gibbs sampler. However, as each conditional
sampling step involves calculation of one-dimensional integrals and sampling from truncated distributions,
the Gibbs sampler is more time-consuming and less efficient than the MLS for all examples on which
we have tested these algorithms. 

\subsection{Conditioning on active set}

Suppose that we have constructed a Lasso-type estimate $\hbmbeta^*$ from an observed dataset
and the set of selected variables is $A^*$, which defines an estimated model.
One may want to study the sampling distribution of the estimator given the estimated model, i.e.,
$[\hbmbeta_{A^*} \mid \calA = A^*]$. Confidence intervals of penalized estimators
have been constructed by approximating this distribution via local expansion of the $\ell_1$ norm \citep{FanLi01,Zou06}. 
Since local approximation 
may not be accurate for a finite sample, Monte Carlo sampling from this conditional distribution
may provide more accurate results. However, the direct sampling approach is not applicable
in practice, because $\calA = A^*$ is often a rare event unless $p$ is very small. On the contrary, it is very efficient
to draw samples by an MH algorithm from the conditional distribution
\begin{equation}\label{eq:condAdensity}
{\pi}(\bfb_{A^*},\bfs_{I^*}\mid A^*) \propto {f}_{\bfU}(\bfH(\bfb_{A^*},\bfs_{I^*},A^*;\bmbeta)),
\end{equation}
where $I^*=\{1,\ldots,p\}\setminus A^*$, according to \eqref{eq:condforbs}.
The distribution of interest, $[\hbmbeta_{A^*} \mid \calA = A^*]$, is a marginal
distribution of \eqref{eq:condAdensity}.
Since evaluation of this density does not involve calculation of determinants, 
each MH step is very fast.
\begin{routine}[MLS given active set]\label{rt:MLSfixedA}
Given the current state $(\bfb_{A^*}^{(t)},\bfs_{I^*}^{(t)})$, 
sequentially draw $b_j^{(t+1)}$ for each $j\in A^*$ by an MH step with proposal (P1) and 
$s_j^{(t+1)}$ for each $j \notin A^*$ with proposal (P2) in one iteration.
\end{routine}

\subsection{Reparameterization view}\label{sec:viewofMCMC}

To ease notation, write $\bmtheta=(\bmtheta_A,\bmtheta_I)$ and $\bfH_A(\bmtheta)=\bfH(\bmtheta,A)=\bfH(\bmtheta_A,\bmtheta_I,A;\bmbeta)$ for $\bmtheta$ defined in \eqref{eq:definetheta}. Suppose that $(\bmtheta^{(t)},A^{(t)})$ are simulated by the MLS (Routine~\ref{rt:MHLS}) and let $\bfu^{(t)}=\bfH(\bmtheta^{(t)},A^{(t)})$. Since $\bfH$ is a bijection, $\bfu^{(t)}$ is a Markov chain that leaves $f_{\bfU}$ invariant. Therefore, the MLS can be understood as an MH algorithm targeting at $f_{\bfU}$ with moves designed under local reparameterization, $\bmtheta=\bfH_{A}^{-1}(\bfu)$ for $\bfu \in U_A=\bfH_A(\Omega_A;\bmbeta)$ \eqref{eq:subspace}. The Jacobian of this reparameterization is $[\bfD(A)]^{-1}$. Under this view, the MH ratio for a proposal $\bfu^{\dag}=\bfH(\bmtheta^{\dag},{A^\dag})$ given the current $\bfu$ is 
\begin{equation*}
\min\left\{1, \frac{f_{\bfU}(\bfu^{\dag}) }{f_{\bfU}(\bfu)} 
\frac{ q((\bmtheta^{\dag},A^{\dag}),(\bmtheta,A))|\det \bfD(A)|^{-1}}{q((\bmtheta,A),(\bmtheta^{\dag},A^{\dag}))|\det \bfD(A^\dag)|^{-1}} \right\},
\end{equation*}
which of course coincides with \eqref{eq:MHratio}. Moreover, when $A=A^{\dag}$ such as in proposals (P1) and (P2), the Jacobian determinants cancel out as we are using the same reparameterization $\bfH_A^{-1}$ for both the proposal and the current state. Otherwise, the ratio of the Jacobian determinants accounts for the use of different reparameterizations. Clearly, if $\bfu^{(t)}\sim f_{\bfU}$ then $(\bmtheta^{(t)},A^{(t)})\sim \pi$. 

This view provides an insight into the computational efficiency of the MLS. Under normal error assumption, $f_{\bfU}$ is the density of a multivariate normal distribution, for which a simple MH algorithm is computationally tractable and can be quite efficient. We make a comparison with the direct sampler at a conceptual level. The direct sampler draws $\bfU$ via a linear transformation of $\bmeps\sim\dnorm_n(\bfzr,\sigma^2\bfI_n)$, which costs $n$ draws from a univariate normal distribution followed by a multiplication with a size $p\times n$ matrix. After that, we find $(\bmtheta,A)=\bfH^{-1}(\bfU)$ by numerical minimization due to the lack of a closed-form inverse of the mapping $\bfH$. The MLS draws $p$ ($<n$) univariate proposals in one iteration, and does not need any numerical procedure to map $\bfu^{(t)}$ back to the space $\Omega$ of $(\bmtheta,A)$ since the moves are by design in that space already. The mapping $\bfH$ from $(\bmtheta^{(t)},A^{(t)})$ to $\bfu^{(t)}$ is simple and can be calculated analytically. This is fundamentally different from direct sampling which replies on a numerical procedure to find the image of each draw of $\bfU$ under the mapping $\bfH^{-1}$. The relatively time-consuming step in the MLS is calculating the ratio \eqref{eq:detDratio} when a model-update proposal is used, which can be done by at most sweeping a $|A|\times |A|$ matrix on a single position (Appendix). Owing to sparsity, $|A|$ is usually much smaller than $p$, which greatly speeds up this step. Since the target distribution in the space of $\bfU$ has a nice unimodal density, the chain $\bfu^{(t)}$ often converges fast and has low autocorrelation. Consequently, we expect to see efficiency gain over direct sampling for the same amount of computing time, which will be confirmed numerically in the next subsection.

As in the following routine, by a special initialization such that $(\bmtheta^{(1)},A^{(1)})\sim\pi$, the MLS can reach equilibrium in one step, which totally removes the need for burn-in iterations. This will make our method suitable for parallel computing. See Section~\ref{sec:limitation} for a more detailed discussion. However, to demonstrate the efficiency of the MLS as an independent method, we did not use Routine~\ref{rt:dirmcmc} in the numerical results. 

\begin{routine}\label{rt:dirmcmc}
Draw $(\bfb^{(1)},\bfs^{(1)})$ from the direct sampler and let $(\bmtheta^{(1)},A^{(1)})$ be its equivalent representation. With $(\bmtheta^{(1)},A^{(1)})$ as the initial state, generate $(\bmtheta^{(t)},A^{(t)})$ for $t=2,\ldots,N$ by an MCMC algorithm targeting at $\pi$.
\end{routine}

\subsection{Numerical examples}\label{sec:numerical}

We demonstrate with numerical examples the effectiveness of the above MCMC
algorithms by comparing against the direct sampling approach. To this end,
we simulated four datasets with different combinations of $n$, $p$, 
and $\sigma^2$ (Table~\ref{tab:data}). The vector of true coefficients $\bmbeta_0$ has 10 nonzero components, 
$\beta_{0j}=1$ for $j=1,\ldots,5$ and $\beta_{0j}=-1$ for $j=6,\ldots,10$. Each row of $\bfX$ was generated independently from $\dnorm_p(\bfzr,\bmSigma_{\bfX})$, where the diagonal and the off-diagonal elements of $\bmSigma_{\bfX}$
are 1 and $0.25$, respectively. Given the design matrix $\bfX$, the response vector $\bfY$ was drawn
from $\dnorm_n(\bfX \bmbeta_0, \sigma^2 \bfI_n)$.

\begin{table}[ht]
\centering
\caption{Simulated datasets for MCMC\label{tab:data}}
\vspace{0.05in}
   \begin{tabular}{c|cccc} 
   \hline
     Dataset	  	&  A & B & C &D  \\
     \hline
   $(n,p,\sigma^2)$ &$(500,100,1)$ & $(500,200,1)$ & $(300,100,4)$ & $(300,200,4)$ \\
   $|A^*|$ & 23 & 22 & 25 & 57 \\
  \hline
  \end{tabular}
\end{table}

The weights $w_j$ \eqref{eq:lassoloss} were set to 1 for all the following numerical results.
The Lars package by Hastie and Efron was applied to find the solution path for each dataset. The value of
$\lambda$ was chosen by minimizing the $C_p$ criterion implemented in the package,
which determined the estimated coefficients, $\hbmbeta^*=(\hbeta^*_1,\ldots,\hbeta^*_p)$, of a dataset.
The number of selected variables, $|A^*|$, for each dataset is given in Table~\ref{tab:data}.
We considered two types of error distributions, 
the normal distribution and the elliptically symmetric distribution.
Correspondingly, we calculated $\hat{\sigma}^2$ by \eqref{eq:estsigma2} with $\ckbmbeta=\hbmbeta^*$
or constructed $\hat{f}_{\bfU}$ by the approach in Section~\ref{sec:asymest}. 
For all the results, step (2) of the direct sampler (Routine~\ref{rt:directsample}) was implemented with the Lars package.

We first examined the performance of the MLS on sampling from the joint distribution
\eqref{eq:jointdensity} given $\hbmbeta^*$ and $\hat{\sigma}^2$ or 
$\hat{f}_{\bfU}$. Let $\omega_j = \Phi(-|\hat{\beta}^*_{j}|/\zeta_j)$ for $j=1,\ldots,p$, where $\zeta_j$ is the
standard error of $\hbeta^{\text{OLS}}_j$ and $\Phi$ is the cumulative distribution function of $\dnorm(0,1)$.
We set $K=p/5$ and $\alpha_j \propto \omega_j + \omega_0$, where $\omega_0=\sum_j \omega_j/(5p)$
serves as a baseline weight so that each variable has a reasonable chance to be selected for model-update proposals.
See Routine~\ref{rt:MHLS} for notations. Under this setting, if the estimate $\hat{\beta}^*_{j}$
is close to zero relative to $\zeta_j$, it will have a higher chance for model-update proposals.
The $\tau_j$ used in the proposals (Definition~\ref{def:proposals}) 
was set to $2\zeta_j$.
The MLS was applied to each dataset 10 times independently. 
Each run consisted of $L=5,500$ iterations with the first 500 as the burn-in period. 
In what follows, the sampler is abbreviated as MLSn and MLSe under the normal and the
elliptically symmetric error distributions, respectively.

Figure~\ref{fig:diagnostic}(a) is the scatter plot of the samples of $\hbeta_1$ and $\hbeta_{50}$,
and illustrates that the distributions of some $\hbeta_j$ indeed have a point mass at zero. The histogram of the subgradient $S_{50}$ is shown in Figure~\ref{fig:diagnostic}(b) with two point masses on $\pm 1$ and otherwise continuous.
Mixing of the MLS was fast, as demonstrated with two chains in Figure~\ref{fig:diagnostic}(c),
where the initial values were chosen to be about 20 standard deviations away
from each other. Figure~\ref{fig:diagnostic}(d) shows the fast decay of the autocorrelation among the samples of a $\hbeta_j$, decreasing to below 0.05 in 10 to 15 iterations. The acceptance rate of the model-update proposals was generally between 0.2 and 0.4.
For the parameter-update proposals, the acceptance rate was between 0.2 and 0.4 for (P1)
and was higher than 0.6 for (P2), which is an independent proposal.

From the MCMC samples, we estimated the selection
probability $P_{s,j}=P(\hbeta_j \ne 0)$, the 2.5\% and the 97.5\% quantiles
of $\hbeta_j$, and the mean and the standard deviation of the conditional distribution $[\hbeta_j \mid \hbeta_j \ne 0]$
for each $j$. Since theoretical values are not available, we applied
the direct sampling approach to simulate 5,000 independent samples for each dataset under the normal error distribution.
These independent samples were used to estimate the above quantities as the ground truth.  
The MSEs across 10 independent runs of the MLS were calculated, 
and reported in Table~\ref{tab:MSEjoint} are the average MSEs over all $j$ for estimating the above five quantities.
One clearly sees that all the estimates were very accurate.
The MSE of the MLSe was greater than, but on the same order as, that of the MLSn for most estimates, 
which is expected due to the loss of efficiency without assuming a normal error distribution.

\begin{table}[t]
%\centering
\caption{MSE comparison for simulation from the joint sampling distribution \label{tab:MSEjoint}}
%\vspace{0.05in} 
\begin{center}
   \begin{tabular}{ccccccc} 
   \hline
      & Method	 & $P_{s}$	& 2.5\%	& 97.5\% & mean & SD \\
     \hline
       & MLSn & $3.38\times 10^{-4}$ &	$1.82\times 10^{-5}$	& $1.79\times 10^{-5}$ &	$4.36\times 10^{-6}$	& $2.78\times 10^{-6}$ \\
    A & MLSe & 1.29&1.20&1.19&	0.97&	1.38\\
\vspace{0.05in}  & DSn &1.11 &	2.28&	2.45& 2.23 &	2.53\\
       & MLSn & $2.13\times 10^{-4}$&	$2.89\times 10^{-5}$	&$1.74\times 10^{-5}$	& $1.22\times 10^{-5}$	& $8.44\times 10^{-6}$\\
    B & MLSe & 1.20	&	1.07 &	1.10 &	1.08 &	1.30 \\
\vspace{0.05in} & DSn & 1.26 &	1.97 &	1.89 &	2.29 &	2.74 \\
        & MLSn &  $4.14\times 10^{-4}$ &	$1.23\times 10^{-4}$&$1.24\times 10^{-4}$&	$3.20\times 10^{-5}$	& $2.28\times 10^{-5}$  \\
     C  & MLSe & 1.55 &	2.09 &	1.78 &	1.03	&2.46\\
 \vspace{0.05in}& DSn & 0.47 &	1.18 &	1.33 &	1.24 &	1.39 \\   
       & MLSn & $4.34\times 10^{-4}$ & $2.96\times 10^{-4}$	& $2.85\times 10^{-4}$ &$6.37\times 10^{-5}$	& $5.02\times 10^{-5}$\\
    D & MLSe & 2.74 &		3.81 &	3.61 &	1.17 &	5.70 \\
         & DSn & 0.69 &	1.52 &	1.34 &	1.21 &	1.57 \\
     \hline
  \end{tabular}
 \end{center} 
Note:  For the MLSe and the DSn, reported is the ratio of MSE to that of the MLSn. The sweep operator was used in the MLS to calculate determinant ratios for dataset D.
\end{table}

We compared the efficiency of the MLS against the direct sampler (DSn)
under the same amount of running time and under the same normal error distribution. 
The DSn generated around 500 samples
in the same amount of time for 5,500 iterations of the MLSn. 
The ratio of the MSE of the DSn 
to that of the MLSn was calculated for each estimate (Table~\ref{tab:MSEjoint}).
For most estimates, the MLSn seems to be more efficient and may reduce the MSE
by 10\% to 60\%. The improvement was more significant for datasets A and B where the sample size $n=500$.
For the other two datasets, the MLSn showed a higher MSE in estimating selection probabilities
but was more accurate for all other estimates. Furthermore, if the error distribution is more complicated 
such that one cannot simulate samples independently from the distribution, the efficiency of the direct sampler may be even lower.
These results clearly confirm the notion that the MLS can serve as an efficient alternative to the direct sampling method
for simulating from the sampling distribution of a Lasso-type estimator.

Next, we implemented Routine~\ref{rt:MLSfixedA} to sample from the conditional distribution of $\hbmbeta$ given the model
selected according to the $C_p$ criterion, i.e., $[\hbmbeta_{A^*} \mid \calA =A^*]$ with $|A^*|$ given in Table~\ref{tab:data}.
The same parameter setting as that in the previous example was used to run the MLSn and the MLSe.  
We estimated the 2.5\% and the 97.5\% quantiles, the mean, and the standard deviation of $\hbeta_j$
for $j \in A^*$. The model space is composed of
$2^p$ models, and the probability of the model $A^*$, $P(\calA =A^*)$, is practically zero 
for the datasets used here. Therefore, the direct sampling approach is not applicable.
This shows the advantage and flexibility of the Monte Carlo algorithms.
Since we cannot construct ground truth for this example, the accuracy of an estimate is measured by its variance across 10 independent runs of the MLS, averaging over $j\in A^*$ (Table~\ref{tab:MSEcond}).
The variance of every estimate was on the order of $10^{-5}$ or smaller for datasets A, B and C and was
on the order of $10^{-4}$ or smaller for dataset D under both error models. This highlights the stability of
the MLS in approximating sampling distributions across different runs. There were cases in which
the variance of the MLSe was smaller. This does not necessarily suggest that the MLSe provided a more accurate
estimate, as the loss of efficiency without the normal error assumption is likely to result in a higher bias. 

\begin{table}[t]
\centering
\caption{Variance comparison for conditional sampling given active set\label{tab:MSEcond}}
\vspace{0.05in}
   \begin{tabular}{cccccc} 
   \hline
      & Method	 & 2.5\%	& 97.5\% & mean & SD \\
     \hline
   				A    & MLSn & $1.21\times 10^{-5}$& $1.28\times 10^{-5}$	& $2.21\times 10^{-6}$	& $1.03\times 10^{-6}$ \\
 \vspace{0.05in}    & MLSe &0.90 &	1.02 &	0.92 &	1.05 \\

   				B    & MLSn & $1.47\times 10^{-5}$ &	$1.19\times 10^{-5}$	& $3.19\times 10^{-6}$ &	$9.60\times 10^{-7}$ \\
 \vspace{0.05in}    & MLSe &1.22 &	1.15 &	1.02 &	1.23 \\
 
 				C   & MLSn & $7.66\times 10^{-5}$	 & $8.65\times 10^{-5}$	 &$1.59\times 10^{-5}$& $7.08\times 10^{-6}$ \\
 \vspace{0.05in}    & MLSe &1.07 &	0.95 &	1.01 &	1.00 \\
 
 				D & MLSn & $1.67\times 10^{-4}$	&$1.78\times 10^{-4}$ &	$2.55\times 10^{-5}$	 &$1.28\times 10^{-5}$\\
   & MLSe &	0.77 &	0.97 &	0.77 &	0.79 \\		
     \hline
  \end{tabular}
  
\vspace{0.05in}Note: Variance of the MLSe is reported as the ratio to that of the MLSn. 
\end{table}

\section{High-dimensional setting}\label{sec:highdim}

Recent efforts have established theoretical properties of $\ell_1$-penalized linear regression
in high dimension with $p>n$ \citep{Meinshausen06,ZhaoYu06,ZhangHuang08,Bickel09}. 
Under this setting, we assume $\rank(\bfX)=n<p$. Consequently, $\row(\bfX)$ is an $n$-dimensional subspace of $\R^p$ and the Gram matrix $\bfC$ has $n$ positive eigenvalues, denoted by $\Lambda_j>0$, $j=1,\ldots,n$. The associated orthonormal eigenvectors $\bfv_j \in \R^p$, $j=1,\ldots,n$, form a basis for $\row(\bfX)$. 
Choose orthonormal vectors $\bfv_{p+1},\ldots,\bfv_p$ to form a basis for the null space of $\bfX$, $\nul(\bfX)$, and let $\bfV=(\bfv_1 | \ldots |\bfv_p)$. Then $R=\{1,\ldots,n\}$ and $N=\{n+1,\ldots,p\}$ index the columns of $\bfV$ that form respective bases for $\row(\bfX)$ and $\nul(\bfX)$.

\begin{assumption}\label{as:X}
Every $n$ columns of $\bfX$ are linearly independent and every $(p-n)$ rows of $\bfV_N$ are linearly independent. 
\end{assumption}
The first part of this assumption is sufficient for the columns of $\bfX$ being in general position, which guarantees that $(\hbmbeta,\bfS)$ is unique for any $\bfY$ and $\lambda>0$ (Lemma~\ref{lm:prelim}). The second part will ease our derivation of the joint density of the augmented estimator. Note that Assumption~\ref{as:X} holds with probability one if the entries of $\bfX$ are drawn from a continuous distribution on $\R^{n\times p}$.

\subsection{The bijection}

Although $\bfU=\frac{1}{n}\bfX^{\trans}\bmeps$ is a $p$-vector, by definition it always lies in $\row(\bfX)$. 
Therefore, $\bfV_N^{\trans} \bfU=\bfzr$ and the $n$-vector $\bfR=\bfV_R^{\trans}\bfU$ 
gives the coordinates of $\bfU$ with respect to the basis $\bfV_R$.
If $\bmeps$ follows a continuous distribution on $\R^n$, then $\bfR$ has
a proper density with respect to $\xi_n$. For example, if $\bmeps\sim \dnorm_n(\bfzr,\sigma^2 \bfI_n)$,
then $\bfR\sim\dnorm_n(\bfzr,\frac{\sigma^2}{n} \bmLmd)$ with $\bmLmd=\diag(\Lambda_1,\ldots,\Lambda_n)$. Now $\bfR$ plays the same role
as $\bfU$ does in the low-dimensional case. We will use the known distribution of $\bfR$
to derive the distribution of the augmented estimator $(\hbmbeta_\calA,\bfS_\calI,\calA)$. 

However, a technical difficulty is that when $p>n$, the map $\bfH$ defined in \eqref{eq:lassojoint}
is not a mapping from $\Omega$
to $\row(\bfX)$ as $\bfH(\bfb_A,\bfs_I,A;\bmbeta)\in\R^p$ is not necessarily in $\row(\bfX)$
for every $(\bfb_A,\bfs_I,A)\in \Omega$. We thus need to remove those ``illegal"
points in $\Omega$ so that the image of $\bfH$ always lies in the row space of $\bfX$. This is achieved by
imposing the constraint that 
\begin{equation}\label{eq:constraintinU}
\bfV_N^{\trans} \bfH(\hbmbeta_{\calA},\bfS_\calI,\calA;\bmbeta)=\bfV_N^{\trans} \bfU=\bfzr,
\end{equation} 
i.e., the image of $\bfH$ must be orthogonal to $\nul(\bfX)$. It is more convenient to use the equivalent definition of $\bfH$ 
in \eqref{eq:lassoKKTinU}, i.e.,
\begin{equation}\label{eq:equivH}
\bfH(\hbmbeta_{\calA},\bfS_\calI,\calA;\bmbeta)=\bfC \hbmbeta + \lambda \bfW \bfS - \bfC \bmbeta.
\end{equation}
Because $\bfC(\hbmbeta-\bmbeta)=\frac{1}{n}\bfX^{\trans}\bfX(\hbmbeta-\bmbeta)\in\row(\bfX)$, 
constraint \eqref{eq:constraintinU} is equivalent to
\begin{equation}\label{eq:constraintonS}
 \bfV_N^{\trans} \bfW \bfS = \bfV_{\calA N}^{\trans} \bfW_{\calA\calA} \bfS_\calA 
 + \bfV_{\calI N}^{\trans} \bfW_{\calI\calI} \bfS_\calI = \bfzr.
\end{equation}
In words, the constraint is that
the vector $\bfW\bfS$ must lie in $\row(\bfX)$. Therefore, we have a more 
restricted space for the augmented estimator $(\hbmbeta_{\calA},\bfS_\calI,\calA)$ in the high-dimensional case,
\begin{equation}\label{eq:spacehigh}
\Omega_r=\{(\bfb_A,\bfs_I,A)\in \Omega: \bfV_{AN}^{\trans} \bfW_{AA} \sgn(\bfb_A) 
 + \bfV_{IN}^{\trans} \bfW_{II} \bfs_I =\bfzr \}.
\end{equation}
Restricted to this space, $\bfH$ is a bijection.
\begin{lemma}\label{lm:bijectioninhighdim}
If $p>n$ and Assumption~\ref{as:X} holds, then for any $\bmbeta$ and $\lambda>0$, the restriction of the mapping $\bfH$ \eqref{eq:lassojoint} to $\Omega_r$, denoted by $\bfH\mid_{\Omega_r}$, 
is a bijection that maps $\Omega_r$ onto $\row(\bfX)$.
\end{lemma}
\begin{proof}
For any $\bfU \in \row(\bfX)$, there is a unique $(\hbmbeta_\calA,\bfS_\calI,\calA)$ 
such that $\bfU=\bfH(\hbmbeta_{\calA},\bfS_\calI,\calA;\bmbeta)\in\row(\bfX)$ by Lemma~\ref{lm:prelim}. 
Thus, $(\hbmbeta_\calA,\bfS_\calI,\calA)$ satisfies the constraint
\eqref{eq:constraintinU} and lies in $\Omega_r$.
For any $(\hbmbeta_\calA,\bfS_\calI,\calA) \in \Omega_r$, 
$\bfV_N^{\trans} \bfH(\hbmbeta_{\calA},\bfS_\calI,\calA;\bmbeta)=\bfzr$
and $\bfH$ maps it into $\row(\bfX)$.
\end{proof}
\begin{remark}\label{rm:dimofH}
Fixing $\calA=A$, \eqref{eq:constraintonS} specifies $|N|=p-n$ constraints, 
and thus, the continuous components  $(\hbmbeta_A,\bfS_I)\in\R^p$
lie in an $n$-dimensional subspace of $\Omega_A$ \eqref{eq:subspace}. %On the other hand, $\row(\bfX)$ is an $\R^n$.
The bijection $\bfH\mid_{\Omega_r}$ 
maps a finite number of $n$-dimensional subspaces onto $\row(\bfX)$ which is an $\R^n$.
\end{remark}
Now we represent the bijection $\bfH\mid_{\Omega_r}$ in terms of its coordinates with respect to $\bfV_R$ 
and equate it with $\bfR=\bfV_R^{\trans}\bfU$:
\begin{equation}\label{eq:HinR}
\bfR=\bfV_R^{\trans} \bfH(\hbmbeta_\calA,\bfS_\calI,\calA;\bmbeta) \defi \bfH_r(\hbmbeta_\calA,\bfS_\calI,\calA;\bmbeta).
\end{equation}

\subsection{Joint sampling distribution}

The distribution for $(\hbmbeta_{\calA},\bfS_{\calI},\calA) \in \Omega_r$
is completely given by the distribution of $\bfR$ via the bijective map $\bfH_r:\Omega_r \to \R^n$. 
The only task left is to determine the Jacobian of $\bfH_r$, taking into account the constraint \eqref{eq:constraintonS}.
Left Multiplying by $\bfV_R^{\trans}$ both sides of Equation~\eqref{eq:lassojoint_a},
with the simple facts that $\bfV_R^{\trans}\bfW_{\calA}=\bfV_{\calA R}^{\trans}\bfW_{\calA\calA}$ and 
$\bfV_R^{\trans}\bfW_{\calI}=\bfV_{\calI R}^{\trans} \bfW_{\calI\calI}$, gives
\begin{equation}\label{eq:Hinhigh}
\bfR =  \bfV_R^{\trans} \bfC_{\calA} \hbmbeta_{\calA} + \lambda \bfV_{\calA R}^{\trans}\bfW_{\calA\calA} \bfS_{\calA} 
+\lambda \bfV_{\calI R}^{\trans} \bfW_{\calI\calI} \bfS_{\calI} -\bfV_R^{\trans}\bfC\bmbeta.
\end{equation}
For any fixed value of $\calA$, differentiating $\bfR$ and both sides of the constraint \eqref{eq:constraintonS} 
with respect to $(\hbmbeta_{\calA},\bfS_\calI)$ give, respectively,
\begin{eqnarray}
& d\bfR =  \bfV_R^{\trans} \bfC_{\calA} d \hbmbeta_{\calA}  
+\lambda \bfV_{\calI R}^{\trans} \bfW_{\calI\calI} d \bfS_{\calI}, \label{eq:dR}\\
&\bfV_{\calI N}^{\trans} \bfW_{\calI\calI} d \bfS_{\calI}=\bfzr. \label{eq:dS}
\end{eqnarray}
Therefore, the constraint implies that $d \bfS_{\calI}$ is in $\nul (\bfV_{\calI N}^{\trans} \bfW_{\calI\calI})$.
\begin{lemma}\label{lm:rank}
If $p>n$ and Assumption~\ref{as:X} is satisfied, then the dimension of $\nul (\bfV_{\calI N}^{\trans} \bfW_{\calI\calI})$ is $n-|\calA|\geq 0$.
\end{lemma}
\begin{proof}
Under the assumption, the minimizer $\hbmbeta$ of \eqref{eq:lassoloss} is unique and always has an active set with size 
$|\calA|\leq \min\{n,p\}=n$. See Lemma 14 in \cite{Tibshirani13}. 
If Assumption~\ref{as:X} is satisfied, any $|N|$ rows of $\bfV_N$ are linearly independent. Since $|\calI|=p-|\calA|\geq p-n=|N|$,
the rank of the $|N| \times |\calI|$ matrix, $\bfV_{\calI N}^{\trans} \bfW_{\calI\calI}$, is $p-n$.
Then it follows that the dimension of $\nul (\bfV_{\calI N}^{\trans} \bfW_{\calI\calI})$ is $|\calI|-(p-n)=n-|\calA|\geq 0$.
\end{proof}
Let $\bfB(\calI)\in\R^{|\calI|\times(n-|\calA|)}$ be an orthonormal basis for $\nul (\bfV_{\calI N}^{\trans} \bfW_{\calI\calI})$. 
Let $d\tdbfS$ be the coordinates of $d\bfS_{\calI}$ with respect to the basis $\bfB(\calI)$, i.e.,
$d\bfS_{\calI}=\bfB(\calI)d\tilde{\bfS}$. Note that $d\tdbfS \in \R^{n-|\calA|}$ according to 
the above lemma. Then \eqref{eq:dR} becomes
\begin{eqnarray}
d\bfR & =  & \bfV_R^{\trans} \bfC_{\calA} d \hbmbeta_{\calA}  
+\lambda \bfV_{\calI R}^{\trans} \bfW_{\calI\calI} \bfB(\calI) d\tdbfS  \nonumber \\
& = & \bfT(\calA) \left(\begin{array}{c} d \hbmbeta_{\calA} \\ d\tdbfS \end{array}\right), \label{eq:Jacob}
\end{eqnarray} 
where $\bfT(\calA)=(\bfV_R^{\trans} \bfC_{\calA}\mid \lambda \bfV_{\calI R}^{\trans} \bfW_{\calI\calI} \bfB(\calI))$,
an $n\times n$ matrix, is the Jacobian of the map $\bfH_r$. 
The dimension of $(d \hbmbeta_{\calA},d\tdbfS)$ is always $n$ for any $\calA$. 
This confirms the notion in Remark~\ref{rm:dimofH} that the continuous components 
$(\hbmbeta_{\calA},\bfS_{\calI})$ lie in an $n$-dimensional subspace when $\calA$ is fixed.

Now we are ready to
derive the density for $(\hbmbeta_\calA,\bfS_{\calI},\calA)$ in high dimension. For $(\bfb_A,\bfs_I,A)\in\Omega_r$, 
$d\bfs_I=\bfB(I)d\tdbfs$ for some $d\tdbfs \in \R^{n-|A|}$ and $\xi_{n-|A|}(d\tdbfs)$
gives the infinitesimal volume at $\bfs_I$ subject to constraint \eqref{eq:spacehigh}.

\begin{theorem}\label{thm:jointhigh}
Assume that $p>n$ and Assumption~\ref{as:X} holds. Let $f_{\bfR}$ be the probability density
of $\bfR$ with respect to $\xi_n$. For $(\bfb_A,\bfs_I,A)\in\Omega_r$, 
the joint distribution of $(\hbmbeta_\calA,\bfS_\calI, \calA)$ is given by
\begin{eqnarray}\label{eq:jointdensityhigh}
 P(\hbmbeta_A\in d\bfb_A,\bfS_I \in d\bfs_I,\calA=A) 
&=&f_{\bfR}(\bfH_r(\bfb_A,\bfs_I,A;\bmbeta))|\det\bfT(A)| \xi_n(d\bfb_A d\tdbfs) \nonumber \\
& \defi & \pi_r(\bfb_A,\bfs_I,A) \xi_n(d\bfb_A d\tdbfs).
\end{eqnarray}
Particularly, if $\bmeps\sim \dnorm_n(\bfzr,\sigma^2 \bfI_n)$, then 
\begin{equation}\label{eq:densitynormalhigh}
\pi_r(\bfb_A,\bfs_I,A)=\phi_n\left(\bfH_r(\bfb_A,\bfs_I,A;\bmbeta);\bfzr,n^{-1}{\sigma^2} \bmLmd\right)|\det\bfT(A)|.
\end{equation}
\end{theorem}
\begin{proof}
The proof is analogous to that of Theorem~\ref{thm:joint}.
Let $\bfr=\bfH_r(\bfb_A,\bfs_I,A;\bmbeta) \in \R^n$. For any fixed $A$, 
\begin{equation*}
d\bfr = \bfT(A) \left(\begin{array}{c} d \bfb_A \\ d\tdbfs \end{array}\right)
\end{equation*} 
from \eqref{eq:Jacob} and thus $\xi_n(d\bfr)=|\det\bfT(A)| \xi_n(d\bfb_A d\tdbfs)$.
With the bijective nature of $\bfH_r$ and its restriction to any $A$, a change of variable gives
\begin{align*}
& P(\hbmbeta_A\in d\bfb_A,\bfS_I \in d\bfs_I,\calA=A) =  P(\bfR \in d\bfr) \\
&\quad\quad = f_{\bfR}(\bfH_r(\bfb_A,\bfs_I,A;\bmbeta))|\det\bfT(A)| \xi_n(d\bfb_A d\tdbfs).
\end{align*}
If $\bmeps\sim \dnorm_n(\bfzr,\sigma^2 \bfI_n)$
then $\bfR\sim\dnorm_n(\bfzr,\frac{\sigma^2}{n} \bmLmd)$, which leads to \eqref{eq:densitynormalhigh} immediately.
\end{proof}
\begin{remark}\label{rm:welldefine}
The density $\pi_r(\bfb_A,\bfs_I,A)$ does not depend on which orthonormal basis we choose 
for $\nul(\bfV_{I N}^{\trans} \bfW_{II})$. If $\bfB'(I)$ is another orthonormal basis, then $\bfB'(I)=\bfB(I)\bfO$,
where $\bfO$ is an $(n-|A|) \times (n-|A|)$ orthogonal matrix and $|\det\bfO|=1$. 
Correspondingly, 
\begin{equation*}
\bfT'(A)=(\bfV_R^{\trans} \bfC_{A}\mid \lambda \bfV_{I R}^{\trans} \bfW_{II} \bfB'(I))=\bfT(A)\diag(\bfI_{|A|},\bfO),
\end{equation*}
and thus $|\det \bfT'(A)|=|\det \bfT(A)|$.
\end{remark}
\begin{remark}\label{rm:unified}
One may unify Theorems~\ref{thm:joint} and \ref{thm:jointhigh} with the use of cumbersome notations, but the idea is simple.
Note that $\bfT(A)$ and $\bfD(A)$ in \eqref{eq:lassojoint} are connected by
\begin{equation*}
\bfT(A)=\bfV_R^{\trans} (\bfC_{A}\mid \lambda \bfW_{I}\bfB(I))=\bfV_R^{\trans}\bfD(A) \diag(\bfI_{|A|},\bfB(I)).
\end{equation*}
If $\rank(\bfX)=p\leq n$, the set $N$ reduces to the empty set and $\bfV_R=\bfV$.  
Hence, the constraint \eqref{eq:constraintonS} no long exists,
the space $\Omega_r$ is the same as $\Omega$,
and $\nul(\bfV_{I N}^{\trans} \bfW_{II})$ is simply $\R^{|I|}$ for any $I=\{1,\ldots,p\}\setminus A$. 
Choosing $\bfB(I)=\bfI_{|I|}$ leads to $d\bfs_I=d\tdbfs$, 
which shows that the probability in \eqref{eq:jointdensityhigh} reduces to
that in \eqref{eq:jointdensity}. In this case, 
$\bfT(A)=\bfV^{\trans}\bfD(A)$,
i.e., a column of $\bfT(A)$ gives the coordinates of the corresponding column of $\bfD(A)$ 
with respect to the basis $\bfV$.
\end{remark}

In principle, one can develop MCMC algorithms to sample from the joint distribution \eqref{eq:jointdensityhigh}.
Development of such an algorithm is a little tedious because of the constraints in the definition of $\Omega_r$ 
and the use of different bases $\bfB(I)$ in different subspaces.
However, the explicit density given in Theorem~\ref{thm:jointhigh} allows us to 
develop very efficient IS algorithms for approximating tail probabilities with respect
to the sampling distribution of a Lasso-type estimator.

\section{P-value calculation by IS}\label{sec:pvalue}

To simplify description, we focus on the high-dimensional setting
with normal errors so that the distribution of interest is $\pi_r$ \eqref{eq:densitynormalhigh}. 
For a fixed $\bfX$, the density $\pi_r$, the bijection $\bfH_r$ \eqref{eq:HinR} and the matrix $\bfT$ \eqref{eq:Jacob} are written
as $\pi_r(\bfb_A,\bfs_I,A;\bmbeta,\sigma^2,\lambda)$, $\bfH_r(\bfb_A,\bfs_I,A; \bmbeta, \lambda)$,
and $\bfT(A; \lambda)$, respectively, to explicitly indicate their dependency on different parameters. 
Suppose we are given a Lasso-type estimate $\hbmbeta^{*}$ for an
observed dataset with a tuning parameter $\lambda^{*}$. Under the null model 
$\calH: \bmbeta=\bmbeta_0, \sigma^2=\sigma^2_0$, we want to
calculate the p-value of some test statistic $T(\hbmbeta) \in \R$ constructed from
the Lasso-type estimator $\hbmbeta$ for $\lambda=\lambda^*$. Precisely, the desired p-value is
\begin{equation}\label{eq:pvalue}
q^*=P(|T(\hbmbeta)|\geq T^*; \calH, \lambda^*)=\int_{\Omega_r^*} \pi_r(\bfb_A,\bfs_I,A;\bmbeta_0,\sigma_0^2,\lambda^*) 
\xi_n(d\bfb_A d\tdbfs),
\end{equation}
where $T^*=|T(\hbmbeta^*)|$ and $\Omega_r^*=\{(\bfb_A,\bfs_I,A)\in \Omega_r: |T(\bfb)|\geq T^*\}$.
Even if we can directly sample from $\pi_r(\bullet;\bmbeta_0,\sigma_0^2,\lambda^*)$,
estimating $q^*$ will be extremely difficult when it is very small. With the closed-form density $\pi_r$,
we can use IS to solve this challenging problem.

\subsection{Importance sampling}

Our target distribution is $\pi_r(\bullet;\bmbeta_0,\sigma_0^2,\lambda^*)$ and
we propose to use $\pi_r(\bullet;\bmbeta_0,(\sigma^2)^{\dag},\lambda^{\dag})$ as a trial distribution
to estimate expectations with respect to the target distribution via IS.
First, note that the trial and the target distributions have the same support as the constraint in \eqref{eq:spacehigh} 
that defines the space $\Omega_r$ only depends on $\bfX$. Thus, a sample from 
the trial distribution $\pi_r(\bullet;\bmbeta_0,(\sigma^2)^{\dag},\lambda^{\dag})$ also satisfies the
constraint for the target distribution. Second, one can easily simulate from the trail distribution 
by the direct sampler (Routine~\ref{rt:directsample}).
Third, the importance weight for a sample $(\bfb_A,\bfs_I,A)$ from the trial distribution can be calculated efficiently.
Let $(r_1(\bmbeta,\lambda),\cdots,r_n(\bmbeta,\lambda))=\bfH_r(\bfb_A,\bfs_I,A;\bmbeta,\lambda)\in \R^n$ and note that $\det\bfT(A;\lambda)=\lambda^{n-|A|} \det \bfT(A;1)$.
Using the fact that $\bmLmd=\diag(\Lambda_1,\ldots,\Lambda_n)$, the importance weight
\begin{eqnarray}\label{eq:ISweight}
& &\frac{\phi_n\left(\bfH_r(\bfb_A,\bfs_I,A;\bmbeta_0, \lambda^*);\bfzr,n^{-1}{\sigma_0^2} \bmLmd\right)|\det\bfT(A;\lambda^*)|}
{\phi_n\left(\bfH_r(\bfb_A,\bfs_I,A;\bmbeta_0, \lambda^{\dag});\bfzr,n^{-1}{(\sigma^2)^{\dag}} \bmLmd\right)|\det\bfT(A;\lambda^{\dag})|} \nonumber\\
&\propto & 
\exp\left[\frac{n}{2(\sigma^2)^{\dag}}\sum_{i=1}^n\frac{r_i^2(\bmbeta_0,\lambda^{\dag})}{\Lambda_i}
-\frac{n}{2 \sigma_0^2}\sum_{i=1}^n\frac{r_i^2(\bmbeta_0,\lambda^{*})}{\Lambda_i} \right]
\left(\frac{\lambda^*}{\lambda^{\dag}}\right)^{n-|A|}  \nonumber \\
&\defi & w(\bfb_A,\bfs_I,A;\sigma_0^2,\lambda^*).
\end{eqnarray}
Essentially, for each sample, we only need to compute the image of the map $\bfH_r$ and two sums of squares.
\begin{routine}\label{rt:IS}
Draw $(\bfb_A,\bfs_I,A)^{(t)}$, $t=1,\ldots,L$, from the trial distribution 
$\pi_r(\bullet;\bmbeta_0,(\sigma^2)^{\dag},\lambda^{\dag})$ by Routine~\ref{rt:directsample}. 
Then the IS estimate for the p-value $q^*$ is given by
\begin{equation}\label{eq:ISest}
\hat{q}^{\text{(IS)}}=\frac{\sum_{t=1}^L w((\bfb_A,\bfs_I,A)^{(t)};\sigma_0^2,\lambda^*)\I(|T(\bfb^{(t)})|\geq T^*)}
{\sum_{t=1}^L w((\bfb_A,\bfs_I,A)^{(t)};\sigma_0^2,\lambda^*)}.
\end{equation}
\end{routine}

The key is to choose the parameters $(\sigma^2)^{\dag}$ and $\lambda^{\dag}$ in the trial distribution so that we have
a substantial fraction of samples for which $|T(\bfb^{(t)})|\geq T^*$. Next we discuss some guidance on
tuning these parameters.

\subsection{Tuning trial distributions}

We illustrate our procedure for tuning the trial distribution assuming $\bmbeta_0=\bfzr$, i.e., the null hypothesis is 
$\calH_0: \bmbeta=\bfzr, \sigma^2=\sigma^2_0$. In this case, the problem is difficult when 
$P(\hbmbeta=\bfzr)$ is close to one under the target distribution $\pi_r(\bullet; \bfzr, \sigma_0^2, \lambda^*)$.
In other words, $\lambda^*$ is too big to obtain any nonzero estimate of the coefficients and consequently
the p-value $P(|T(\hbmbeta)|\geq T^*; \calH_0, \lambda^*)$ becomes a tail probability.
Thus, one may want to choose the trial distribution
$\pi_r(\bullet; \bfzr, (\sigma^2)^{\dag}, \lambda^{\dag})$ under which there is a higher probability for 
nonzero $\hbmbeta$. In general, we achieve this by choosing $(\sigma^2)^{\dag}=M^{\dag} \sigma_0^2$ ($M^{\dag}>1$) 
and then tuning $\lambda^{\dag}$ accordingly. When we increase $\sigma^2$, the variance of $\bfU$
increases and thus $\bfU$ will have a wider spread in $\row(\bfX)$. This will increase
the variance of the augmented estimator $(\hbmbeta_\calA,\bfS_\calI,\calA)$.
As illustrated in Figure~\ref{fig:bijection}, a larger variance in $\bfU$ will lead to a more uniform distribution over
different subspaces $\{\Omega_A\}$.

The following simple procedure is used to determine $\lambda^{\dag}$ given $(\sigma^2)^{\dag}$, which
works very well based on our empirical study. 
\begin{routine}\label{rt:tuning}
Draw $\bfY^{(t)}$ from $\dnorm_n(\bfzr,(\sigma^2)^{\dag} \bfI_n)$
and calculate 
$\lambda^{(t)}=n^{-1}\|\bfW^{-1}\bfX^{\trans}\bfY^{(t)}\|_{\infty}$
for $t=1,\ldots,L_{\text{pilot}}$. Then set $\lambda^{\dag}$ to the first quartile of
$\{\lambda^{(t)}: t=1,\ldots,L_{\text{pilot}}\}$. 
\end{routine}
Setting $\hbmbeta=\bfzr$ in \eqref{eq:lassograd}, we have
\begin{equation*}
n^{-1}\|\bfW^{-1}\bfX^{\trans}\bfY\|_{\infty}=\lambda \|\bfS\|_{\infty}\leq \lambda,
\end{equation*}
which shows that the $\lambda^{(t)}$ calculated in Routine~\ref{rt:tuning} is the minimum value of $\lambda$
with which $\hbmbeta=\bfzr$ for $\bfY^{(t)}$. Therefore, 
under the trial distribution $\pi_r(\bullet;\bfzr,(\sigma^2)^{\dag},\lambda^{\dag})$, 
$P(\hbmbeta=\bfzr)$ is around 25\% and there is a 75\% of chance
for $\hbmbeta$ to have some nonzero components. This often results in a good balance between the
dominating region of the target distribution ($\hbmbeta=\bfzr$) and the region of interest $\Omega_r^*$
for p-value calculation \eqref{eq:pvalue}.
For all numerical examples in this article, we choose $M^{\dag}=5$ and $L_{\text{pilot}}=100$.

\subsection{Multiple tests}

Consider multiple linear models with the same set of predictors,
\begin{equation}\label{eq:multilinear}
\bfY_k=\bfX \bmbeta_k + \bmeps_k, \;\; k=1,\ldots,m,
\end{equation} 
where $\bfY_k\in \R^n$, $\bmbeta_k \in \R^p$, and $\bmeps_k \sim \dnorm_n(\bfzr, \sigma_k^2 \bfI_n)$.
After proper rescaling of $\bfY_k$ and $\bmbeta_k$, we may assume that all $\sigma_k^2$ are identical,
i.e., $\sigma_k^2=\sigma^2$.
Suppose we are interested in testing against $m$ null hypotheses $\calH_k:$ 
$\bmbeta_k=\bfzr$ and $\sigma^2=\sigma_0^2$, given Lasso-type estimates $\hbmbeta_k^*$ with $\lambda=\lambda_k^*$
for $k=1,\ldots,m$. There are $m$ p-values to calculate,
\begin{equation}\label{eq:multipvalue}
q^*_k=P(|T(\hbmbeta)|\geq T_k^*; \calH_k, \lambda_k^*),
\end{equation}
where $T_k^*=|T(\hbmbeta_k^*)|$ 
for $k=1,\ldots,m$. This problem occurs in various genomics applications. To give an example, $\bfY_k$
may be the expression level of gene $k$ and $\bfX$ the expression levels of
$p$ transcription factors across $n$ individuals. The transcription factors may
potentially regulate the expression of a gene
through the linear model \eqref{eq:multilinear}. Rejection of $\calH_k$ indicates that gene $k$ 
is regulated by at least one of the $p$ transcription factors.

To estimate all $q^*_k$, we only need to draw $(\bfb_A,\bfs_I,A)^{(t)}$, $t=1,\ldots,L$, from one trial distribution 
$\pi_r(\bullet;\bfzr,(\sigma^2)^{\dag},\lambda^{\dag})$, in which $(\sigma^2)^{\dag}=M^{\dag} \sigma_0^2$ and
$\lambda^{\dag}$ is obtained by applying the same tuning procedure (Routine~\ref{rt:tuning}) once.
Then we calculate the importance weights by \eqref{eq:ISweight} for all target distributions, 
$\{w((\bfb_A,\bfs_I,A)^{(t)};\sigma_0^2,\lambda_k^*)\}_{1\leq t \leq L}$, $k=1,\ldots,m$, 
and construct estimates for all $q^*_k$ by \eqref{eq:ISest}.

\begin{remark}\label{rm:pathpvalue}
Alternatively, one may apply the Lars algorithm in the direct sampler to draw from the sampling
distribution of $\hbmbeta$ given $\bmbeta=\bfzr$ and $\sigma^2=\sigma_0^2$ for all $\lambda_k^*$,
as the Lars algorithm provides the whole solution path. The computing time of both methods
is dominated by drawing samples and thus is comparable. 
However, when $q^*_k$ is small, the IS method will be
orders of magnitude more efficient than direct sampling, and when $q^*_k$ is not too small, the accuracy
of the two methods is on the same order. We will see this in the numerical examples. In addition, we do not have 
to use the Lars algorithm to draw from the trial distribution since there is no need to compute the solution path
for importance sampling. One thus has the freedom to choose other algorithms, 
such as coordinate descent \citep{Friedman07,WuLange08},
which may be more efficient when both $n$ and $p$ are large.
\end{remark}

\subsection{Numerical examples}

We first simulated two datasets to demonstrate the effectiveness in p-value calculation 
by the IS method for individual tests. 
Each row of $\bfX$ was generated from 
$\dnorm_p(\bfzr,\bmSigma_{\bfX})$, where the diagonal and the off-diagonal elements of $\bmSigma_{\bfX}$
are 1 and $0.05$, respectively. Given the predictors $\bfX$, the response vector $\bfY$ was drawn
from $\dnorm_n(\bfX \bmbeta_0, \sigma_0^2 \bfI_n)$. We set all weights $w_j=1$. 
Table~\ref{tab:dataindividual} reports the
values of $n$, $p$, $\sigma_0^2$, and $\bmbeta_0$ for the two datasets.
We applied the Lars algorithm on the two datasets and 
chose $\lambda^*$ as the first $\lambda$ along the solution path such that
the Lasso estimate $\hbmbeta^*$ gave the correct number of active coefficients (Table~\ref{tab:dataindividual}). 
It turned out that $\hbmbeta^*$ only included one true active coefficient for both datasets.
Let $A^*=\supp(\hbmbeta^*)$ be the active set of $\hbmbeta^*$. 
We designed the following test statistics, $T_1=\|\hbmbeta\|_1$, $T_2=\|\hbmbeta\|_{\infty}$,
and $\tilde{T}_j=|\hbeta_j|$ for $j \in A^*$, and aimed to calculate p-values under the null hypothesis 
$\calH_0: \bmbeta=\bfzr$ and $\sigma^2=\sigma_0^2$.

\begin{table}[ht]
\centering
\caption{Simulated datasets for individual tests \label{tab:dataindividual}}
\vspace{0.05in}
   \begin{tabular}{c|ccclcc} 
   \hline
     Dataset	& $n$ & $p$ & $\sigma_0^2$ & $\bmbeta_0$ & $\lambda^*$  & $\lambda^{\dag}$\\
     \hline
     E & 5 & 10 & 1/4 & $(2,-2, 0, \ldots, 0)$ & 1.65 & 0.60\\
     F & 10 & 20 & 1/4 & $(1,1,-1,-1, 0, \ldots,0)$ & 0.315 & 0.57\\
    \hline
  \end{tabular}
\end{table}

We chose $(\sigma^2)^{\dag}=5 \sigma_0^2$ and used Routine~\ref{rt:tuning} to choose $\lambda^{\dag}$
for the trial distributions. The values of $\lambda^{\dag}$ for the two datasets are given in Table~\ref{tab:dataindividual}.
When $(\sigma^2)^{\dag}$ is sufficiently large for a dataset, the $\lambda^{\dag}$ tuned by Routine~\ref{rt:tuning}
can be greater then $\lambda^*$ (dataset F).
The IS method (Routine~\ref{rt:IS}) was applied with $L=5,000$ to estimate p-values for all the above tests.
This estimation procedure was repeated 10 times independently to obtain the standard deviation of an
estimated p-value. We quantify the efficiency of an estimated p-value, $\hat{q}$, by its coefficient of variation
$\text{cv}(\hat{q})=\SD(\hat{q})/\E(\hat{q})$,
where the standard deviation and the mean are calculated across multiple runs. Table~\ref{tab:indpval}
summarizes the results, where $A^*=\{2,3\}$ for dataset E and $A^*=\{4,9,15,17\}$ for dataset F.
One sees that the IS estimates were very accurate: Even for a tail probability as small as $10^{-21}$,
the coefficient of variation was less than or around 2. To benchmark the performance, we approximated
the coefficient of variation of the estimate $\hat{q}^{\text{(DS)}}$ constructed by direct sampling from the target distribution, 
$\text{cv}(\hat{q}^{\text{(DS)}})=\sqrt{(1-\bar{q})/(L\bar{q}})$ with $\bar{q}=\E(\hat{q}^{\text{(IS)}})$.
As reported in the table, for estimating an extremely small p-value, $\text{cv}(\hat{q}^{\text{(DS)}})$ can
be orders of magnitude greater than that of an IS estimate, and for a moderate p-value (around $10^{-2}$), 
the two methods showed comparable performance.

\begin{table}[t]
\centering
\caption{Estimation of p-values for datasets E and F\label{tab:indpval}}
\vspace{0.05in}
   \begin{tabular}{cc|cccc} 
   \hline
  &  & $\E(\hat{q}^{\text{(IS)}})$ & $\SD(\hat{q}^{\text{(IS)}})$ & $\text{cv}(\hat{q}^{\text{(IS)}})$ & $\text{cv}(\hat{q}^{\text{(DS)}})$\\
     \hline
 & $T_1$ & $3.7 \times 10^{-19}$ & $8.7\times 10^{-19}$ & 2.37 & $2.32\times 10^7$ \\   
E & $T_2$ & $5.7 \times 10^{-15}$ & $6.2\times 10^{-19}$ & 1.09 & $1.88\times 10^5$ \\   
& $\tilde{T}_2$ & $6.3 \times 10^{-21}$ & $1.1\times 10^{-20}$ & 1.81 & $1.79\times 10^8$ \\   
& $\tilde{T}_3$ & $8.4 \times 10^{-15}$ & $9.8\times 10^{-15}$ & 1.16 & $1.54\times 10^5$ \\   
   \hline
& $T_1$ & $1.5 \times 10^{-6}$ & $4.5\times 10^{-7}$ & 0.30 & 11.5 \\   
& $T_2$ & $1.2 \times 10^{-3}$ & $1.2\times 10^{-4}$ & 0.11 & 0.41 \\   
F & $\tilde{T}_4$ & $5.7 \times 10^{-5}$ & $2.2\times 10^{-5}$ & 0.38 & 1.87 \\   
& $\tilde{T}_9$ & $1.1 \times 10^{-2}$ & $1.3\times 10^{-3}$ & 0.12 & 0.14 \\   
& $\tilde{T}_{15}$ & $2.5 \times 10^{-2}$ & $1.9\times 10^{-3}$ & 0.08 & 0.09 \\   
& $\tilde{T}_{17}$ & $4.8 \times 10^{-5}$ & $1.5\times 10^{-5}$ & 0.31 & 2.04\\  
\hline 
  \end{tabular}
\end{table}

Next, we simulated $m=50$ datasets to test our p-value calculation for the multiple testing problem.
We used the design matrix $\bfX$ in dataset F and $\sigma_0^2=1/4$. The response vector $\bfY_k$
was drawn from $\dnorm_n(\bfX \bmbeta_k, \sigma_0^2 \bfI_n)$, where the true coefficient vector $\bmbeta_k$
is given in Table~\ref{tab:datamultiple} for $k=1,\ldots,50$. 
For 10 datasets, $\bmbeta_k=\bfzr$ and the null hypothesis is true. For 20 datasets, there are
two large coefficients, which represents the case that the true model is sparse. The other 20 datasets
mimic the scenario in which the true model has many relatively small coefficients.
We chose $\lambda_k^*$ as the first $\lambda$
that gave two active coefficients along the solution path 
and used $T(\hbmbeta)=\|\hbmbeta\|_1$ as the test statistic. Summaries of $\lambda_k^*$ and $T_k^*=\|\hbmbeta_k^*\|_1$
for the 50 datasets are provided in Table~\ref{tab:datamultiple} as well, from which we see that these datasets
cover a wide range of $\lambda_k^*$ and $T_k^*$.
We chose $(\sigma^2)^{\dag}=5\sigma_0^2$. As seen from Routine~\ref{rt:tuning}, for identical $\bfX$
and $(\sigma^2)^{\dag}$, the tuning procedure is the same. Therefore, we simply set $\lambda^\dag=0.57$,
the value we used for dataset F (Table~\ref{tab:dataindividual}).

\begin{table}[ht]
\centering
\caption{Simulated datasets for multiple tests \label{tab:datamultiple}}
\vspace{0.05in}
   \begin{tabular}{c|lccc} 
   \hline
     Dataset	& $\bmbeta_k$ & range of $\lambda_k^*$  & range of $T_k^*$\\
     \hline
     1-10 & $(0, \ldots, 0)$ & (0.16, 0.34) & (0.08, 0.23)\\
     11-30 & $(2,-2, 0, \ldots, 0)$ & (0.88, 1.31) & (0.27, 0.84)\\
     31-50 & $(1/4,\ldots,1/4)$ & (0.70, 1.13) & (0.04, 0.51)\\
    \hline
  \end{tabular}
\end{table}

We simulated $L=5,000$ samples from the trial distribution and estimated the p-values for all the 50 datasets.
This procedure was repeated 10 times independently. The average over 10 runs of the estimated p-value,
$\E(\hat{q}_k^{\text{(IS)}})$, is shown in Figure~\ref{fig:multipval}(a) for $k=1,\ldots,50$. 
As expected, most of the p-values for the first 10 datasets
were not significant, while those for the other 40 datasets ranged from $10^{-4}$ to $10^{-30}$,
which confirms that $T(\hbmbeta)=\|\hbmbeta\|_1$ is a reasonable test statistic.
Again, we see that even for p-values on the order of $10^{-30}$,
the coefficient of variation of an IS estimate was at most around 3 (Figure~\ref{fig:multipval}(b)). 
This provides huge gain in accuracy compared to direct sampling. Figure~\ref{fig:multipval}(c) plots
$\log_{10}[\text{cv}(\hat{q}_k^{\text{(DS)}})/\text{cv}(\hat{q}_k^{\text{(IS)}})]$ for the 50 datasets, where  
$\text{cv}(\hat{q}_k^{\text{(DS)}})$ was approximated in the same way as in the previous example.
It is comforting to see that while the IS estimates $\hat{q}_k^{\text{(IS)}}$ 
showed huge improvement over the DS estimates in estimating a tail probability, they were only slightly
worse than the DS estimates for an insignificant p-value. For the first 10 datasets, the coefficient of variation
of $\hat{q}_k^{\text{(IS)}}$ was at most 7.9 times that of $\hat{q}_k^{\text{(DS)}}$. For majority of the other 40 datasets,
the ratio of $\text{cv}(\hat{q}_k^{\text{(DS)}})$ over $\text{cv}(\hat{q}_k^{\text{(IS)}})$ was between
$100$ and $10^{10}$.

\begin{figure}[t]
\centering
   \includegraphics[width=0.55\linewidth,trim=0in 0in 0in 0in,clip]{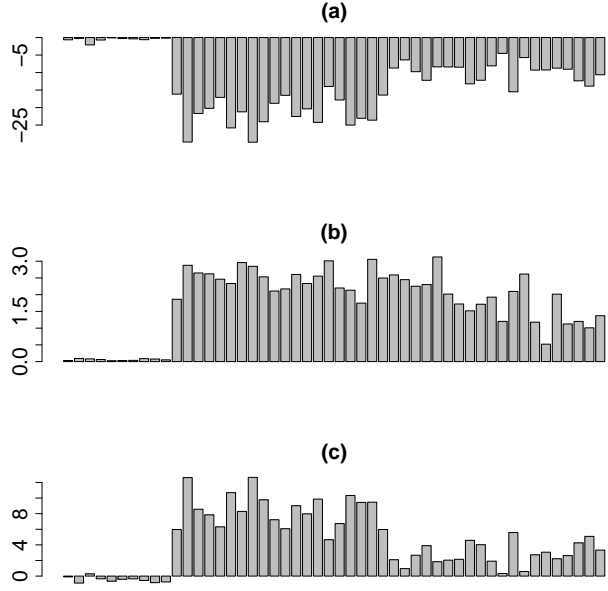} \\
   \caption{Estimation of p-values for 50 simulated datasets by IS with a single trial distribution:
    (a) $\log_{10}\E(\hat{q}_k^{\text{(IS)}})$,
   (b) $\text{cv}(\hat{q}_k^{\text{(IS)}})$, and (c) $\log_{10}[\text{cv}(\hat{q}_k^{\text{(DS)}})/\text{cv}(\hat{q}_k^{\text{(IS)}})]$ for $k=1,\ldots,50$. Each bar in a plot gives the result for one dataset.
    \label{fig:multipval}}
\end{figure}

\section{Estimating sampling distributions}\label{sec:errorbounds}

For a vector $\bfv=(v_j)$, denote its active set by $\calA(\bfv)=\{j:v_j\ne 0\}$. Let $\bmbeta_0=(\beta_{0j})_{1:p}$ be the true coefficient vector and $A_0=\calA(\bmbeta_0)$. Consider model \eqref{eq:linearmodel} with $\bmbeta=\bmbeta_0$ and $\bmeps \sim \dnorm_n(\bfzr,\sigma^2\bfI_n)$. The penalized loss in \eqref{eq:lassoloss} is, up to an additive constant, 
\begin{align}
 &\, \frac{1}{2} \|\bfY-\bfX \bmbeta \|_2^2 -\frac{1}{2} \|\bfY-\bfX \bmbeta_0 \|_2^2 
 + n \lambda \sum_{j=1}^p w_j (| \beta_{j}|-|\beta_{0j}|)  \nonumber\\
= &\, \frac{n}{2r_n^2} \bmdelta^{\trans} \bfC \bmdelta - \frac{n}{r_n}\bmdelta^{\trans} \bfU 
+ n\lambda \sum_{j \in A_0} w_j \left(| \beta_{0j}+ r_n^{-1}\delta_{j}|- | \beta_{0j}|\right) 
+ \frac{n\lambda}{r_n} \sum_{j \notin A_0} w_j |\delta_{j} | \nonumber\\
\defi &\, V(\bmdelta; \bmbeta_0, \bfU), \label{eq:defVn}
\end{align}
where $r_n>0$, $\bmdelta = (\delta_j)=r_n (\bmbeta - \bmbeta_0)$, and $\bfU=\bfX^{\trans}\bmeps/n$. 
Assuming $V$ has a unique minimizer, 
\begin{equation}\label{eq:defdelta}
\arg\min_{\bmdelta} V(\bmdelta;\bmbeta_0,\bfU)=r_n(\hbmbeta-\bmbeta_0)\defi \hbmdelta,
\end{equation} 
where $\hbmbeta$ is the Lasso-type estimator that minimizes \eqref{eq:lassoloss}. Suppose that $\ckbmbeta$ and $\hsigma$ are estimators of $\bmbeta_0$ and $\sigma$, following their respective sampling distributions. Let
\begin{equation}\label{eq:defdeltastar}
\sbmdelta=r_n(\sbmbeta-\ckbmbeta)=\arg\min_{\bmdelta} V(\bmdelta;\ckbmbeta,\sbfU), 
\end{equation}
where $\sbfU=\bfX^{\trans}\bmeps^*/n$ and $\bmeps^*\mid \hsigma \sim \dnorm_n(\bfzr,\hsigma^2\bfI_n)$. For random vectors $\bfZ_1$ and $\bfZ_2$, we use $\nu[\bfZ_1]$ to denote the distribution of $\bfZ_1$ and $\nu[\bfZ_1\mid\bfZ_2]$ to denote the conditional distribution of $\bfZ_1$ given $\bfZ_2$. In general, $\nu[\bfZ_1\mid\bfZ_2]$ is a random probability measure. The goal of this section is to derive nonasymptotic bounds on the difference between $\nu[\hbmdelta]$ and $\nu[\sbmdelta\mid \ckbmbeta,\hsigma]$, with all proofs relegated to Section~\ref{sec:proofs}. Our results provide theoretical justifications for any method that estimates the uncertainty in $\hbmbeta$ via simulation from $\nu[\sbmdelta\mid \ckbmbeta,\hsigma]$. In particular, for estimator augmentation, we draw $(\sbmbeta,\bfS^*)|\ckbmbeta, \hsigma^2$ from an estimated sampling distribution whose density is given by \eqref{eq:jointdensitynormal} or \eqref{eq:densitynormalhigh} with $\bmbeta=\ckbmbeta$ and $\sigma^2=\hsigma^2$. 

\subsection{Relevant existing results}

We compile relevant published results here. The following restricted eigenvalue (RE) assumption is a special case of the one used in \cite{Lounici11}, which extends the original definition in \cite{Bickel09}.

\begin{assumption}[RE$(m,c_0)$]
For some positive integer $m \leq p$ and a positive number $c_0$, the following condition holds: 
\begin{equation*}
\kappa(m,c_0) \defi \min_{|A| \leq m} \min_{\bmdelta \ne \bfzr} 
\left\{\frac{\|\bfX \bmdelta\|_2}{\sqn \|\bmdelta_{A} \|_2}:  
\sum_{j\in A^c} w_j |\delta_{j}|  \leq c_0 \sum_{j\in A} w_j | \delta_{j}|\right\} >0.
\end{equation*}
\end{assumption}

Lemma~\ref{lm:lassobounds} is from Theorem 3.1 and its proof in \cite{Lounici11}, regarding the Lasso as a special case of the group Lasso \citep{YuanLin06} with the size of every group being one. 
\begin{lemma}\label{lm:lassobounds}
Consider the model \eqref{eq:linearmodel} with $\bmbeta=\bmbeta_0$ and $\bmeps \sim \dnorm_n(\bfzr,\sigma^2\bfI_n)$, $\sigma^2>0$, and let $p\geq 2$, $n \geq 1$. Suppose that all the diagonal elements of $\bfC$ are $1$ and $\bfU=(U_j)_{1:p}=\frac{1}{n}\bfX^{\trans}\bmeps$. Let $q_0= |A_0|\leq q$, where $1\leq q \leq p$, and Assumption RE$(q,3)$ be satisfied. Choose $u>1$ and 
\begin{align}\label{eq:condlambda}
\lambda \geq  \lambda_0 \defi\frac{2\sigma}{w_{\min}} \sqrt{(2+5u\log p)/n},
\end{align}
where $w_{\min}=\inf\{w_j: j=1,\ldots,p\}$.
Then on the event $\calE=\cap_{j=1}^p\{|U_j|\leq w_j \lambda/2\}$, which happens with probability at least $1-2p^{1-u}$, for any minimizer $\hbmbeta$ of \eqref{eq:lassoloss} we have
\begin{align}
|\calA(\hbmbeta)| & \leq \frac{64\phi_{\max}}{\kappa^2(q,3)} \sum_{A_0} \frac{w_j^2}{w_{\min}^2},  \label{eq:boundsize}\\
\|\hbmbeta-\bmbeta_0 \|_1 & \leq \frac{16\lambda}{\kappa^2(q,3)} \sum_{A_0} \frac{w_j^2}{w_{\min}}, 
\label{eq:boundL1}
\end{align}
where $\phi_{\max}$ is the maximum eigenvalue of $\bfC$. If Assumption RE$(2q,3)$ is satisfied, then on the same event,
\begin{align}
\|\hbmbeta-\bmbeta_0 \|_2 & \leq \frac{4\sqrt{10}}{\kappa^2(2q,3)} \frac{\lambda \sum_{A_0}w_j^2}{w_{\min}\sqrt{q}}\defi \tau.
\label{eq:boundL2}
\end{align}
\end{lemma}

As an immediate consequence of this lemma, the Lasso-type estimator has the screening property assuming a suitable beta-min condition.
\begin{lemma}\label{lm:screen}
Let the assumptions in Lemma~\ref{lm:lassobounds} be satisfied. If $\inf_{A_0} | \beta_{0j}|>\tau$,
then on the event $\calE$, $A_0\subseteq \calA(\hbmbeta)$ for any $\hbmbeta$. If $\inf_{A_0} | \beta_{0j}| > 2\tau$,
then on the same event, 
\begin{equation*}
\left\{j: |\hbeta_j|>\tau\right\} = A_0.
\end{equation*}
\end{lemma}

\subsection{Known variance}\label{sec:errorboundsknownvar}

In this subsection we assume that the noise variance $\sigma^2$ is known and fix $\hsigma=\sigma$ in the definition of $\sbmdelta$ \eqref{eq:defdeltastar}. We first regard $\ckbmbeta=(\ckbeta_j)_{1:p}$ as a fixed vector and find conditions which are sufficient for the distribution of $\sbmdelta$ to be close to that of $\hbmdelta$. Then we construct an estimator that satisfies these conditions with high probability. Let 
\begin{equation}\label{eq:difftdbeta}
\eta  \defi \sup_{j \in A_0} \frac{|\ckbeta_{j}-\beta_{0j}|}{| \beta_{0j}|}.
\end{equation}

\begin{lemma}\label{lm:fixknownvar}
Assume that the columns of $\bfX$ are in general position. Fix $\hsigma=\sigma$. Suppose that $\ckbmbeta$ is a fixed vector in $\R^p$ so that (i) $\ckbeta_{j}=0$ for all $j \notin A_0$ and (ii) $\eta\in[0,1)$ as defined in \eqref{eq:difftdbeta}. Let $M_1>0$ and assume  
\begin{align}\label{eq:assbeta0j}
\inf_{j\in A_0}|\beta_{0j}|>  \frac{M_1}{r_n(1-\eta)}.
\end{align} 
Then we have
\begin{equation*}
\nu[\sbmdelta\mid \|\sbmdelta \|_{\infty}< M_1]=\nu[\hbmdelta\mid \|\hbmdelta\|_{\infty}<M_1].
\end{equation*}
\end{lemma}

One possible way to construct $\ckbmbeta$ that satisfies conditions (i) and (ii) in Lemma~\ref{lm:fixknownvar} is to threshold the Lasso-type estimator $\hbmbeta$ by a constant $b_{\supth}>0$, i.e.,
\begin{equation}\label{eq:thresholddef}
\ckbeta_{j}=\hbeta_{j}\I(|\hbeta_{j}|>b_{\supth}), \quad j=1,\ldots,p.
\end{equation}

\begin{theorem}\label{thm:threshold}
Let the assumptions in Lemma~\ref{lm:lassobounds} be satisfied and assume that the columns of $\bfX$ are in general position. Choose $r_n$ such that $\|\hbmdelta\|_{\infty} < M_1$ with probability at least $1-\alpha_1$. Fix $\hsigma=\sigma$, define $\ckbmbeta$ by \eqref{eq:thresholddef} with $b_{\supth}=\tau$ \eqref{eq:boundL2}, and assume 
\begin{align}\label{eq:betamin}
\inf_{j\in A_0}|\beta_{0j}|>  \max\left\{2\tau, \frac{M_1}{r_n(1-\eta)}\right\}.
\end{align}
Then with probability at least $1-2p^{1-u}$, we have 
\begin{equation}
\sup_{B\in \scrR^p} |P(\sbmdelta \in B \mid \ckbmbeta) - P(\hbmdelta \in B)| \leq 2\alpha_1, \label{eq:supPrdiff}
\end{equation}
where $\scrR^p$ is the $\sigma$-field of $p$-dimensional Borel sets.
\end{theorem}

\begin{remark}
Depending on the estimator $\ckbmbeta$, the conditional probability $P(\sbmdelta \in B \mid \ckbmbeta)$ is a random variable. The probability $1-2p^{1-u}$ is with respect to the sampling distribution of $\bfy$. This theorem gives an explicit nonasymptotic bound \eqref{eq:supPrdiff} on the difference between the two probability measures, $\nu[\sbmdelta \mid \ckbmbeta]$ and $\nu[\hbmdelta]$, and justifies simulation from the estimated sampling distribution with $\bmbeta=\ckbmbeta$. 
\end{remark}

\begin{remark}\label{rmk:asymknownvar}
Consider the asymptotic implications of Theorem~\ref{thm:threshold} by allowing $q_0,p\to \infty$ as $n\to \infty$. Suppose that $w_{\min}$ and $w_{\max}=\sup_j w_j$ stay bounded away from 0 and $\infty$. For the Lasso, $w_{\min}=w_{\max}=1$. Choose $r_n$ so that $\|\hbmdelta\|_{\infty}\asymp_P 1$, i.e., the convergence rate of $\|\hbmbeta-\bmbeta_0\|_{\infty}$ is $1/r_n$. For any $\alpha_1>0$, there is $M_1<\infty$ so that $P(\|\hbmdelta\|_{\infty}<M_1)\geq 1-\alpha_1$. Let $q=q_0$ in Lemma~\ref{lm:lassobounds} and assume that $\liminf_n\kappa(m,3)>0$ when $m=O(q_0)=o(n)$. With a suitable choice of $\lambda \asymp \sqrt{(\log p)/n}$,
\begin{equation*}
\tau \asymp \sqrt{q_0} \lambda \asymp \sqrt{q_0 (\log p)/n}.
\end{equation*}
Because $1 \asymp_P r_n\|\hbmbeta-\bmbeta_0\|_{\infty}=O_P(r_n\tau)$, we have $r_n^{-1}=O(\tau)$. Thus, the order of $r_n$ satisfies 
\begin{equation}\label{eq:orderofr_n}
r_n^{-1}=O(\sqrt{q_0 (\log p)/n}) \text{ and } r_n = O(\sqn).
\end{equation}
Consequently, a sufficient condition for \eqref{eq:betamin} is
\begin{equation}\label{eq:orderofbetamin}
\inf_{A_0}|\beta_{0j}|\gg  \sqrt{q_0 (\log p)/n}
\end{equation}
with the order of $b_{\supth}$ in between. Theorem~\ref{thm:threshold} then implies that \eqref{eq:supPrdiff} holds with probability at least $1-2p^{1-u}\to 1$ as $p\to \infty$. Choosing $\alpha_1$ arbitrarily close to zero, this demonstrates that 
\begin{equation}\label{eq:supPrdifftinP}
\sup_{B\in \scrR^p} |P(\sbmdelta \in B \mid \ckbmbeta) - P(\hbmdelta\in B )|\toinP 0.
\end{equation}
As $\inf_{A_0}|\beta_{0j}|$ may decay to zero at a rate slower than $\sqrt{{q_0 (\log p)}/{n}}\to 0$, Theorem~\ref{thm:threshold} applies in the high-dimensional setting ($p\gg n\to\infty$). 
\end{remark}

\subsection{Unknown variance}\label{sec:theoryunknownvar}

When $\sigma^2$ is unknown, recall that $\bmeps^*\mid \hsigma \sim \dnorm_n(\bfzr,\hsigma^2\bfI_n)$ and $\sbfU=\bfX^{\trans}\bmeps^*/n$. Denote the components of $\sbfU$ by $U_j^*$, $j=1,\ldots,p$. Define 
\begin{align}\label{eq:defdeltastar0}
\sbmdelta_0=\arg\min_{\bmdelta} V(\bmdelta;\ckbmbeta,(\sigma/\hsigma)\sbfU), 
\end{align}
whose distribution does not depend on $\hsigma$ and is identical to that of $\sbmdelta$ when $\hsigma$ is fixed to $\sigma$. Thus, results in Section~\ref{sec:errorboundsknownvar} show that $\nu[\sbmdelta_0\mid \ckbmbeta]$ is close to $\nu[\hbmdelta]$. We will further bound the difference, $\sbmdelta-\sbmdelta_0$, in this subsection. In other words, $\sbmdelta_0$ serves as an intermediate variable between $\hbmdelta$ and $\sbmdelta$ to help quantify the difference between their distributions.

\begin{lemma}\label{lm:fixunknownvar}
Assume that the columns of $\bfX$ are in general position. Let $\ckbmbeta$ and $\hsigma$ be fixed such that $|\calA(\ckbmbeta)|\leq q$ with $1\leq q \leq p$ and
\begin{equation}\label{eq:diffhsigma}
\max\{\left|{\hsigma}/{\sigma}-1\right|, \left|{\sigma}/{\hsigma}-1\right|\} \leq \zeta \in [0,1).
\end{equation}
If Assumption RE$(q,3)$ is satisfied, then on the event $\cap_{j=1}^p\{|U_j^*| \leq (1-\zeta)w_j\lambda/2\}$, 
\begin{equation}\label{eq:boundquadratic}
\frac{1}{n}\|\bfX (\sbmdelta-\sbmdelta_0) \|^2_2 \leq \frac{32 w_{\max}r_n^2\lambda^2 \zeta}{w_{\min}\kappa^2(q,3)}  \sum_{A_0} w_j^2.
\end{equation}
\end{lemma}

Based on \eqref{eq:boundquadratic}, we can obtain an upper bound on $\|\sbmdelta-\sbmdelta_0\|_2$ via the restricted eigenvalues of $\bfC$. Define the minimum restricted eigenvalue of $\bfC$ for an integer $m\leq p$ by
\begin{equation}\label{eq:remin}
\phi_{\min}(m)=\min_{1\leq |\calA(\bfv)|\leq m} \frac{\bfv^{\trans}\bfC \bfv}{\|\bfv\|_2^2},
\end{equation}
and let
\begin{equation*}
M_2=\frac{128\phi_{\max}}{\kappa^2(q,3)} \sum_{A_0} \frac{w_j^2}{w_{\min}^2},
\end{equation*}
which is twice the upper bound on $|\calA(\hbmbeta)|$ in \eqref{eq:boundsize}.

\begin{theorem}\label{thm:thresholdsigmaest}
Let the assumptions in Theorem~\ref{thm:threshold} be satisfied but without fixing $\hsigma$ to $\sigma$. In addition, assume that \eqref{eq:diffhsigma} holds with probability at least $1-\alpha_2$ and $\phi_{\min}(M_2)>0$. Choose
\begin{equation}\label{eq:condlambdasigmaest}
\lambda \geq \frac{1+\zeta}{1-\zeta} \lambda_0,
\end{equation}
where $\lambda_0$ is defined in \eqref{eq:condlambda}.
Then with probability at least $1-(\alpha_2+2p^{1-u})$, we have  
\begin{equation}\label{eq:supPrdiffdelta0}
\sup_{B\in \scrR^p} |P(\sbmdelta_0 \in B \mid \ckbmbeta) - P(\hbmdelta \in B)| \leq 2\alpha_1
\end{equation}
and 
\begin{equation}\label{eq:L2Deltabound}
P\left\{\left.\|\sbmdelta-\sbmdelta_0 \|^2_2 \leq \frac{32 w_{\max}r_n^2\lambda^2 \zeta \sum_{A_0} w_j^2}{w_{\min}\kappa^2(q,3)\phi_{\min}(M_2)}\right| \ckbmbeta, \hsigma\right\}\geq 1- 2p^{1-u}. 
\end{equation}
\end{theorem}

\begin{remark}
Assume that $\max\{\left|{\hsigma}/{\sigma}-1\right|, \left|{\sigma}/{\hsigma}-1\right|\} =O_P(\zeta_n)$ with $\zeta_n\to0$. Following the asymptotic framework in Remark~\ref{rmk:asymknownvar}, we can establish by the same reasoning that
\begin{equation}\label{eq:uniformdelta0}
\sup_{B\in \scrR^p} |P(\sbmdelta_0 \in B \mid \ckbmbeta) - P(\hbmdelta\in B )|\toinP 0.
\end{equation}
Suppose that $\liminf_n\phi_{\min}(m)>0$ for $m=o(n)$ and that $\phi_{\max}$ is bounded from above by a constant (or at least does not diverge too fast). Then $M_2=O(\phi_{\max}q_0)=o(n)$ and $\phi_{\min}(M_2)$ is bounded from below by a positive constant. The upper bound in \eqref{eq:L2Deltabound} becomes $O(r_n^2 \lambda^2 \zeta_n q_0)=O(\zeta_n q_0 \log p)$ as $r_n=O(\sqn)$ \eqref{eq:orderofr_n} and $\lambda^2\asymp \log(p)/n$. If $\zeta_n =o(1/(q_0 \log p))$, then \eqref{eq:L2Deltabound} implies that
\begin{equation*}
P(\|\sbmdelta-\sbmdelta_0\|_2>\epsilon\mid\ckbmbeta,\hsigma)\to 0
\end{equation*}
for any $\epsilon>0$. Combing with \eqref{eq:uniformdelta0} this shows that, with probability tending to one, $\nu[\sbmdelta_J\mid \ckbmbeta,\hsigma]$ converges weakly to $\nu[\hbmdelta_J]$ for any fixed index set $J\subseteq\{1,\ldots,p\}$. If $\sigma$ is estimated by the scaled Lasso \citep{SunZhang12}, one may reach $\zeta_n=O(q_0(\log p)/n)$ under certain conditions by their Theorem 2 and it is then sufficient to have $q_0\log p\ll \sqn$. If $p$ is fixed, we only need $\zeta_n =o(1)$ and hence any consistent estimator of $\sigma$ will be sufficient. In this case, $r_n\asymp \sqn$ \eqref{eq:orderofr_n} and with probability tending to one, $\nu[\sbmdelta\mid \ckbmbeta,\hsigma]$ converges weakly to $\nu[\hbmdelta]$.
\end{remark}

The key assumptions on the underlying model are the RE assumption on the Gram matrix $\bfC$ and beta-min and sparsity assumptions on the true coefficients $\bmbeta_0$, which are comparable to those in Lemma~\ref{lm:lassobounds} and Lemma~\ref{lm:screen}. There is an extra assumption that $\phi_{\min}(M_2)>0$ in Theorem~\ref{thm:thresholdsigmaest}, which is again imposed on the restricted eigenvalues of $\bfC$. If $\bfX$ is drawn from a continuous distribution on $\R^{n\times p}$, then with probability one $\phi_{\min}(m)>0$ for any $m\leq n$. It should be noted that we do not assume the irrepresentable condition \citep{ZhaoYu06,Meinshausen06,Zou06}, which is much stronger than the assumptions on the restricted eigenvalues of $\bfC$.

We compare our results to residual bootstrap for approximating the sampling distribution of the Lasso. To be precise, a residual bootstrap is equivalent to Routine~\ref{rt:directsample} with $\bmbeta$ estimated by $\ckbmbeta$ and $\bmeps^{(t)}$ drawn by resampling residuals. Assuming $p$ is fixed, \cite{KnightFu00} argue that the residual bootstrap may be consistent if $\ckbmbeta$ is model selection consistent and \cite{Chatterjee11} establish such fixed-dimensional consistency when $\ckbmbeta$ is constructed by thresholding the Lasso, in the same spirit as \eqref{eq:thresholddef}. As discussed above, Theorem~\ref{thm:thresholdsigmaest} applied to a fixed $p$ is clearly in line with these previous works. However, our results are much more general by providing explicit nonasymptotic bounds that imply consistency when $p\gg n\to \infty$. More recently, \cite{Chatterjee13} have shown that the residual bootstrap is consistent for the adaptive Lasso \citep{Zou06} when $p>n\to \infty$ under a number of conditions. A fundamental difference is that the weights $w_j$ are specified by an initial $\sqn$-consistent estimator in their work and will not stay bounded as $n\to\infty$. Therefore, their results do not apply to the Lasso. On the contrary, the results in this section are derived assuming the weights $w_j$ are constants without being specified by any initial estimator. In addition, \cite{Chatterjee13} impose in their Theorem 5.1 that $\inf_{A_0}|\beta_{0j}|\geq K$ for some $K\in(0,\infty)$ when $p>n$, which disallows the decay of the magnitudes of nonzero coefficients as $n$ grows. This is considerably stronger than our assumption \eqref{eq:orderofbetamin}.

\section{Generalizations and discussions}\label{sec:generalization}

\subsection{Random design}

We generalize the Monte Carlo methods to a random design, 
assuming that $\bfX$ is drawn from a distribution $f_{\bfX}$. 
The distribution of the augmented estimator $(\hbmbeta_{\calA},\bfS_{\calI},\calA)$, \eqref{eq:jointdensity} 
and \eqref{eq:jointdensityhigh}, becomes a conditional distribution given $\bfX=\bfx$, written as
$\pi(\bfb_A,\bfs_I,A\mid \bfx)$ and ${\pi}_r(\bfb_A,\bfs_I,A\mid \bfx)$, respectively. 

In the low-dimensional setting, 
we may generalize the MLS (Routine~\ref{rt:MHLS}) to draw samples from %the joint distribution 
${\pi}(\bfb_A,\bfs_I,A,\bfx)={\pi}(\bfb_A,\bfs_I,A\mid \bfx){f}_{\bfX}(\bfx)$
and approximate the sampling distribution of $\hbmbeta$.
It may be difficult to assume or estimate a reliable density for $\bfX$, 
but it is sufficient for the development of an MH sampler under a random design (rdMLS) if we can draw from ${f}_{\bfX}(\bfx)$.
As seen below, we do not need an
explicit form of $f_{\bfX}(\bfx)$ for computing the MH ratio \eqref{eq:rMHration}
and thus may draw $\bfx^{\dag}$ by the bootstrap.

\begin{routine}[rdMLS]\label{rt:rMHLS}
Suppose the current sample is $(\bfb_A,\bfs_I,A,\bfx)^{(t)}$.
\begin{itemize}
\item[(1)]{Draw $\bfx^{\dag}$ from ${f}_{\bfX}$, and accept it as $\bfx^{(t+1)}$ with probability
\begin{equation}\label{eq:rMHration}
\min\left\{1, \frac{\pi((\bfb_A,\bfs_I,A)^{(t)}\mid \bfx^{\dag})}
{\pi((\bfb_A,\bfs_I,A)^{(t)}\mid \bfx^{(t)})} \right\};
\end{equation}
otherwise, set $\bfx^{(t+1)}=\bfx^{(t)}$.}
\item[(2)]{Regarding ${\pi}(\bfb_A,\bfs_I,A\mid \bfx^{(t+1)})$ as the target density,
apply one iteration of the MLS (Routine~\ref{rt:MHLS}) to obtain $(\bfb_A,\bfs_I,A)^{(t+1)}$.}
\end{itemize}
\end{routine}

Generalization of the IS algorithm (Routine~\ref{rt:IS}) is also straightforward.
Draw $\bfx^{(t)}$ from $f_{\bfX}$ and draw $(\bfb_A,\bfs_I,A)^{(t)}$ from the trial distribution
given $\bfX=\bfx^{(t)}$. Calculate importance weights by \eqref{eq:ISweight} with $\bfX=\bfx^{(t)}$,
and apply the same estimation \eqref{eq:ISest}. Again, an explicit expression for $f_{\bfX}$ is unnecessary.
But bootstrap sampling from $\bfX$ is not a choice for the high-dimensional setting, 
because a bootstrap sample from $\bfX$ violates Assumption~\ref{as:X}.

\subsection{Model selection consistency}

The distribution of the augmented estimator may help 
establish asymptotic properties of a Lasso-type
estimator. Here, we demonstrate this point by studying the model selection consistency of
the Lasso. Our goal is not to establish new asymptotic results, but to provide an intuitive
and geometric understanding of the technical conditions in existing work.
Recall that $\calA$ and $A_0$ are the respective active sets of $\hbmbeta$ and $\bmbeta_0$. 
Let $q_0=|A_0|$ and $\bfs_{0}=\sgn(\bmbeta_{0A_0})$. 
Without loss of generality, assume $A_0=\{1,\ldots,q_0\}$ and $I_0=\{q_0+1,\ldots,p\}$.
We allow both $p$ and $q_0$ to grow with $n$.

\begin{definition}[sign consistency \citep{MeinshausenYu09}]\label{def:signconsistency}
We say that $\hbmbeta$ is sign consistent for $\bmbeta_0$ if
\begin{equation}\label{eq:signconsistency}
P(\calA=A_0, \sgn(\hbmbeta_{A_0})=\bfs_0) \to 1, \text{ as } n\to \infty.
\end{equation}
\end{definition}

If $\hbmbeta$ is unique, the size of its active set $|\calA|\leq n$ (Lemma~\ref{lm:rank}), 
and thus $\bfD(\calA)$ is invertible from \eqref{eq:defD}. 
Therefore, the definitions of $\bmmu(A,\bfs_A;\bmbeta)$ and $\bmSigma(A;\sigma^2)$ in \eqref{eq:munorm}
and \eqref{eq:varnorm} are also valid for any $(\bfb_A,\bfs_I,A)\in \Omega_r$ when $p>n$. 
Rewrite the KKT condition \eqref{eq:lassojoint} as
\begin{equation}\label{eq:defTheta}
\bmTheta=[\bfD(\calA)]^{-1} \bfU + \bmmu(\calA,\bfS_{\calA};\bmbeta_0),
\end{equation}
where $\bmTheta=(\hbmbeta_{\calA},\bfS_{\calI}) \in \R^p$.
Fixing $\calA=A_0$ and $\bfS_{A_0}=\bfs_0$ in \eqref{eq:defTheta}, we define a random vector
\begin{equation}\label{eq:defZ}
\bfZ=[\bfD(A_0)]^{-1} \bfU + \bmmu(A_0,\bfs_0;\bmbeta_0)
\end{equation}
via an affine map of $\bfU$. 
Note that we always have $\E(\bfU)=\bfzr$ and $\Var(\bfU)=\frac{\sigma^2}{n}\bfC \geq 0$, regardless of the sizes of $n$ and $p$.
When $p>n$, $\Var(\bfU)$ is semipositive definite, 
meaning that components of $\bfU$ are linearly dependent of each other, since $\bfU$ only lies in $\row(\bfX)$,
a proper subspace of $\R^p$. Consequently, $\E(\bfZ)=\bmmu(A_0,\bfs_{0};\bmbeta_0)\defi \bmmu^0$ and 
$\Var(\bfZ)=\bmSigma(A_0;\sigma^2) \defi \bmSigma^0\geq 0$.
Simple calculation from \eqref{eq:munorm} and \eqref{eq:varnorm} gives
\begin{eqnarray}
&\left(\begin{array}{c}
\bmmu^0_{A_0} \\ \bmmu^0_{I_0} \end{array}\right)  =
\left(
\begin{array}{c}
\bmbeta_{0A_0}-\lambda\bfC_{A_0A_0}^{-1}\bfW_{A_0A_0} \bfs_{0}  \\
\bfW_{I_0I_0}^{-1}\bfC_{I_0A_0}\bfC_{A_0A_0}^{-1}\bfW_{A_0A_0}\bfs_{0}
\end{array}
\right), \label{eq:muA0} \\
&\bmSigma^0 = 
\frac{\sigma^2}{n}\left( \begin{array}{cc}
\bfC_{A_0A_0}^{-1} &  \bfzr \\
\bfzr & \lambda^{-2} \bfW_{I_0I_0}^{-1}\bfC_{I_0\mid A_0} \bfW_{I_0I_0}^{-1}
 \end{array}\right), \label{eq:sigmaA0}
\end{eqnarray}
where $\bfC_{I_0\mid A_0} = \bfC_{I_0I_0}-\bfC_{I_0A_0}\bfC_{A_0A_0}^{-1}\bfC_{A_0I_0}$. 

\begin{lemma}\label{lm:signtoZ}
If the columns of $\bfX$ are in general position and $\bmeps$ is i.i.d. with mean zero and variance $\sigma^2$, then for any $p\geq 1$ and $n \geq 1$,
\begin{equation}\label{eq:normsign}
P(\calA=A_0, \sgn(\hbmbeta_{A_0})=\bfs_0) = P(\bfZ \in \Omega_{A_0,\bfs_0}),
\end{equation}
where $\Omega_{A_0,\bfs_0}$ is defined by \eqref{eq:subsetAs}.
\end{lemma}
\begin{proof}
If $\calA=A_0$ and $\bfS_{A_0}=\bfs_0$, then $\bmTheta\in\Omega_{A_0,\bfs_0}$ by definition. In this case, \eqref{eq:defTheta} reduces to \eqref{eq:defZ} and we have $\bfZ=\bmTheta \in \Omega_{A_0,\bfs_0}$. Reversely, if $\bfZ \in \Omega_{A_0,\bfs_0}$, then $(\bfZ,A_0)$ is a solution to the KKT condition \eqref{eq:defTheta}. By uniqueness, $(\bmTheta,\calA)=(\bfZ,A_0)$ and therefore $\calA=A_0$ and $\bfS_{A_0}=\bfs_0$.
\end{proof}

Consequently, to establish sign consistency, we only need a set of sufficient conditions for
$P(\bfZ \in \Omega_{A_0,\bfs_0})\to 1$: 
(C1) $\sgn(\bmmu^0_{A_0})= \bfs_0$. 
(C2) $\|\bmmu^0_{I_0} \|_{\infty} \leq c$ for some $c \in (0,1)$. 
(C3) Let $Z_j$ and $\mu^0_j$ be the $j^\supth$ components of $\bfZ$ and $\bmmu^0$, respectively.
As $n\to\infty$,
\begin{equation}\label{eq:C3}
P\left(|Z_j - \mu^0_j| < \delta_j, \,\forall j \right)\to 1,
\end{equation}
where $\delta_j=|\mu^0_j|$ for $j\in A_0$ and $\delta_j=1-c$ for $j\in I_0$.

The first two conditions ensure that 
$\bmmu^0=\E(\bfZ)$ lies in the interior of $\Omega_{A_0,\bfs_0}$. The third condition guarantees that $\bfZ$
always stays in a box centered at $\bmmu^0$, and the box is contained in $\Omega_{A_0,\bfs_0}$ if (C1) and (C2) hold.
These conditions have a simple and intuitive geometric interpretation illustrated in Figure~\ref{fig:consistent}.
\begin{lemma}\label{lm:signconsistency}
Assume the columns of $\bfX$ are in general position and $\bmeps$ is i.i.d. with mean zero and variance $\sigma^2$. If conditions
(C1), (C2), and (C3) hold as $n\to\infty$, then $\hbmbeta$ is sign consistent for $\bmbeta_0$, regardless of the relative
size between $p$ and $n$.
\end{lemma}

\begin{figure}[t]
\centering
   \includegraphics[width=0.35\linewidth,angle=-90,trim=1in 2in 1.5in 1in,clip]{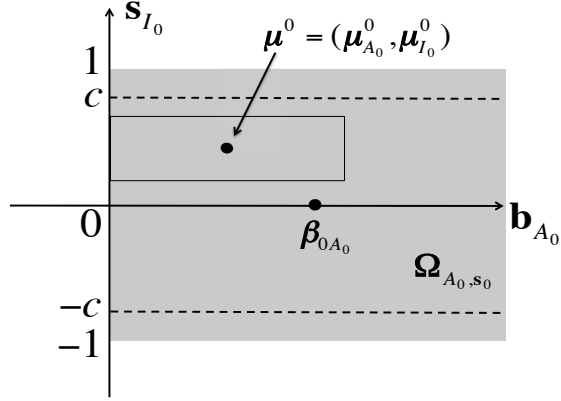} \\
   \caption{Geometric interpretation of the conditions for sign consistency. Shaded area represents $\Omega_{A_0,\bfs_0}$,
   where $\bfs_0=\sgn(\bmbeta_{0A_0})$.
    \label{fig:consistent}}
\end{figure}

Now we may recover some of the conditions for establishing consistency of the Lasso in the literature.
In what follows, let $\bfW=\bfI_p$ in \eqref{eq:muA0} and \eqref{eq:sigmaA0}. Condition (C2) is 
the strong irrepresentable condition \citep{ZhaoYu06,Meinshausen06,Zou06}:
\begin{equation*}
\|\bfC_{I_0A_0}\bfC_{A_0A_0}^{-1}\bfs_{0}\|_{\infty} \leq c \in (0,1).
\end{equation*}
Assume that the minimum eigenvalue of $\bfC_{A_0A_0}$ is bounded from below by $\phi_0>0$, which is equivalent to condition (6) in \cite{ZhaoYu06}. Condition (C1) holds if
\begin{equation}\label{eq:C1}
\frac{\lambda\|\bfC_{A_0A_0}^{-1}\bfs_{0}\|_{\infty}}{\inf_{j \in A_0} |\beta_{0j}|} 
\leq \frac{\lambda \phi_0^{-1} \|\bfs_0\|_2}{\inf_{j \in A_0} |\beta_{0j}|}
=\frac{\phi_0^{-1}\lambda \sqrt{q_0}}{\inf_{j \in A_0} |\beta_{0j}|}\to 0.
\end{equation}
This shows that some version of a beta-min condition is necessary to enforce
a lower bound for $\inf_{j \in A_0} |\beta_{0j}|$. For example, we may assume that
\begin{equation}\label{eq:betaminEA}
\lim_{n\to\infty} n^{a_1}\inf_{j \in A_0} |\beta_{0j}| \geq M_3, 
\end{equation}
for some positive constants $M_3$ and $a_1$, which is the same as (8) in \cite{ZhaoYu06}.
Then one needs to choose $\lambda=o(n^{-a_1}/\sqrt{q_0})\to 0$ for (C1) to hold.

Let $\bfd=(d_1,\ldots,d_p)$ such that $\bfd_{A_0}=\sigma^2\diag(\bfC_{A_0A_0}^{-1})$
and $\bfd_{I_0}=\sigma^2 \diag(\bfC_{I_0\mid A_0})$.
To establish condition (C3), assume that $\bmeps\sim\dnorm_n(\bfzr,\sigma^2 \bfI_n)$ and
$\diag(\bfC)$ is bounded from above.
Then all $d_j$ are bounded and let $d_*<\infty$ be an upper bound of $\{d_1,\ldots,d_p\}$.
Furthermore, $Z_j$ follows a univariate normal distribution:
For $j\in A_0$, $Z_j \sim \dnorm(\mu^0_j,n^{-1}d_j)$ and for $j\in I_0$,
$Z_j \sim \dnorm(\mu^0_j,n^{-1}\lambda^{-2} d_j)$ according to \eqref{eq:sigmaA0} with $\bfW=\bfI_p$. 
By \eqref{eq:C1} and \eqref{eq:betaminEA},
$\delta_j=|\mu^0_j|\geq \frac{1}{2}M_3 n ^{-a_1}$ for $j\in A_0$ as $n\to \infty$, and
\begin{equation}\label{eq:C3A0}
P\left(\sup_{j \in A_0}|Z_j - \mu^0_j| \geq \frac{1}{2}M_3 n ^{-a_1} \right)
\leq 2 q_0 \exp\left(-\frac{M_3^2 n^{1-2a_1}}{8 d_*}  \right) \to 0, 
\end{equation}
as long as $a_1<1/2$ and $q_0<n$. Since $\delta_j=1-c$ for all $j\in I_0$,
\begin{equation}\label{eq:C3I0}
P\left(\sup_{j \in I_0}|Z_j - \mu^0_j| \geq 1-c \right) 
\leq  2 \exp\left(-\frac{n \lambda^2(1-c)^2}{2 d_*}+ \log p\right)\to 0,
\end{equation}
if $(\log p)/(n \lambda^2) \to 0$ and $n\lambda^2 \to \infty$. Clearly, the above two inequalities imply \eqref{eq:C3}.
Therefore, $\lambda$ must satisfy
$\sqrt{(\log p)/n} \ll \lambda = o(n^{-a_1}/\sqrt{q_0})$, which implies that $\sqrt{q_0 (\log p)/n} =o(n^{-a_1})$. 
This is consistent with the beta-min condition in \cite{Meinshausen06}: $\inf_{j \in A_0} |\beta_{0j}|\gg \sqrt{q_0 (\log p)/n}$. 
In summary, choosing $a_1,a_2,a_3>0$ such that $a_2+a_3< 1-2a_1$, the Lasso can be consistent for model selection 
with $q_0 = O(n^{a_2})$ and $p=O(\exp(n^{a_3}))$, both diverging with $n$.
For more general scaling of $(n, p, q_0)$, see the work by \cite{Wainwright09}.

\begin{remark}
The term $\log p$ in \eqref{eq:C3I0} can be replaced by $\log(p-q_0)$, which will improve the bound
if $q_0/p$ does not vanish as $n\to\infty$. Moreover, both inequalities \eqref{eq:C3A0} and \eqref{eq:C3I0} are applicable
to sub-Gaussian noise. 
\end{remark}

\subsection{Bayesian interpretation}\label{sec:Bayesian}

It is well-known that the Lasso can be interpreted as the mode of the posterior distribution of $\bmbeta$
under a Laplace prior. However, the posterior distribution itself is continuous on $\R^p$. 
If we draw $\bmbeta$ from this posterior distribution, every component of $\bmbeta$ will be nonzero with
probability one. In this sense, sampling from this posterior distribution does not provide a direct solution
to model selection, which seems unsatisfactory from a Bayesian perspective. Here,
we discuss a different Bayesian interpretation of the Lasso-type estimator $\hbmbeta$ 
from a sampling distribution point of view. 

Assume that $\rank(\bfX)=p< n$ and thus $\bfC$ is invertible.
Under the noninformative prior $p(\bmbeta,\sigma^2) \propto 1/\sigma^2$ and
the assumption that $\bmeps\sim \dnorm_n (\bfzr, \sigma^2 \bfI_n)$, the conditional and marginal 
posterior distributions of $\bmbeta$ are 
\begin{eqnarray}
\bmbeta \mid \sigma^2, \bfY & \sim & \dnorm_p(\hbmbeta^{\text{OLS}},n^{-1}\sigma^2\bfC^{-1}), \label{eq:condposterior}\\
\bmbeta \mid \bfY & \sim & t_{n-p}(\hbmbeta^{\text{OLS}},n^{-1}\hat{\sigma}^2\bfC^{-1}), \label{eq:marginalposterior}
\end{eqnarray}
where $\hat{\sigma}^2$ is given by \eqref{eq:estsigma2}
with $\ckbmbeta=\hbmbeta^{\text{OLS}}$ and $t_{n-p}(\bmmu,\bmSigma)$ is the multivariate $t$ distribution 
with $(n-p)$ degrees of freedom, location $\bmmu$, and scale matrix $\bmSigma$.

Following the decision theory framework, let 
$\bmeta\in\R^p$ be a decision regarding $\bmbeta$ that incurs the loss 
\begin{equation}\label{eq:Bayesloss}
\ell_{B}(\bmeta,\bmbeta)= 
\frac{1}{2}(\bmeta-\bmbeta)^{\trans} \bfC (\bmeta-\bmbeta) + \lambda \| \bfW \bmeta\|_1.
\end{equation}
Since the covariance of $\bmbeta$ is proportional to $\bfC^{-1}$ with respect to 
the posterior distribution \eqref{eq:condposterior} or \eqref{eq:marginalposterior},
$\ell_{B}(\bmeta,\bmbeta)$ is essentially the squared Mahalanobis distance between $\bmeta$
and $\bmbeta$, plus a weighted $\ell_1$ norm of $\bmeta$ to encourage sparsity. % in $\bmeta$.
Denote by $\tdbmbeta$ the optimal decision that minimizes the 
loss $\ell_B$ for a given $\bmbeta$, i.e.,
$\tdbmbeta=\arg\min_{\bmeta} \ell_{B}(\bmeta,\bmbeta)$.
Let $\tilde{\bfS}$ be the subgradient of $\|\bmeta\|_1$ at $\tdbmbeta$. 
The KKT condition for $\tdbmbeta$ is
\begin{equation}\label{eq:Bayesiangrad}
\bfC \tdbmbeta + \lambda \bfW \tilde{\bfS} =\bfC\bmbeta.
\end{equation}
Since $\bmbeta$ is a random vector in Bayesian inference, the distribution of $\bmbeta$
determines the joint distribution of $\tdbmbeta$ and $\tilde{\bfS}$ via the above KKT condition.
Represent $(\tdbmbeta,\tilde{\bfS})$ by its equivalent form $(\tdbmbeta_{\tdA},\tilde{\bfS}_{\tdI},\tdA)$ in the same way as for $(\hbmbeta,\bfS)$ in Section~\ref{sec:distribution}.

The conditional posterior distribution \eqref{eq:condposterior} implies that 
$\bfC\bmbeta \mid \sigma^2, \bfY\sim \dnorm_p(\bfC\hbmbeta^{\text{OLS}}, n^{-1}{\sigma^2}\bfC)$.
Thus, conditional on $\bfY$ and $\sigma^2$, Equation \eqref{eq:Bayesiangrad} implies that
\begin{equation}\label{eq:Bayesiangraddistr}
\bfC \tdbmbeta + \lambda \bfW \tilde{\bfS} - \bfC\hbmbeta^{\text{OLS}} \eqinL \bfU,
\end{equation}
where $\bfU\sim\dnorm_p(\bfzr,n^{-1}{\sigma^2}\bfC)$. One sees that \eqref{eq:Bayesiangraddistr}
is identical to the KKT condition \eqref{eq:lassoKKTinU} with $\hbmbeta^{\text{OLS}}$ in place of $\bmbeta$. 
Therefore, the conditional distribution $[\tdbmbeta_{\tdA},\tilde{\bfS}_{\tdI},\tdA\mid \sigma^2, \bfY]$, 
determined by \eqref{eq:Bayesiangraddistr},
is identical to the estimated sampling distribution $\hat{\pi}$ \eqref{eq:estjointdensity} under a normal error
distribution with $\bmbeta$ estimated by $\hbmbeta^{\text{OLS}}$, i.e., $\ckbmbeta=\hbmbeta^{\text{OLS}}$. 
Furthermore, $\bfC\bmbeta \mid \bfY\sim t_{n-p}(\bfC\hbmbeta^{\text{OLS}}, n^{-1}\hat{\sigma}^2\bfC)$ due to
\eqref{eq:marginalposterior}. By a similar reasoning, the conditional distribution
$[\tdbmbeta_{\tdA},\tilde{\bfS}_{\tdI},\tdA\mid \bfY]$ is the same as $\hat{\pi}$ if $\ckbmbeta=\hbmbeta^{\text{OLS}}$ 
and if $f_{\bfU}$ is estimated by the density of $t_{n-p}(\bfzr, n^{-1}\hat{\sigma}^2\bfC)$. 
This motivates our proposal to use $t_{n-p}(\bfzr, n^{-1}\sigma^2\bfC)$ 
as a parametric model for ${\bfU}$ and estimate $\sigma^2$ from data 
to construct $\hat{f}_{\bfU}$. The above discussion also provides a Bayesian justification for sampling from $\hat{\pi}$.

Under this framework, we may define
a point estimator $\hbmbeta^{\supP}=(\hbeta^{\supP}_j)_{1:p}$ by the decision that 
minimizes the posterior expectation of the loss $\ell_{B}(\bmeta,\bmbeta)$,
\begin{equation}\label{eq:defBest}
\hbmbeta^{\supP}
\defi \arg\min_{\bmeta}\int \ell_{B}(\bmeta,\bmbeta) p(\bmbeta \mid \bfY) d\bmbeta,
\end{equation}
provided that the expectation exists. Although $\hbmbeta^{\supP}$ minimizes the posterior expected loss, its Bayes risk is not well-defined due to our use of an improper prior. To avoid any potential confusion, we call $\hbmbeta^{\supP}$ a posterior point estimator instead of a Bayes estimator. Taking subderivative of $\ell_{B}(\bmeta,\bmbeta)$ with respect to $\bmeta$ leads to the following equation to solve for the minimizer $\hbmbeta^{\supP}$:
\begin{equation}\label{eq:KKTBest}
\bfC \hbmbeta^{\supP} + \lambda \bfW \bfS^{\supP} 
= \int  \bfC\bmbeta \cdot p(\bmbeta \mid \bfY) d\bmbeta=\E(\bfC\bmbeta \mid \bfY),
\end{equation}
where $\bfS^{\supP}$ is the subgradient of $\|\bmeta\|_1$ at $\hbmbeta^{\supP}$. Under the noninformative prior,
the posterior mean $\E(\bmbeta \mid \bfY)=\hbmbeta^{\text{OLS}}$. In this case,
$\E(\bfC\bmbeta \mid \bfY)=n^{-1}\bfX^{\trans}\bfY$ and Equation \eqref{eq:KKTBest}
is identical to the KKT condition \eqref{eq:lassograd} for the Lasso-type estimator $\hbmbeta$. 
Therefore, $\hbmbeta$ can be interpreted as the estimator \eqref{eq:defBest} that minimizes the posterior expected loss.

\begin{remark}\label{rm:Bayes}
These results provide a Bayesian interpretation of the Lasso-type estimator $\hbmbeta$ and its sampling distribution. 
Assume a normal error distribution with a given $\sigma^2$ and the noninformative prior.
The posterior distribution of the optimal decision, $[\tdbmbeta\mid \bfY]$, is identical to the sampling distribution of $\hbmbeta$
assuming $\hbmbeta^{\text{OLS}}$ is the true coefficient vector. Therefore, a posterior probability interval for $\tdbmbeta$, the optimal decision, constructed according to $[\tdbmbeta\mid \bfY]$ is the same as the confidence interval 
constructed according to $\hat{\pi}$ with
$\ckbmbeta=\hbmbeta^{\text{OLS}}$. Point estimation about $\bmbeta$ also coincides between the Bayesian
and the penalized least-squares methods ($\hbmbeta^{\supP}=\hbmbeta$). 
Lastly, if we set $\lambda=0$ in the loss \eqref{eq:Bayesloss}, then the optimal decision $\tdbmbeta$ is simply $\bmbeta$.
In this special case, the aforementioned coincidences become the familiar correspondence between
the posterior distribution \eqref{eq:condposterior} and the sampling distribution of $\hbmbeta^{\text{OLS}}$ 
and that between the posterior mean and $\hbmbeta^{\text{OLS}}$.
\end{remark}

It is worth mentioning that, in a loose sense, this Bayesian interpretation also applies when $p>n$.
In this case, the posterior distribution \eqref{eq:condposterior} does not exist, but $[\bfC\bmbeta \mid\sigma^2, \bfY]$
is a well-defined normal distribution in $\row(\bfX)$. From KKT conditions \eqref{eq:Bayesiangrad} and \eqref{eq:KKTBest},
we see that the posterior point estimator $\hbmbeta^{\supP}$ and the posterior distribution $[\tdbmbeta\mid \bfY]$
only depend on $\bfC\bmbeta$. Therefore, they are well-defined and have the same coincidence
with the Lasso-type estimator and its sampling distribution.

\subsection{Bootstrap versus Monte Carlo}\label{sec:limitation}

We have demonstrated that Monte Carlo sampling via estimator augmentation has substantial advantages in approximating tail probabilities and conditional distributions, say $[\hbmbeta_A \mid \calA=A]$, over direct sampling (or  bootstrap). The MH Lasso sampler also showed some improvement in efficiency when compared against direct sampling in the low-dimensional setting. Now we discuss some limitations of estimator augmentation relative to bootstrap.

The joint density of the augmented estimator is derived for a given $\lambda$, and thus does not take into account the randomness in $\lambda$ when it is chosen via a data-dependent way, say via cross-validation. Denote by $\hbmbeta(\bfy,\hat{\lambda}(\bfY))$ the Lasso-type estimator when $\lambda=\hat{\lambda}(\bfY)$, where $\hat{\lambda}(\bfY)$ is estimated from the data $\bfY$. We stress that the density in Theorem~\ref{thm:joint} or Theorem~\ref{thm:jointhigh} does not apply to the sampling distribution of $\hbmbeta(\bfY,\hat{\lambda}(\bfY))$ and it is only valid for $\hbmbeta(\bfY,\lambda)$ with $\lambda$ being fixed during the repeated sampling of $\bfY$. However, the direct sampler (or bootstrap in a similar way) can handle data-dependent $\lambda$ by adding one additional step to determine $\hat{\lambda}(\bfY^{(t)})$ after each draw of $\bfY^{(t)}$ in Routine~\ref{rt:directsample}.

Bootstrap and the direct sampler can be parallelized. The importance sampling algorithm (Routine~\ref{rt:IS}) can easily be parallelized as well, since it uses the direct sampler to generate proposals and calculates importance weights independently for each sample. An MCMC algorithm needs a certain number of burn-in iterations before the Markov chain reaches its stationary distribution. It seems that naively running multiple short chains in parallel may impair the overall efficiency due to the computational waist of multiple burn-in iterations. Initialized with one draw from the direct sampler, a Markov chain simulated by Routine~\ref{rt:dirmcmc}, however, reaches its equilibrium at the first iteration and thus is suitable for parallel computing. Its efficiency relative to direct sampling when both are parallelized can be calculated as follows.

Suppose our goal is to estimate $\E_{\pi}[g(\hbmbeta)]$ and assume that $\Var_{\pi}[g(\hbmbeta)]=1$ without loss of generality. Assume that the time to run one iteration of the direct sampler allows for running $m$ iterations of an MCMC algorithm. Suppose that we have access to $K$ computing nodes and the available computing time from each node allows for the simulation of $(1+N_1)$ samples from the direct sampler, where $N_1$ may be small. Thus, on a single node we can run $N_2=mN_1$ MCMC iterations plus an initial draw from the direct sampler in the same amount of time. In other words, we can run Routine~\ref{rt:dirmcmc} for $1+N_2$ iterations to draw $\bmbeta^{(t)}$ for $t=1,\ldots,1+N_2$. Note that this Markov chain reaches equilibrium from $t=1$. Let $\rho_t=\cor(g(\bmbeta^{(1)}),g(\bmbeta^{(t+1)}))$ and 
\begin{equation*}
\psi(N)=1+2 \sum_{t=1}^{N-1} \left(1-\frac{t}{N}\right)\rho_t
\end{equation*}
for an integer $N\geq 1$. Then we have
\begin{equation*}
\Var\left[\frac{1}{N_2+1} \sum_{t=1}^{N_2+1}g(\bmbeta^{(t)})\right]=\frac{1}{N_2+1}\psi(N_2+1)\defi V_2(N_2+1).
\end{equation*}
Denote by $V_1(N)=1/N$ the variance in estimating $g$ by the mean of an i.i.d. sample of size $N$, and let 
\begin{equation*}
\gamma=\lim_{N\to\infty}\frac{V_1(N)}{V_2(mN)}= \frac{m}{\psi(\infty)}.
\end{equation*}
The efficiency of Routine~\ref{rt:dirmcmc} relative to direct sampling is
\begin{align*}
\frac{V_1(N_1+1)}{V_2(N_2+1)} &=\frac{mN_1+1}{N_1+1}\frac{1}{\psi(N_2+1)} \\
& >\frac{N_1}{N_1+1}\frac{m}{\psi(N_2+1)} \geq \frac{N_1}{N_1+1}\gamma,
\end{align*}
where we have assumed that $\psi(N_2+1)\leq \psi(\infty)$ for the last inequality. This assumption holds if $\psi(N)$ is nondecreasing in $N$. This derivation shows that Routine~\ref{rt:dirmcmc} will be more efficient than the direct sampler on each computing node if $N_1\geq 1/(\gamma-1)$, which can be as small as 1 when $\gamma>2$.
We have observed two decay patterns of the autocorrelation $\rho_t$ of the MLS in the simulation study in Section~\ref{sec:numerical}. For some components of $\hbmbeta$, $\rho_t$ is always positive before it decays to zero, in which case $\psi(N)$ is obviously nondecreasing. For other components, $\rho_t$ first decreases monotonely to zero and then shows small fluctuations around zero. In the second case, we empirically observed that $\psi(N)$ is nondecreasing as well. The efficiency comparison in Table~\ref{tab:MSEjoint}, with $N_1$ and $N_2=mN_1$ both large, suggests that for most functions estimated there, $\gamma\in(2,3)$  for datasets A and B and $\gamma\in (1.2,1.6)$ for the other two datasets. Therefore, as long as we need to run a few iterations of the direct sampler on each node, parallelizing Routine~\ref{rt:dirmcmc} can bring computational gain. Of course, if the number of computing nodes $K$ is so large that only one draw is needed from each node, direct sampling or bootstrap will be a better choice.

\subsection{Concluding remarks}\label{sec:discuss}

Utilizing the density of an augmented estimator, this article develops
MCMC and IS methods to approximate sampling distributions
in $\ell_1$-penalized linear regression. This approach is clearly different from
existing methods based on resampling or asymptotic approximation. The numerical results have already
demonstrated the substantial gain in efficiency and the great flexibility offered by this approach. 
These results are mostly for a proof of principle, 
and there is room for further development of more efficient Monte Carlo algorithms based on
the densities derived in this article. 

In principle, the idea of estimator augmentation can be applied to the use of concave penalties in linear regression
\citep{FrankFriedman93,FanLi01,Friedman08,Zhang10} for studying the sampling distribution.
However, there are at least two additional technical 
difficulties for the high-dimensional setting. First, we need to find conditions for the uniqueness of a concave-penalized
estimator in order to construct a bijection between $\bfU$ and the augmented estimator. Second, the constraint in
\eqref{eq:spacehigh} will become nonlinear in general, even for a fixed $\bfs_A$, when a concave penalty is used,
which means that the sample space is composed of a finite number of manifolds.  
Another future direction is to investigate theoretically and empirically the finite-sample
performance in variable selection by the Lasso sampler which may take into account the uncertainty in parameter estimation in a coherent way.

\section{Proofs}\label{sec:proofs}

Let $n\geq 1$ and $p\geq 2$ throughout this section.

\subsection{Proof of Theorem~\ref{thm:threshold}}

\begin{lemma}\label{lm:truncated}
Let $\bfZ\in\R^p$ be a random vector, $K \in \scrR^p$, and $\bfZ_K$ be the truncation of $\bfZ$ to $K$ such that $P(\bfZ_K\in B)=P(\bfZ\in B \mid \bfZ \in K)$ for $B \in \scrR^p$. If $P(\bfZ \in K)\geq 1-\alpha>0$, then 
\begin{align*}
\sup_{B\in \scrR^p}|P(\bfZ_K\in B)-P(\bfZ \in B)|\leq \alpha.
\end{align*}
\end{lemma}
\begin{proof}
For any $B\in\scrR^p$, $P(\bfZ \in B \cap K) = P(\bfZ_K\in B) P(\bfZ\in K)$ and thus
\begin{align*}
0 \leq P(\bfZ_K\in B)- P(\bfZ \in B \cap K) = P(\bfZ_K\in B) P(\bfZ\in K^c) \leq \alpha.
\end{align*}
On the other hand,
\begin{align*}
0 \leq P(\bfZ\in B)- P(\bfZ \in B \cap K) =P(\bfZ\in B \cap K^c) \leq \alpha.
\end{align*}
Therefore, $|P(\bfZ_K\in B)-P(\bfZ \in B)|\leq \alpha$ for any $B$ and the conclusion follows.
\end{proof}

\begin{lemma}\label{lm:Veq}
Assume that $\ckbmbeta$ satisfies conditions (i) and (ii) in Lemma~\ref{lm:fixknownvar}, and let \eqref{eq:assbeta0j} be satisfied.  
Then $V(\bmdelta; \ckbmbeta, \bfu)=V(\bmdelta; \bmbeta_0, \bfu)$ for any $\bfu \in \R^p$ if $\|\bmdelta\|_{\infty}\leq M_1$.
\end{lemma}
\begin{proof}
The assumptions on $\ckbmbeta$ imply that $\calA(\ckbmbeta)=A_0$ and $\sgn(\ckbeta_j)=\sgn(\beta_{0j})$ for all $j \in A_0$. By the definition of $V$ \eqref{eq:defVn} it then suffices to show that 
\begin{align}\label{eq:signequal}
| \beta_{0j}+ r_n^{-1}\delta_{j}|- | \beta_{0j}|=| \ckbeta_{j}+ r_n^{-1}\delta_{j}|- | \ckbeta_{j}|
\end{align}
for all $j\in A_0$. By the definition of $\eta$ in \eqref{eq:difftdbeta}, for $j\in A_0$
\begin{align*}
\eta |\beta_{0j}| \geq |\ckbeta_{j}-\beta_{0j}| \geq |\beta_{0j}| - |\ckbeta_{j}|
\end{align*}
and therefore
\begin{align*}
|\ckbeta_{j}| \geq (1-\eta) |\beta_{0j}| > M_1/r_n \geq |r_n^{-1}\delta_j|,
\end{align*}
where we have used \eqref{eq:assbeta0j} and that $\|\bmdelta\|_{\infty}\leq M_1$. Consequently, for $j \in A_0$ we have
\begin{align*}
| \ckbeta_{j}+ r_n^{-1}\delta_{j}|- | \ckbeta_{j}|= \sgn(\ckbeta_j) \delta_j/r_n.
\end{align*}
On the other hand, by \eqref{eq:assbeta0j} and $\eta \in[0,1)$, 
$|\beta_{0j}| > M_1/r_n \geq |r_n^{-1}\delta_j|$ for $j\in A_0$ and thus
\begin{align*}
| \beta_{0j}+ r_n^{-1}\delta_{j}|- | \beta_{0j}|= \sgn(\beta_{0j}) \delta_j/r_n.
\end{align*}
Now \eqref{eq:signequal} follows since $\sgn(\ckbeta_j)=\sgn(\beta_{0j})$ for all $j \in A_0$.
\end{proof}

\begin{proof}[Proof of Lemma~\ref{lm:fixknownvar}]
Define $\tbmdelta=\arg\min_{\bmdelta}V(\bmdelta;\bmbeta_0,\sbfU)$, which follows the same distribution as $\hbmdelta$ \eqref{eq:defdelta}. This is because $\sbfU\eqinL \bfU$ by fixing $\hsigma=\sigma$ and $V(\bmdelta;\bmbeta_0,\bfu)$ has a unique minimizer for any $\bfu$ if the columns of $\bfX$ are in general position (Lemma~\ref{lm:prelim}). Consequently,
\begin{align*}
\nu[\hbmdelta \mid \|\hbmdelta\|_{\infty} < M_1] =\nu[\tbmdelta\mid \|\tbmdelta\|_{\infty} < M_1].
\end{align*}
Let $\calK=\{\bmdelta\in\R^p:\|\bmdelta\|_{\infty}< M_1\}$.
According to Lemma~\ref{lm:Veq}, $V(\bmdelta;\ckbmbeta,\sbfU)=V(\bmdelta;\bmbeta_0,\sbfU)$ for all $\bmdelta\in\calK$.
As the unique minimizer of $V(\bmdelta;\ckbmbeta,\sbfU)$ \eqref{eq:defdeltastar}, $\|\sbmdelta \|_{\infty}< M_1$ implies that $\sbmdelta$ is also a local minimizer of $V(\bmdelta;\bmbeta_0,\sbfU)$. Since $V(\bmdelta;\bmbeta_0,\sbfU)$ is convex in $\bmdelta$ and has only a unique minimizer $\tbmdelta$, we must have $\sbmdelta=\tbmdelta$ and $\|\tbmdelta \|_{\infty}< M_1$. Furthermore, using the same argument in the other direction, one can show that $\|\tbmdelta \|_{\infty}< M_1$ implies $\|\sbmdelta \|_{\infty}< M_1$, and thus $\{\|\tbmdelta \|_{\infty}< M_1\}$ is equivalent to $\{\|\sbmdelta \|_{\infty}< M_1\}$. This completes the proof. 
\end{proof}

\begin{proof}[Proof of Theorem~\ref{thm:threshold}]
Let $E_1$ be the event that $\ckbmbeta$ satisfies conditions (i) and (ii) in Lemma~\ref{lm:fixknownvar} and $E_2=\{\|\hbmdelta\|_{\infty} < M_1\}$. We first show that \eqref{eq:supPrdiff} holds on $E_1$. Obviously, \eqref{eq:assbeta0j} holds because of \eqref{eq:betamin}. The argument in the proof of Lemma~\ref{lm:fixknownvar} implies that, on event $E_1$,
\begin{align}
P(\|\sbmdelta\|_{\infty} < M_1\mid \ckbmbeta) &=P(\|\tbmdelta\|_{\infty} < M_1) \nonumber \\
& = P(E_2)\geq 1-\alpha_1, \label{eq:truncatePr}
\end{align}
where the second equality is due to $\tbmdelta\eqinL \hbmdelta$. Let $\bmdelta^*_{\calK}$ and $\hbmdelta_{\calK}$ be the respective truncations of $\sbmdelta$ and $\hbmdelta$ to $\calK$. Lemma~\ref{lm:fixknownvar} implies $\nu[\bmdelta^*_{\calK} \mid \ckbmbeta]=\nu[\hbmdelta_{\calK}]$ on event $E_1$. A direct consequence is that on $E_1$, $P(\bmdelta^*_{\calK}\in B \mid \ckbmbeta)=P(\hbmdelta_{\calK} \in B)$ for any $B\in\scrR^p$ and therefore,
\begin{align*}
\sup_{B\in \scrR^p} |P(\sbmdelta \in B \mid \ckbmbeta) - P(\hbmdelta \in B)| \leq 2\alpha_1
\end{align*}
by Lemma~\ref{lm:truncated} and \eqref{eq:truncatePr}.

Next we find a lower bound for $P(E_1)$. Since $\inf_{A_0}|\beta_{0j}|> 2\tau$ \eqref{eq:betamin}, on event $\calE$, we have $\calA(\ckbmbeta)=A_0$, according to Lemma~\ref{lm:screen}. By construction $\ckbmbeta_{A_0}=\hbmbeta_{A_0}$ \eqref{eq:thresholddef} and consequently
\begin{equation}\label{eq:etaupper}
\eta =  \sup_{j \in A_0} \frac{|\hbeta_{j}-\beta_{0j}|}{| \beta_{0j}|} \leq \frac{\|\hbmbeta-\bmbeta_{0}\|_2}{\inf_{A_0}| \beta_{0j}|} < \frac{1}{2}
\end{equation}
again on $\calE$. Thus, $P(E_1)\geq P(\calE) \geq 1-2p^{1-u}$.
\end{proof}

\subsection{Proof of Theorem~\ref{thm:thresholdsigmaest}}

\begin{lemma}\label{lm:devVmin}
Let $\bmgamma$ be any minimizer of $V(\bmdelta;\ckbmbeta,\bfu)$ for $\ckbmbeta\in\R^p$ and $\bfu \in \R^p$. For any $\bmDelta \in \R^p$, we have
\begin{align}\label{eq:bounddevVmin}
V(\bmgamma+\bmDelta;\ckbmbeta,\bfu) \geq V(\bmgamma;\ckbmbeta,\bfu) + \frac{n}{2r_n^2} \bmDelta^{\trans} \bfC \bmDelta.
\end{align}
\end{lemma}
\begin{proof}
Let $\bfb=r_n^{-1}\bmgamma+\ckbmbeta=(b_j)_{1:p}$ and $\bmDelta=(\Delta_j)_{1:p}$. Direct calculations give
\begin{align*}
& \quad V(\bmgamma+\bmDelta;\ckbmbeta,\bfu)- V(\bmgamma;\ckbmbeta,\bfu)  \\
 & = \frac{n}{2r_n^2} \bmDelta^{\trans} \bfC (\bmDelta + 2 \bmgamma)
- \frac{n}{r_n}\bmDelta^{\trans} \bfu 
+ n\lambda \sum_{j=1}^p w_j (| b_j+ r_n^{-1}\Delta_{j}|-|b_j|). 
\end{align*}
The KKT condition for $\bmgamma$ to minimize $V(\bmdelta;\ckbmbeta,\bfu)$ is
\begin{align}\label{eq:KKTminV}
\bfC(\bfb-\ckbmbeta) + \lambda \bfW \bfs -\bfu =\bfzr,
\end{align}
where $\bfs=(s_j)_{1:p}$ is the subgradient of $\|\bmbeta\|_1$ at $\bfb$. By the definition of a subgradient,\begin{align}\label{eq:subgradineq}
| b_j+ r_n^{-1}\Delta_{j}|-|b_j| \geq s_j r_n^{-1}\Delta_j
\end{align}
for all $j=1,\ldots,p$. Now we have
\begin{align*}
& V(\bmgamma+\bmDelta;\ckbmbeta,\bfu)-V(\bmgamma;\ckbmbeta,\bfu) \\
& \quad \geq \frac{n}{2r_n^2} \bmDelta^{\trans} \bfC \bmDelta + \frac{n}{r_n^2} \bmDelta^{\trans} \bfC \bmgamma
- \frac{n}{r_n}\bmDelta^{\trans} \bfu +\frac{n\lambda}{r_n} \sum_{j=1}^p w_j s_j \Delta_j \\
& \quad = \frac{n}{2r_n^2} \bmDelta^{\trans} \bfC \bmDelta + \frac{n}{r_n} \bmDelta^{\trans} 
\left[\bfC(\bfb-\ckbmbeta) - \bfu + \lambda \bfW \bfs \right] = \frac{n}{2r_n^2} \bmDelta^{\trans} \bfC \bmDelta,
\end{align*}
where we have used \eqref{eq:subgradineq} and \eqref{eq:KKTminV}.
\end{proof}

\begin{lemma}\label{lm:scaleu}
Assume that $|\calA(\ckbmbeta)|\leq q$. Let $c>0$, $\bfu=(u_j)_{1:p}\in\R^p$, and $\bmgamma$ be any minimizer of $V(\bmdelta;\ckbmbeta,\bfu)$. If Assumption RE$(q,3)$ is satisfied and $|u_j|\leq w_j\lambda/2$ for all $j=1,\ldots,p$, then
\begin{align}\label{eq:boundscaleu}
|V(\bmgamma;\ckbmbeta,c\bfu)-V(\bmgamma;\ckbmbeta,\bfu)| 
\leq \frac{8 w_{\max}|1-c|}{w_{\min}\kappa^2(q,3)} n\lambda^2 \sum_{A_0} w_j^2.
\end{align}
\end{lemma}
\begin{proof}
Let $\bfb=r_n^{-1}\bmgamma+\ckbmbeta$. Direct calculations give
\begin{align*}
 |V(\bmgamma;\ckbmbeta,c\bfu)-V(\bmgamma;\ckbmbeta,\bfu)| 
& =  n \left|(1-c) \bfu^{\trans} (\bfb-\ckbmbeta) \right| \\
& \leq  n |1-c| \cdot\|\bfu\|_{\infty} \|\bfb-\ckbmbeta \|_1.
\end{align*}
It is seen from \eqref{eq:KKTminV} that $\bfb$ is a minimizer of the loss \eqref{eq:lassoloss} if $\ckbmbeta$ is the true coefficient vector and if $\bfX^{\trans}\bmeps/n=\bfu$. Inequality \eqref{eq:boundL1} in Lemma~\ref{lm:lassobounds} applied under these assumptions leads to
\begin{align*}
\|\bfb-\ckbmbeta \|_1 \leq \frac{16\lambda}{\kappa^2(q,3)} \sum_{A_0} \frac{w_j^2}{w_{\min}}
\end{align*}
if $|u_j|\leq w_j\lambda/2$ for all $j$. Moreover, $\|\bfu\|_{\infty} \leq w_{\max} \lambda/2$ and
hence \eqref{eq:boundscaleu} follows.
\end{proof}

\begin{proof}[Proof of Lemma~\ref{lm:fixunknownvar}]
To simplify notation, let $\hat{c}=\sigma/\hsigma$ and
\begin{equation*}
h=\frac{8 w_{\max}\zeta}{w_{\min}\kappa^2(q,3)} n\lambda^2 \sum_{A_0} w_j^2.
\end{equation*}
If $|U_j^*| \leq (1-\zeta)w_j\lambda/2$, then $|U_j^*| \leq w_j\lambda/2$ and by \eqref{eq:diffhsigma} $|\hat{c}U_j^*| \leq w_j\lambda/2$. Lemma~\ref{lm:scaleu} with \eqref{eq:diffhsigma} implies
\begin{align*}
|V(\sbmdelta;\ckbmbeta,\hat{c}\sbfU)-V(\sbmdelta;\ckbmbeta,\sbfU)| &\leq h,  \\
|V(\sbmdelta_0;\ckbmbeta,\sbfU)-V(\sbmdelta_0;\ckbmbeta,\hat{c}\sbfU)| &\leq h. 
\end{align*}
Let $\bmDelta=\sbmdelta-\sbmdelta_0$. Now we have
\begin{align*}
V(\sbmdelta;\ckbmbeta,\sbfU) & \geq V(\sbmdelta;\ckbmbeta,\hat{c}\sbfU) - h \\
& \geq V(\sbmdelta_0;\ckbmbeta,\hat{c}\sbfU) + \frac{n}{2r_n^2} \bmDelta^{\trans} \bfC \bmDelta -h \\
& \geq V(\sbmdelta_0;\ckbmbeta,\sbfU)+ \frac{n}{2r_n^2} \bmDelta^{\trans} \bfC \bmDelta -2h,
\end{align*}
where the second inequality is due to Lemma~\ref{lm:devVmin}. Lastly, since $V(\sbmdelta_0;\ckbmbeta,\sbfU)\geq V(\sbmdelta;\ckbmbeta,\sbfU)$ by definition \eqref{eq:defdeltastar}, $\bmDelta^{\trans} \bfC \bmDelta \leq 4 r_n^2 h/n$ which coincides with \eqref{eq:boundquadratic}.
\end{proof}

\begin{proof}[Proof of Theorem~\ref{thm:thresholdsigmaest}]
Recall that $\calE$ is the event $\cap_{j=1}^p\{|U_j| \leq w_j\lambda/2\}$. Since the distribution of $\sbmdelta_0$ does not depend on $\hsigma$ and is identical to the distribution of $\sbmdelta$ when $\hsigma$ is fixed to the true noise level $\sigma$, \eqref{eq:supPrdiffdelta0} follows immediately from \eqref{eq:supPrdiff} which holds on $\calE$.

Let $E_3$ be the event in \eqref{eq:diffhsigma} and $\calE^*$ be the event that $\cap_{j=1}^p\{|U_j^*| \leq (1-\zeta)w_j\lambda/2\}$. 
By Lemma~\ref{lm:screen}, on $\calE$ we have $\calA(\ckbmbeta)=A_0$, $|\calA(\ckbmbeta)|=q_0\leq q$ and by \eqref{eq:etaupper}
\begin{equation}\label{eq:bmbetacheck}
\inf_{A_0} |\ckbeta_j| \geq \frac{1}{2} \inf_{A_0} |\beta_{0j}| > \tau. 
\end{equation}
Therefore, all the assumptions on $\ckbmbeta$ and $\hsigma$ in Lemma~\ref{lm:fixunknownvar} are satisfied on $\calE \cap E_3$, which happens with probability at least $1-(\alpha_2+2p^{1-u})$. Moreover, the conditional probability of \eqref{eq:boundquadratic} given $(\ckbmbeta,\hsigma)$ is at least 
\begin{equation*}
P(\calE^* \mid\hsigma) \geq 1-2p^{1-u},
\end{equation*}
by choosing $\lambda \geq (\hsigma/\sigma)\lambda_0/(1-\zeta)$. For the lower bound of the above probability, see \eqref{eq:condlambda} in Lemma~\ref{lm:lassobounds} with $(1-\zeta)w_j$ in place of $w_j$ and $\hsigma$ in place of $\sigma$. As $\hsigma/\sigma \leq 1+\zeta$ on $E_3$, it suffices to choose $\lambda$ as in \eqref{eq:condlambdasigmaest}. What remains is to show that $\bmDelta=\sbmdelta-\sbmdelta_0$ is $M_2$-sparse on the event $\calE^*$. Then \eqref{eq:L2Deltabound} follows from \eqref{eq:boundquadratic} and the definition of $\phi_{\min}(M_2)$ \eqref{eq:remin}.  Regarding $\ckbmbeta$ and $\hsigma$ as the true parameters, Lemma~\ref{lm:screen} with \eqref{eq:bmbetacheck} implies that $\calA(\ckbmbeta) \subseteq \calA(\sbmbeta)$ on $\calE^*$ and therefore, $|\calA(\sbmdelta)| \leq |\calA(\sbmbeta)|\leq M_2/2$ by \eqref{eq:boundsize}. Since $\calE^*$ with \eqref{eq:diffhsigma} implies $|(\sigma/\hsigma)U_j^*| \leq w_j\lambda/2$ for all $j$, by a similar reasoning we also have $|\calA(\sbmdelta_0)| \leq M_2/2$ and thus $|\calA(\bmDelta)|\leq M_2$ on $\calE^*$.
\end{proof}

\appendix
\renewcommand{\theequation}{A.\arabic{equation}}     % redefine the command that creates the equation no.    
\setcounter{equation}{0}  % reset counter
\renewcommand{\thefigure}{A.\arabic{figure}}  
\setcounter{figure}{0}  % reset counter 
\renewcommand{\thetheorem}{A.\arabic{theorem}}  
\setcounter{theorem}{0}  % reset counter 
\renewcommand{\theremark}{A.\arabic{remark}}
\setcounter{remark}{0}

\section*{Appendix}
Recall that $A^{\dag}=A\setminus \{j\}$ in proposal (P3) and $A^{\dag}=A\cup \{j\}$ in (P4). Let $B=A\cap A^{\dag}$. For both proposals,
\begin{equation}\label{eq:detCAA}
 \frac{\det\bfC_{A^{\dag}A^{\dag}}}{\det\bfC_{AA}}=(\bfC_{jj}-\bfC_{jB}\bfC_{BB}^{-1}\bfC_{Bj})^{|A^\dag|-|A|}
 \defi (r_{\det})^{|A^\dag|-|A|}.
\end{equation} 
Suppose that the matrix $\bfC_{AA}^{-1}$ is given. 

When (P3) is proposed, let $k(j)\in\{1,\ldots,|A|\}$ index the position of $j$ in the set $A$ and $d_{k}$ be the $k^\supth$ diagonal element of $\bfC_{AA}^{-1}$. Then $d_{k(j)}=1/r_{\det}$ and thus the ratio \eqref{eq:detCAA} is immediately obtained. If this proposal is rejected, no further computation is necessary. If it is accepted, $\bfC_{A^{\dag}A^{\dag}}^{-1}$ can be obtained after a reverse sweeping of $(-\bfC_{AA}^{-1})$ on position $k(j)$.
When (P4) is proposed, $r_{\det}=\bfC_{jj}-\bfC_{jA}\bfC_{AA}^{-1}\bfC_{Aj}$ and thus the ratio \eqref{eq:detCAA} can be readily calculated. Again, if the proposal is rejected, no further computation is needed. If it is accepted, add $j$ to the last position in the set $A^{\dag}$ and then sweep the matrix
\begin{equation*}
\left(\begin{array}{cc}
-\bfC_{AA}^{-1} & \bfC_{AA}^{-1}\bfC_{Aj} \\
\bfC_{jA}\bfC_{AA}^{-1} & r_{\det}
\end{array}
\right)
\end{equation*}
on the last position to obtain $-\bfC_{A^{\dag}A^{\dag}}^{-1}$. It is seen that for both proposals, the ratio \eqref{eq:detCAA} can be calculated easily and sweeping on a single position is all we need to update $\bfC_{AA}^{-1}$.

\bibliographystyle{asa}
\bibliography{sparsereferences}

\end{document}

%% file: NewCommands.tex
\newcommand{\bfx}{\mathbf{x}} 
\newcommand{\bfy}{\mathbf{y}}

\newcommand{\bfz}{\mathbf{z}} 
\newcommand{\bfX}{\mathbf{X}}
\newcommand{\bfY}{\bfy}
\newcommand{\bfW}{\mathbf{W}}

\newcommand{\bfV}{\mathbf{V}}
\newcommand{\bfU}{\mathbf{U}}
\newcommand{\sbfU}{\mathbf{U}^*}
\newcommand{\bftdU}{\tilde{\mathbf{U}}}

\newcommand{\bfB}{\mathbf{B}}
\newcommand{\calK}{\mathcal{K}}
\newcommand{\calA}{\mathcal{A}}
\newcommand{\calI}{\mathcal{I}}

\newcommand{\tdA}{\tilde{\mathcal{A}}}
\newcommand{\tdI}{\tilde{\mathcal{I}}}
\newcommand{\tdbfS}{\tilde{\mathbf{S}}}
\newcommand{\bfC}{\mathbf{C}}
\newcommand{\bfD}{\mathbf{D}}
\newcommand{\bfT}{\mathbf{T}}

\newcommand{\bfH}{\mathbf{H}}

\newcommand{\bfM}{\mathbf{M}}
\newcommand{\bfO}{\mathbf{O}}
\newcommand{\bfR}{\mathbf{R}}
\newcommand{\bfS}{\mathbf{S}}
\newcommand{\bfZ}{\mathbf{Z}}
\newcommand{\bfI}{\mathbf{I}}
\newcommand{\bfa}{\mathbf{a}}
\newcommand{\bfb}{\mathbf{b}}
\newcommand{\bfd}{\mathbf{d}}

\newcommand{\bfv}{\mathbf{v}}
\newcommand{\bfs}{\mathbf{s}}
\newcommand{\bfr}{\mathbf{r}}
\newcommand{\tdbfs}{\tilde{\mathbf{s}}}
\newcommand{\bfu}{\mathbf{u}}
\newcommand{\bfzr}{\mathbf{0}}

\newcommand{\calD}{\mathcal{D}}

\newcommand{\scrR}{\mathscr{R}}
\newcommand{\calH}{\mathcal{H}}
\newcommand{\R}{\mathbb{R}}

\newcommand{\bmbeta}{\bm{\beta}}
\newcommand{\hbmbeta}{\hat{\bmbeta}}
\newcommand{\hbeta}{\hat{\beta}}
\newcommand{\tdbmbeta}{\tilde{\bmbeta}}
\newcommand{\ckbmbeta}{\check{\bmbeta}}
\newcommand{\sbmbeta}{{\bmbeta}^*}
\newcommand{\tdbeta}{\tilde{\beta}}
\newcommand{\ckbeta}{\check{\beta}}
\newcommand{\bmalpha}{\bm{\alpha}}
\newcommand{\bmdelta}{\bm{\delta}}
\newcommand{\hbmdelta}{\hat{\bm{\delta}}}
\newcommand{\tbmdelta}{\tilde{\bmdelta}}
\newcommand{\sbmdelta}{\bmdelta^*}
\newcommand{\bmDelta}{\bm{\Delta}}
\newcommand{\bmeps}{\bm{\varepsilon}}
\newcommand{\bmgamma}{\bm{\gamma}}
\newcommand{\bmtheta}{\bm{\theta}}
\newcommand{\bmTheta}{\bm{\Theta}}

\newcommand{\bmmu}{\bm{\mu}}

\newcommand{\bmeta}{\bm{\eta}}
\newcommand{\bmSigma}{\bm{\Sigma}}
\newcommand{\hsigma}{\hat{\sigma}}
\newcommand{\bmLmd}{\bm{\Lambda}}

\newcommand{\E}{\mathbb{E}}
\newcommand{\Var}{\text{Var}}
\newcommand{\SD}{\text{SD}}
\newcommand{\cor}{\text{cor}}
\newcommand{\calE}{\mathcal{E}}
\newcommand{\I}{\bm{1}}
\newcommand{\defi}{\mathop{=}\limits^{\Delta}}  
  
\newcommand{\toinL}{\mathop{\to}\limits^{d}}      % mathop for define
\newcommand{\toinP}{\mathop{\to}\limits^{P}}      % mathop for define
\newcommand{\eqinL}{\mathop{=}\limits^{d}}      % mathop for define
\newcommand{\sgn}{\mbox{sgn}}
\newcommand{\dnorm}{\mathcal{N}}

\newcommand{\diag}{\mbox{diag}}
\newcommand{\rank}{\text{rank}}
\newcommand{\row}{\text{row}}
\newcommand{\nul}{\text{null}}
\newcommand{\supp}{\text{supp}}
\newcommand{\trans}{\mathsf{T}}

\newcommand{\supth}{\text{th}}

\newcommand{\supP}{\text{P}}
\newcommand{\sqn}{\sqrt{n}}